\documentclass[letterpaper,11pt,longbibliography]{amsart}
\pdfoutput=1

\usepackage{color}
\definecolor{darkblue}{rgb}{0.,0.,0.4}
\definecolor{darkred}{rgb}{0.5,0.,0.}
\usepackage[pdftex,colorlinks=true,linkcolor=blue,citecolor=darkred,urlcolor=darkblue]{hyperref}

\usepackage{amsmath,amsthm,amsfonts,amssymb,braket,mathtools}
\usepackage[all]{xy}
\usepackage{fullpage}

\usepackage[numbers]{natbib}

\newtheorem{lemma}{Lemma}[section]
\newtheorem{theorem}[lemma]{Theorem}

\newtheorem{proposition}[lemma]{Proposition}
\newtheorem{proposition-definition}[lemma]{Proposition-Definition}
\theoremstyle{definition}
\newtheorem{definition}[lemma]{Definition}
\newtheorem{remark}[lemma]{Remark}

\newcommand{\BB}{{\mathbb{B}}}
\newcommand{\CC}{{\mathbb{C}}}
\newcommand{\RR}{{\mathbb{R}}}
\newcommand{\ZZ}{{\mathbb{Z}}}

\newcommand{\FF}{{\mathbb{F}}}

\newcommand{\lat}[2]{{{#1}\mathbb{Z}^{#2}}}
\newcommand{\half}{{\frac{1}{2}}}

\DeclareMathOperator{\Tr}{{Tr}}
\DeclareMathOperator*{\im}{{im}}
\DeclareMathOperator*{\diag}{{diag}}

\DeclareMathOperator*{\Mat}{{\mathsf{Mat}}}

\DeclareMathOperator{\Hom}{{Hom}}

\newcommand{\bd}{{\partial}}

\newcommand{\hp}{{\mathsf{U}}}
\newcommand{\ehp}{{\mathsf{EU}}}
\newcommand{\qhp}{{\widetilde{\mathsf{U}}}}
\newcommand{\qehp}{{\widetilde{\mathsf{EU}}}}
\newcommand{\gl}{{\mathsf{GL}}}
\newcommand{\sgl}{{\mathsf{SL}}}
\newcommand{\bass}{{\mathsf{B}}}
\newcommand{\vbass}{{\mathsf{V}}}
\newcommand{\sym}{{\mathsf{S}}}

\newcommand{\qwitt}{{\widetilde{\mathfrak{W}}}}
\newcommand{\switt}{{\mathfrak{W}}}
\newcommand{\umod}{{\mathfrak{U}}}
\newcommand{\qumod}{{\widetilde{\mathfrak{U}}}}
\newcommand{\vtheory}{{\mathfrak{L}}}
\newcommand{\clifqca}{{\mathfrak{C}}}

\newcommand{\cx}{{\mathcal{X}}}
\newcommand{\cz}{{\mathcal{Z}}}
\newcommand{\qcz}{{\widetilde{\mathcal{Z}}}}
\newcommand{\hada}{{\mathcal{H}}}

\newcommand{\dd}{{\mathsf{d}}} 
\newcommand{\Rz}{{R[z,\tfrac 1 z]}}
\newcommand{\Ry}{{R[y,y^{-1}]}}
\newcommand{\Ryz}{{R[y,\tfrac 1 y,z,\tfrac 1 z]}}
\newcommand{\projz}{{\mathbf{z}}}
\newcommand{\projy}{{\mathbf{y}}}
\newcommand{\ev}{{\mathfrak{u}}}
\newcommand{\coe}{{\mathfrak{c}}}
\newcommand{\cg}[1]{{\phi_\#^{({#1})}}}

\newcommand{\longcell}[2][c]{\begin{tabular}[#1]{@{}l@{}}#2\end{tabular}}

\begin{document}

\title[Topological phases of unitary dynamics: Classification in Clifford category]
{Topological phases of unitary dynamics:\\ Classification in Clifford category}
\author{Jeongwan Haah}
\address{Microsoft Quantum, Redmond, Washington, USA}
\address{Station Q, Santa Barbara, California, USA}

\begin{abstract}
	A quantum cellular automaton (QCA) or a causal unitary is by definition 
	an automorphism of a local operator algebra,
	by which local operators are mapped to nearby local operators.
	Quantum circuits of small depth, local Hamiltonian evolutions for short time, 
	and translations (shifts) are examples.
	A Clifford QCA is one that maps any Pauli operator 
	to a finite tensor product of Pauli operators.
	Here, we obtain a complete table of groups~$\mathfrak C(\dd,p)$
	of translation invariant Clifford QCA in any spatial dimension~$\dd \ge 0$
	modulo Clifford quantum circuits and shifts over prime~$p$-dimensional qudits,
	where the circuits and shifts are allowed to obey only coarser translation invariance.
	The group $\mathfrak C(\dd,p)$ is nonzero only for $\dd = 2k+3$ if $p=2$ and $\dd = 4k+3$ if $p$ is odd
	where~$k \ge 0$ is any integer,
	in which case $\mathfrak C(\dd,p) \cong \qwitt(\FF_p)$,
	the classical Witt group of nonsingular quadratic forms over the finite field~$\FF_p$.
	It is well known that $\qwitt(\FF_2) \cong \ZZ/2\ZZ$, $\qwitt(\FF_p) \cong \ZZ/4\ZZ$ if $p = 3 \bmod 4$, and
	$\qwitt(\FF_p)\cong \ZZ/2\ZZ \oplus \ZZ/2\ZZ$ if $p = 1 \bmod 4$.
	The classification is achieved by a dimensional descent,
	which is a reduction of Laurent extension theorems 
	for algebraic $L$-groups of surgery theory in topology.
\end{abstract}

\maketitle

\section{Introduction}

A topological phase in condensed matter physics contexts
is a concept for objects that are insensitive to local changes
and hence lack any local order parameters.
A topological phase of gapped systems in particular,
is an equivalence class of ground states of gapped local Hamiltonians
under an equivalence relation given by continuous families of gapped Hamiltonians
or local quantum circuits of small depth.
The enduring problem in this context is of course the classification of such phases,
which is to identify classes of objects to classify,
find complete invariants,
and construct representatives.

Simplifying the notion of equivalence 
in terms of local quantum circuits or locally generated unitaries,
we ask a similar question to unitary dynamics.
Physically relevant unitary dynamics must preserve locality ---
it should be \emph{causal} that 
a local operator evolves into a local operator with possibly enlarged support.
We consider the set of all unitary automorphisms of a quasi-local algebra,
and regard those of form $\mathcal T \exp( -i \int_0^T \mathrm{d} t H_t )$ ({\it i.e.}, \emph{locally generated}) trivial 
where $H_t$ is a time-dependent local Hamiltonian and $T$ is a constant independent of system size.
Thus, any local quantum circuit of small depth is regarded trivial.
We may call such an equivalence class a topological phase of unitary dynamics.

There appears to be at least two branches of study 
related to our topological phases of unitary dynamics.
One branch comes under the name of \emph{quantum cellular automata} (QCA).
This term is motivated from the viewpoint that 
a cellular automaton describes dynamics of some discrete system
that may model computation~\cite{Watrous1995}.
Quantumly, a natural assumption has been imposed that a quantum cellular automaton
be causal and reversible~\cite{SchumacherWerner2004,GNVW},
leading to a definition that matches ours.
The other branch is concerned with periodically driven systems~\cite{Rudner2013,fermionGNVW2,Harper2019}.
(See also \cite{Keyserlingk2016,Keyserlingk2016a}.)
where edge dynamics realizes a unitary that is not necessarily locally generated.
In fact, every causal unitary dynamics in $\dd$ dimensions, regardless whether it is locally generated,
can be realized at the boundary of a periodic (Floquet) unitary dynamics in $\dd+1$ dimensions%
\footnote{
Since ``time'' is always short, every dimension in this paper is spatial.
} 
such that the unitary in the bulk after one period is identity.%
\footnote{
	This uses Eilenberg's swindle argument, which is often attributed to A.~Y.~Kitaev.
}
Recently, a nontrivial causal unitary (a QCA) was identified
as an induced boundary action of an onsite unitary symmetry
of a symmetry-protected topological (SPT) phase 
that lies beyond group cohomology-based examples~\cite{Fidkowski2020},
obstructing a dimensional descent procedure of symmetry actions~\cite{ElseNayak}.

The classification problem on topological phases of unitary dynamics
is widely open especially in three and higher dimensions.
In one dimension, it is well established
that translation is a nontrivial instance of QCA that cannot be locally 
generated~\cite{GNVW,fermionGNVW2,Ranard2022}.
Translation in one dimension can be viewed as a flow,
which generalizes to higher dimensions~\cite{FreedmanHastings2019QCA}.
Actually, in one and two dimensions, flows combined with locally generated unitaries
exhaust all%
\footnote{
	This assumes strict locality and nonfermionic locally finite dimensional degrees of freedom.
	There is still a nontrivial open problem in two dimensions due to ``tails.''
}
possibilities~\cite{GNVW,Ranard2022,FreedmanHastings2019QCA,clifqca1}.

In three dimensions and higher, new phenomena appear.
In~\cite{nta3} a causal unitary (QCA) was constructed 
which is strongly argued to be not locally generated
even combined with translations.
This nontriviality is backed by the conjecture that
\emph{no} commuting strictly local Hamiltonian on a two-dimensional lattice 
with locally finite dimensional degrees of freedom
can realize a chiral topological order~\cite{Kitaev_2005}.
This QCA disentangles a Walker--Wang state~\cite{Walker_2011}
with an input modular tensor category being the three-fermion theory.
The state is believed to exhibit a nontrivial SPT
under time-reversal symmetry~\cite{BCFV},
where we quoted the statement conjectural 
because it is not proved that the state is topologically trivial, {\it i.e.},
disentangled by a locally generated unitary in the absence of the time-reversal symmetry;
the QCA of~\cite{nta3} disentangles every eigenstate of a Hamiltonian 
whose ground state is a particular Walker--Wang state
and is believed to be not locally generated,
but we do not know any construction or an existence proof of 
a locally generated unitary that disentangles the one particular eigenstate that is the Walker--Wang groundstate.
This subtlety on the disentanglability is observed 
in all examples to date~\cite{nta3,clifqca1,Shirley2022} of 
putative nontrivial QCA (modulo locally generated unitaries and translations)
in three dimensions.
In four dimensions, there is a construction of a putative nontrivial QCA~\cite{Chen2021}
that disentangles the ground state of an invertible state 
that does not require any symmetry~\cite{Chen2021,Fidkowski2021}.

Let us clarify a technical point on locality.
The unitary generated by a generic local Hamiltonian on a lattice for a short time,
maps a single-site operator on a site~$x$ 
to a quasi-local operator~$\sum_{X \ni x} O_X$
where $X$ ranges over all subsets of sites containing~$x$,
the operator~$O_X$ is supported on~$X$,
and the norm of operator~$O_X$ is a fast-decaying function of the diameter of~$X$.
On the other hand, a discrete analog of Hamiltonian time evolution operator, a quantum circuit of small depth,
maps a local operator to a strictly local operator.
We may require either the quasi-locality or the strict locality on causal unitaries
and study their classification against locally generated unitaries or shallow quantum circuits.
These different assumptions can be summarized as a map
\begin{align}
	\{\text{Strictly locality-preserving }&\text{causal unitaries}\} \Big/ \{\text{shallow quantum circuits}\}\nonumber\\
	&\xrightarrow{\qquad\overline{inclusion}\qquad}\label{eq:curiousmap}\\
	\{\text{Automorphisms of quasi-local }&\text{algebra}\} \Big/ \{\text{locally generated short-time evolutions}\}.\nonumber
\end{align}
The answer to the classification problems might depend on the scope of locality sensitively,
though no result, not even a plausible example, is known to this end.
Put verbosely, we do not know if the map~\eqref{eq:curiousmap} is injective 
and we do not know if \eqref{eq:curiousmap} is surjective.
Note that the reductions of the nontriviality (with translations deemed trivial) of the cited 
QCA~\cite{nta3,clifqca1,Chen2021,Shirley2022} in $\dd$ dimensions 
to certain conjectures on Hamiltonians in $\dd -1$ dimensions,
assume the strict locality.

In this paper, we restrict ourselves to the situation with the strict locality.
Hence, the word ``QCA'' rather than ``causal unitary'' seems more appropriate
and we will use QCA.
In addition, we impose two extra conditions.
The first is translation invariance, {\it i.e.}, our QCA must commute with all translations.
So, our lattice is always $\ZZ^\dd$ for some $\dd \ge 0$.
Each site in this lattice is occupied by a finite dimensional degree of freedom~$\CC^{p^q}$
for some prime power~$p^q$.
That $p$ is a prime number is a technically important assumption that makes our results complete.
The matrix algebra~$\Mat(p;\CC)$ on~$\CC^p$
is generated over $\CC$ by two matrices $\sum_{j \in \ZZ/p\ZZ} \ket{j+1}\bra{j}$ 
and $\sum_{j \in \ZZ/p\ZZ} e^{2\pi \mathrm i j/p} \ket j \bra j$.
By {\bf Pauli operators} we mean these two generators, their finite products, 
or their finite tensor products of products,
as they are indeed Pauli matrices if~$p=2$.
The second extra condition is that our QCA must map a Pauli operator to a Pauli operator.
This is a rather special condition that merits a new name, {\bf Clifford QCA}.
Thus, we study translation invariant Clifford QCA.
A trivial QCA to us is similarly restricted.
Namely, we consider local Clifford circuits of small depth and translations (shifts),
both of which are regarded {\bf trivial}.
The classification problem we solve is hence formally different
from the problem we have defined in the beginning.
However, we have a map
\begin{align}
	\{\text{prime Clifford QCA}\}&\Big/\{\text{prime Clifford circuits and shifts}\}\nonumber\\
	&\xrightarrow{\qquad\overline{inclusion}\qquad}\label{eq:clifqcaToQCA}\\
	\{\text{QCA = Strictly locality-preserving}&\text{ causal unitaries}\}\Big/\{\text{shallow quantum circuits and shifts}\}\nonumber
\end{align}
which is likely not surjective.
Indeed, if we believe in Walker's conjecture~\cite{Walker_pc}
that any strictly local commuting Hamiltonian in two dimensions 
whose ground state subspace is locally indistinguishable
can only realize a Levin--Wen model~\cite{Levin_2005},
then the QCA in~\cite{Shirley2022} does not belong to the image of~\eqref{eq:clifqcaToQCA}.
Also, the QCA in~\cite{Chen2021} does not belong to the image of~\eqref{eq:clifqcaToQCA}
assuming a different conjecture that the $3+1$d topological order 
of emergent fermionic charges and fermionic loops cannot be realized 
in any $3$d lattice Hamiltonian~\cite{Chen2021,Fidkowski2021}.
On the other hand, Walker's conjecture 
implies that the map~\eqref{eq:clifqcaToQCA} is injective for $\dd=3$.

We present a complete classification of translation invariant Clifford QCA 
in all dimensions for any prime~$p$,
modulo a stable equivalence relation by 
translation invariant Clifford circuits of finite depth and shifts
where the circuits and shifts are allowed to obey a coarser translation invariance.
Given a prime~$p$,
let $\lat q \dd$ denote the $\dd$-dimensional 
lattice of qudits~$(\CC^{p})^{\otimes q}$.
We say two TI Clifford QCA are equivalent $\alpha \simeq \beta$
if $(\alpha \otimes I_\lat{q'}{\dd})(\beta \otimes I_\lat{q''}{\dd})^{-1}$ 
is a Clifford circuit combined with a shift,
that is overall translation invariant under a subgroup of $\ZZ^\dd$ of finite index.
The set of all equivalence classes is denoted by~$\mathfrak C(\dd,p)$.
For any translation invariant Clifford QCA~$\alpha$ on $\lat q \dd$ 
and a lattice~$\lat{q'} \dd$ we have a stabilization:
if $q\le q'$, 
then $\alpha \mapsto \alpha \otimes I$ 
on $\lat q \dd \sqcup \lat {(q'-q)} \dd = \lat {q'} \dd$.
Modulo Clifford circuits and shifts with respect to the unaltered translation group~$\ZZ^\dd$,
we have the direct limit~$\umod^-(R)$ under the stabilization
where $R = \FF_p[x_1,\tfrac 1 {x_1},\ldots,x_\dd, \tfrac 1 {x_\dd}]$.%
\footnote{
	We have not explained what $\umod^-(R)$ is,
	but it will become clear in \S\ref{sec:unitary}.
	Literally, it is the quotient group of all TI Clifford QCA on the infinite space
	modulo TI Clifford circuits of finite depth and shifts with stabilization.
}
On top of that, if $q' = b^\dd q$ for some integer $b \ge 1$,
then $\lat {q'} \dd$ is a lattice with a supersite $\CC^{p^q} \times \{1,2,\ldots,b\}^\dd$
sitting at every point of $\ZZ^d$.
Then we have a coarse-graining map 
$\cg{b} : \alpha \mapsto \alpha'$ where $\alpha'$ is the same as $\alpha$ 
but with the supersite interpreted as a new site.
This coarse-graning makes a directed system of stable groups~$\umod^-(R)$ via $\cg \cdot$
with respect to the partial order of integers~$b$ given by division.
Let~$\mathfrak C(\dd,p)$ be the direct limit.
Our main result is%
\footnote{
	One may find it not general enough that 
	we take translation subgroups of a special form;
	however, any subgroup of $\ZZ^\dd$ of finite index 
	is given by an integer $\dd$-by-$\dd$-matrix 
	of nonzero determinant, and further coarse-graining 
	using the adjugate matrix shows that our choice of $\cg b$ is sufficient.
}
\begin{theorem}\label{thm:main}
	$\mathfrak C(\dd,p)$ is an abelian group 
	isomorphic to the group listed in Table~\ref{tb:result}.
\end{theorem}

\begin{table}[b]
	\begin{tabular}{c||c|c|c|c|c|c|c}
	\hline\hline
	$\dd$ & $0$ & $1$ & $2$ & $3$ & $4$ & $5$ & $2k+3$ \\
	\hline
	$p=2$  & 0 & 0 & 0 & $\ZZ_2$ & 0 & $\ZZ_2$ & $\ZZ_2$ \\
	\hline\hline
	\end{tabular}

	\vspace{3ex}

	\begin{tabular}{c||c|c|c|c|c|c|c|c|c}
	\hline\hline
	$\dd$ & $0$ & $1$ & $2$ & $3$ & $4$ & $5$ & $6$ & $7$ & $4k+3$ \\
	\hline
	$p=3 \bmod 4$ & 0 & 0 & 0 & $\ZZ_4$ & 0 & 0 & 0 & $\ZZ_4$ & $\ZZ_4$\\
	$p=5 \bmod 4$ & 0 & 0 & 0 & $\ZZ_2 \oplus \ZZ_2$ & 0 & 0 & 0 & $\ZZ_2 \oplus \ZZ_2$ & $\ZZ_2 \oplus \ZZ_2$\\
	\hline
	\hline
	\end{tabular}

	\vspace{2ex}
	\caption{
		Classification of translation invariant (TI) Clifford QCA in $\dd \ge 0$ dimensions
		over prime $p$-dimensional qudits.
		$k \ge 0$ is any integer.
		Modulo Clifford circuits and shifts,
		the set of all TI Clifford QCA is an abelian group~$\mathfrak C(\dd,p)$
		isomorphic to the listed group,
		which is zero unless specified otherwise.
	}
	\label{tb:result}
\end{table}

Each nonzero group in the table 
is the classical Witt group of nonsingular quadratic forms over the finite field~$\FF_p$.
A quadratic form~$\phi$ over~$\FF_p$ is a square matrix 
such that $\phi + \phi^T$ has nonzero determinant,
up to an equivalence relation $\phi \sim \phi + \xi - \xi^T$ for any matrix~$\xi$.
If we extend the relation to the stable congruence, 
{\it i.e.}, $\phi \simeq E^T(\phi \oplus \eta)E$
where $\eta =\begin{pmatrix} 0 & 1 \\ 0 & 0\end{pmatrix}$ is a trivial form
and $E$ is any invertible matrix,
then the set of all stable equivalence classes is an abelian group
with the zero class represented by~$\eta$,
call the (classical) Witt group of~$\FF_p$~\cite{Kniga}.

The appearance of the classical Witt groups 
has a direct interpretation in terms of static Hamiltonian systems,
especially for $\dd = 3$~\cite{clifqca1}.
In~\cite{nta3} we showed that a QCA 
can be thought of as a collection of mutually commuting local operators, called a {\bf separator},
whose common eigenspaces are all one-dimensional given eigenvalues,
with a property that some local operators allow transitions among the eigenstates (locally flippable separator).
For the translation invariant Clifford case, 
every QCA~$\alpha_p$ modulo Clifford circuits and shifts
is fully determined by a local commuting Pauli Hamiltonian~$\alpha_p(-\sum_x Z_x)$
whose ground state is unique~\cite[IV.4 \& IV.10]{nta3}.
Note that $-\sum_x Z_x$ is a gapped Hamiltonian with a trivial, product state ground state.
It is then natural to think of what boundary theory can be introduced given a bulk $\dd$-dimensional Hamiltonian.
In~\cite{nta3,clifqca1}, representative QCA~$\alpha_p$ 
for all the cases in Table~\ref{tb:result} with $\dd = 3$ were reported
with a specific boundary theory, such that the overall system with the boundary
has a commuting parent Hamiltonian.
In every case, the boundary theory is some anyon theory.
With the commutativity imposed on parent Hamiltonians, 
such an anyon theory is only meaningful up to a standalone two-dimensional theory,
which is believed~\cite{Walker_pc} to be always Levin--Wen type~\cite{Levin_2005}.
Then, the boundary anyon theory defines 
an element in the Witt group of the modular tensor category~\cite{Davydov2010,Davydov2011},
which reduces to the classical Witt group 
if the fusion group is the additive group of an $\FF_p$-vector space.
In short, the Witt group captures the boundary topological order.

In fact, if one knew that the Hamiltonian defined by a locally flippable separator in~$\dd=3$
always admits a local commuting gapped Hamiltonian at the boundary,
one could use the classification result for local commuting Pauli Hamiltonians in two dimensions~\cite{Haah2018},
to solve the classification problem of translation invariant Clifford QCA in~$\dd=3$.
The result in this paper
can be thought of as a proof 
that one can always gap out boundaries of three-dimensional bulk governed by a commuting Pauli Hamiltonian.

To prove Table~\ref{tb:result},
we use certain isomorphisms of algebraic $L$-groups over Laurent polynomial rings.
The theory of algebraic $L$-groups sometimes goes under the name of $K$-theory of 
quadratic and hermitian forms~\cite{Wall1970Hermitian,Solovev1989Survey}.
This is developed to study obstruction to surgery in geometric topology,
but has a purely algebraic formulation~\cite{Ranicki1,Ranicki2},
which we follow mostly.
Laurent polynomials in our context have been used 
before~\cite{Haah2013,nta3,Haah2018,RubaYang2022}.
The largest portion (\S\ref{sec:updown}) of the main text below is devoted to review
the material of~\cite{Novikov1,Novikov2,Ranicki1,Ranicki2}.
The purpose is to check all the results of~\cite{Ranicki2},
which is written in a slightly different setting,
and to correct some details;
see~\ref{thm:BassUpFromUnitary} and~\ref{thm:up-down-and-down-up}.

We would like to note a recurring theme, the notion of blending~\cite{GNVW,FHH2019,clifqca1}.
Since a QCA is a homomorphism of a locally generated $*$-algebra,
we think of a QCA as a map from single-site operators.
A general QCA~$\alpha$ is said to blend into another QCA~$\beta$
if there exists a QCA~$\gamma$ such that $\alpha$ agrees with~$\gamma$ in a half space
and with~$\beta$ in the other half, 
except for some thin crossover region in between the two half spaces.
This is an equivalence relation and is applicable beyond Clifford QCA.
Over the course of the proof of our Table~\ref{tb:result},
we will encounter various boundary objects,
which we interpret in terms of obstruction to blending.
Given a QCA in $\dd$ dimensions,
one defines a \emph{boundary algebra} along an artificially introduced spatial 
boundary~\cite{GNVW,FreedmanHastings2019QCA,clifqca1}.
Within the Clifford category, this boundary algebra is succinctly described by an antihermitian form
over a Laurent polynomial ring with $\dd-1$ variables~\cite{clifqca1}.
The antihermitian form can be interpreted as a map from a module to its dual module,
and it makes sense to speak of blending of two such maps.
It turns out that a faithful obstruction for the blending 
is valued in an automorphism group of a trivial quadratic form 
over a Laurent polynomial ring with $\dd-2$ variables.
After this two-step descent,
we proceed purely algebraically and iterate the descent procedure,
leading to Table~\ref{tb:result}.

The rest of the paper is organized as follows.

In~\S\ref{sec:forms} we set up foundation on hermitian and quadratic forms.
The Witt group will be defined over a ring with involution
and the classical Witt group of finite fields will be computed fully.
We will explain fine differences between quadratic and hermitian Witt groups.
This only matters with rings in which $2$ is not invertible,
which is important for us when we consider Clifford QCA over qubits~$\CC^2$.
Over rings where $2$ is invertible, there are two Witt groups depending on a sign.
This will become important in our final classification calculation
and is the origin of the mod~4 periodicity in~$\dd$.

In~\S\ref{sec:unitary} we introduce automorphism groups of trivial quadratic forms~$\eta$ 
and of trivial hermitian forms~$\lambda$.
The group of all translation invariant Clifford QCA on~$\ZZ^\dd$
will be naturally identified with one of these groups~\cite{Haah2013}.
The infinite system~$\ZZ^\dd$ is convenient for us
because, roughly speaking,
any finite quantity will correspond to a constant-sized quantity independent of system size
if we compactify the physical space to have an infinite family of finite systems.
Similarly to the Witt groups, there will be four kinds of automorphism groups.
Following classical literature, we call these \emph{unitary groups}.
Perhaps, this is an unwelcome overloading of terminology 
for the unitary group over complex numbers,
but the complex unitary group is also an automorphism group of a hermitian form
represented by the identity matrix with the involution of complex numbers 
being the complex conjugation.
So, the name \emph{unitary group} is orthodox.
To avoid confusion, we will write \emph{$\eta$-unitary} and \emph{$\lambda$-unitary}
for the automorphism groups over Laurent polynomial rings
and reserve $\CC$-unitary for the complex unitary group.
In our previous papers~\cite{Haah2013,clifqca1} 
we called the $\lambda^-$-unitary group 
a ``symplectic'' group, but we no longer do.
We briefly digress to investigate the relation between Clifford circuits 
and commutators of Clifford QCA.
It turns out that the commutator subgroup is a subgroup of index at most~$2$
of the group of all Clifford circuits.

In~\S\ref{sec:timereversal} we reprove the fact that 
that every QCA over qubits $\CC^2$
may be modified by a Clifford circuit so that it maps 
every real operator to a real operator~\cite{nta3}.
This is needed to explain the peculiarity in Table~\ref{tb:result}
that $\mathfrak C(\dd=1,p=2) = 0$ but $\mathfrak C(\dd,p=2)$ has period~$2$ in~$\dd \ge 3$.

In~\S\ref{sec:updown} we reprove the Laurent extension theorems of 
Novikov and Ranicki~\cite{Novikov1,Novikov2,Ranicki1,Ranicki2},
specialized for finitely generated free modules
with some amendments in~\ref{thm:BassUpFromUnitary} 
and~\ref{thm:up-down-and-down-up}.
We keep our exposition self-contained and every result will be proved,
except for the Quillen--Suslin--Swan theorem~\cite{Suslin1977Stability,Swan1978}.
Experts in surgery theory would recall that there are three $L$-groups,
called ``$U$''-, ``$V$''-, and ``$W$''-theories,%
\footnote{
	The end of~\cite{Ranicki2} 
	contains a correction for the definition of $W$-theories~\cite{Ranicki1}.
}
corresponding to finitely generated
projective modules, stably free modules, 
and based free modules with simple equivalences, respectively.
Over our base ring $R = \FF_p[x_1,\tfrac 1 {x_1},\ldots,x_\dd,\tfrac 1 {x_\dd}]$,
the Quillen--Suslin--Swan~\cite{Suslin1977Stability,Swan1978} theorem states
that every finitely generated projective module is free.
Hence, $U$- and $V$-theories coincide.
Moreover, Suslin's stability theorem~\cite{Suslin1977Stability} says
that any invertible matrix over~$R$ 
is a product of elementary matrices and a unit,
so $V$- and $W$-theories coincide.
This is the reason we take $p$ to be a prime.

In~\S\ref{sec:cg} we complete our classification proof
by breaking the translation invariance down to a coarser translation group.
Without such weaker translation invariance,
all the lower dimensional invariants survive
and the classification table would have been more complicated.
It is not difficult to track all those lower dimensional invariants, 
but we focus on the top dimensional invariant that is arguably the most important.
We interpret the Laurent extension theorems
as blending obstructions.
The constructions of the isomorphisms of~\S\ref{sec:updown},
though not necessarily their proofs,
will be needed to understand the blending obstructions
and to obtain concrete representatives of nontrivial Clifford QCA.

{\bf Acknowledgments.}
I thank Mike Freedman, Matt Hastings, and Zhenghan Wang for useful discussions.

\tableofcontents

{\bf Conventions.}
Throughout, $R$ denotes the Laurent polynomial ring~$\FF_p[x_1,x_1^{-1},\ldots,x_\dd,x_\dd^{-1}] = \FF_p[\ZZ^\dd]$
for some integer~\mbox{$\dd \ge 0$} over a prime field~$\FF_p$.
There is an $\FF_p$-linear involution~$R \ni r \mapsto \bar r \in R$
such that~$x_i \mapsto x_i^{-1}$.
So,~$\overline{x_i} = \tfrac 1 {x_i}$.
For any matrix~$M$, the adjoint~$M^\dag$ is the transpose of~$M$ followed by entrywise involution.
We use $\gl(q;R)$ to denote the set of all $q$-by-$q$ invertible matrices over~$R$,
{\it i.e.}, 
those of determinant~$f x_1^{t_1} \cdots x_\dd^{t_\dd} \in R$ for some nonzero~$f \in \FF_p$ and~$t_j \in \ZZ$.
We use $\gl(R)$ to denote the infinite general linear group,
the direct limit of the inclusions $M \mapsto M \oplus I$.
We define~$\sgl(q;R)$ and~$\sgl(R)$ to be the subgroup of~$\gl(q;R)$ and~$\gl(R)$, respectively,
consisting of all matrices of determinant~$1 \in R$.
For any ring~$T$ and a positive integer~$n$, 
we denote by~$\Mat(n;T)$ the set of all $n$-by-$n$ matrices with entries in a ring~$T$.
We sometimes omit $n$ when the size is inferred from the context.

Given a free $R$-module~$R^m = \bigoplus_{j=1}^m R e_j$ 
we define a dual~$(R^m)^*$, the set of all $R$-linear functionals on~$R^m$,
with a standard basis~$f_1,\ldots,f_m : R^m \to R$
defined by~$f_j(e_j) = 1$ but $f_j(e_i) = 0$ if $j \neq i$.
The ring action on the dual is given by~$(r \cdot f)(\circ) = f( \bar r\, \circ)$.%
\footnote{
	In addition to this ring action,
	there is another ring action which is $(rf)(\circ) = r f(\circ)$.
	If $R$ is noncommutative,
	the two actions are clearly distinguished.
	Suppose we have a right module~$M$;
	since we write a module element as a column matrix, this is appropriate.
	Then, the dual~$M^*$ is naturally a left module.
	With the involution that is an antiautomorphism of~$R$,
	we make the dual into a right module:
	if $f \in M^*$, then for $r \in R$, we define $(f \cdot r)(\circ) = f(\circ\, \bar r)$.
	This is our ring action on the dual.
	Since our $R$ is commutative, there is no distinction between left and right modules,
	but it continues to be true that there are two different actions on a dual.
	We always use the one by the involution.
}
Since the dual~$(R^m)^*$ is a free $R$-module,
we may represent its elements by a column matrix using the standard basis.
When doing so, one needs to be careful since elements of $R^m$ are also represented by column matrices.
If a column matrix $v$ represents an element of~$(R^m)^*$,
then the linear functional by~$v$ is~$R^m \ni w \mapsto v^\dag w \in R$.
We see that the ring action on~$(R^m)^*$ is consistent: 
$(av)^\dag w = \bar a v^\dag w$ for any~$a \in R$.

\begin{remark}\label{rem:matrixIntoDual}
	Below we will frequently consider a $R$-linear map~$\varphi : Q \to Q^*$
	for some finitely generated free $R$-module~$Q$.
	Since~$Q^*$ and~$Q$ are isomorphic (though not canonically) using our standard basis,
	we can represent~$\varphi$ as a matrix~$P$ that acts on column vectors of~$Q$.
	If $v \in Q$, then $Pv$ is a column vector representing a linear functional,
	that is $Q \ni w \mapsto (Pv)^\dag w = v^\dag P^\dag w \in R$.
	On the other hand, we may think of $(\varphi(e_j))(e_k) \in R$
	which is more canonical since it does not require any choice of basis in~$Q^*$.
	We write~$\varphi(e_j,e_k)$ for~$(\varphi(e_j))(e_k)$.
	This gives a matrix $P'$ whose $(j,k)$ entry is~$\varphi(e_j,e_k)$.
	That is, $P'_{jk} = \overline{P_{kj}}$.
	One should not confuse two matrices~$P$ and~$P'$.
	This is especially important when we compose several maps;
	{\it e.g.},~$Q \to Q^* \to Q \to Q^*$.
	It is unprimed~$P$ that follows the usual matrix multiplication rule in the composition.
	When it comes to concrete calculations, it is perfectly fine to 
	``remove all the stars on modules'' and work in the standard basis;
	however, keeping ``stars'' helps preventing nonsense expressions such as~``$\varphi^2$.''
\end{remark}

\begin{table}[t]
	\caption{
	Abbreviated definitions and conventions.
	}
	\begin{tabular}{c|l}
	\hline\hline
	$\overline \cdots $ & $\FF_p$-linear involution of the base ring by which $x_i \mapsto \tfrac 1 {x_i}$.\\
	\hline
	$\lambda^\mp$- \& $\eta^\mp$-unitary  & \longcell{An automorphism of a module~$Q \oplus Q^*$ \\preserving a trivial hermitian or quadratic form.}\\
	\hline
	$\hat\oplus$ & \longcell{Direct sum  (when applied to concrete unitary matrices,\\ upper left blocks are collected in the upper left).}\\
	$\oplus$ & \longcell{Direct sum (when applied to matrices of forms,\\ the summands are placed block diagonally as usual).}\\
	\hline
	$\sim$ & Equivalence of quadratic $\mp$-forms by adding even forms $\mu \pm \mu^\dag$.\\
	$\cong$ & Congruence without stabilization.\\
	$\simeq$ & Witt equivalence; stable congruence including~$\sim$.\\
	\hline
	$\hp^\mp$,$\qhp^\mp$, $\ehp^\mp$, $\qehp^\mp$ & \longcell{
		Unitary groups; tilde for quadratic, $\mathsf E$ for elementary.\\ For $p=2$, ``Pauli $X$ to Pauli $Y$'' is forbidden in the quadratic case.}\\
	\hline
	$\switt^\mp, \qwitt^\mp$ & Witt groups of nonsingular forms under stable congruence.\\
	$\umod^\mp, \qumod^\mp$ & Unitary groups modulo elementary ones; tilde for quadratic.\\
	\hline
	orthogonal~$\perp$ & Always with respect to the (associated) hermitian form.\\
	sublagrangian & \longcell{A self-orthogonal direct summand \\such that a quadratic $\mp$-form restricts to an even hermitian $\pm$-form.}\\
	lagrangian & A sublagrangian whose orthogonal complement is itself.\\
	\hline
	$
	\begin{Bmatrix}A \\ B\end{Bmatrix}
	\begin{pmatrix} \alpha & \beta \\ \gamma & \delta \end{pmatrix}
	\begin{Bmatrix}C \\ D\end{Bmatrix}
	$
	&
	\longcell{
	shorthand for four maps $\alpha : C \to A$, $\beta : D \to A$, $\gamma: C \to B$, \\and $\delta : D \to B$.}
	\\
	\hline\hline
	\end{tabular}
	\vspace{2ex}
	\label{tb:glossary}
\end{table}

Since there are similar but importantly different notions,
we collect their abbreviated definitions in Table~\ref{tb:glossary} for convenient reference;
precise definitions are given in the text.

\section{Quadratic and hermitian forms}\label{sec:forms}

For self-contained exposition,
we collect some standard results on Witt groups of quadratic forms.
See~\cite{Wall1970Hermitian,Ranicki1} for general account.

Let $Q, Q'$ be finitely generated $R$-modules.
A {\bf sesquilinear} function~$\Phi$ whose domain is~$Q' \times Q$ 
is identified with a map~$Q' \to Q^*$.
This is appropriate because if~$v \in Q'$, $w \in Q$, and~$a \in R$
then $\Phi(a v)(w) = (a \cdot \Phi(v))(w) = \Phi(v)(\bar a w) = \bar a \Phi(v)(w)$
in accordance to our ring action on dual modules.
As noted in~\ref{rem:matrixIntoDual}, we denote~$\Phi(v)(w)$ by~$\Phi(v,w)$.
The {\bf adjoint}~$\Phi^\dag$ of a sesquilinear function~$\Phi$ on~$Q' \times Q$ 
is defined by $\Phi^\dag(u,v) = \overline{\Phi(v,u)}$;
the double adjoint is itself.%
\footnote{
	For any finitely generated $R$-module~$L$,
	the isomorphism between~$L$ and~$L^{**}$ is given by
	$L \ni u \mapsto u^{**} = (L^* \ni f \mapsto \overline{f(u)}) \in L^{**}$.
	Note that the involution makes this isomorphism $R$-linear under our ring action convention on duals:
	for $r \in R$, we have 
	$r \cdot u = r u \mapsto \big(
	f \mapsto \overline{f(ru)} = \bar r \overline{f(u)} = \bar r u^{**}(f) 
	= u^{**}(\bar r \cdot f) = (r \cdot u^{**})(f) 
	\big) 
	= r \cdot u^{**}$.
	So, our adjoint is the usual dual map.
}

If $Q$ is finitely generated free with a basis chosen,
$\Phi$ can be represented by a matrix of form values
evaluated at all pairs of basis elements.
The matrix for the adjoint is nothing but 
the transpose followed by the entrywise involution.

A sesquilinear function~$\Delta$ on~$Q \times Q$ 
is a {\bf hermitian $\mp$-form} if
$\overline{\Delta(u,v)} = \mp \Delta(v,u)$ for all~$u,v \in Q$.
A hermitian form is {\bf nonsingular} 
if $Q \ni u \mapsto \Delta(u,\circ) \in Q^*$ is an isomorphism.
If $Q$ is a free module of rank~$n$,
we say that the hermitian form has dimension~$n$.
Conversely, an $n$-by-$n$ matrix~$\Delta$ such that $\Delta^\dag = \mp \Delta$
defines a hermitian $\mp$-form of dimension~$n$.
If $\Delta$ is nonsingular, then $\det \Delta \in R^\times$ in a representation using any basis.
This is true regardless of whether we consider the matrix of a linear map or the matrix of form values;
see~\ref{rem:matrixIntoDual}.
We say $v,w \in Q$ are {\bf orthogonal} and write $v \perp w$ if $\Delta(v,w) = 0$.
A {\bf trivial} hermitian $\mp$-form is a nonsingular form represented by
\begin{align}
	\lambda^\mp = 
	\begin{Bmatrix}
	L^* \\ L
	\end{Bmatrix}
	\begin{pmatrix}
	0 & 1 \\
	\mp 1 & 0
	\end{pmatrix}
	\begin{Bmatrix}
	L \\ L^*
	\end{Bmatrix}
\end{align}
for a finitely generated projective 
(and hence free due to Quillen--Suslin--Swan theorem~\cite{Suslin1977Stability,Swan1978}) 
$R$-module~$L$, where we identified the double dual~$L^{**}$ with~$L$ in the second component of the codomain.
Here, the column matrix with braces denote the direct summands of the domain (the right braced column matrix) 
and codomain (the left),
and the middle matrix collects four $R$-linear maps.
In view of~\ref{rem:matrixIntoDual} we note that this is not a matrix of values.

A {\bf quadratic $\mp$-form} is
a sesquilinear function~$\phi$ on $Q \times Q$
together with a choice of sign 
in the {\bf associated} hermitian $\mp$-form~$\phi \mp \phi^\dag$.
A quadratic form is {\bf nonsingular} if its associated hermitian form is nonsingular.
A quadratic $\mp$-form~$\phi$ is {\bf equivalent} to another~$\xi$
if $\phi = \xi + \theta \pm \theta^\dag$ for some sesquilinear function~$\theta$,
in which case we write $\phi \sim \xi$.
Equivalent quadratic forms have the same associated hermitian form.
A hermitian $\mp$-form is {\bf even} if it is associated with some quadratic form.
We say $v,w \in Q$ are {\bf orthogonal} if $(\phi \mp \phi^\dag)(v,w) = 0$;
this is a weaker condition than~$\phi(v,w) = 0$.
A {\bf trivial} quadratic $\mp$-form is a nonsingular form represented by
\begin{align}
	\eta = 
	\begin{Bmatrix}
	L^* \\ L
	\end{Bmatrix}
	\begin{pmatrix}
	0 & 1 \\
	0 & 0
	\end{pmatrix}
	\begin{Bmatrix}
	L \\ L^*
	\end{Bmatrix}\label{eq:eta}
\end{align}
for a finitely generated projective (hence free) $R$-module~$L$.

If $\half \in R$, 
then the associated hermitian form determines a quadratic form up to equivalence~$\sim$.
The ambiguity in the decomposition~$\Delta = \half \Delta \mp \half \Delta^\dagger$
is precisely in the form~$\theta \pm \theta^\dagger$.
So, one does not have to distinguish the two notions, quadratic and hermitian forms,
and the evenness condition is always satisfied.
However, if $R$ has characteristic~$2$,
then the two inequivalent quadratic forms~$\begin{pmatrix} 1 & 1 \\ 0 & 1 \end{pmatrix}$
and~$\begin{pmatrix} 0 & 1 \\ 0 & 0 \end{pmatrix}$
give the same associated hermitian form.
In addition, a nonsingular hermitian form may not be even:
$\begin{pmatrix} 1 & 1 \\ 1 & 0 \end{pmatrix}$.

In all cases, a quadratic form can be characterized by its values in the following sense.

\begin{lemma}[\cite{Wall1970Hermitian}]\label{thm:quadraticFormByValues}
	Let $R_\pm = \{ r \pm \bar r ~|~ r \in R \}$ be an additive subgroup of~$R$
	and let $R/R_\pm$ be a quotient additive group.
	Any equivalence ($\sim$) class of quadratic $\mp$-forms~$[\phi]$ 
	on a projective $R$-module~$Q$
	defines a function $\beta: Q \ni u \mapsto \phi(u,u) + R_\pm \in R/R_\pm$
	and a hermitian $\mp$-form~$\Phi$ associated with it,
	satisfying for all $u,u' \in Q$ and $a \in R$
	\begin{itemize}
	\item[(i)] $\beta(u+u') = \beta(u) + \beta(u') + (\Phi(u,u') + R_\pm)$,
	\item[(ii)] $\beta(a u) = \bar a \beta(u) a$,  and
	\item[(iii)] $\Phi(u,u) = r \mp \bar r$ for any~$r \in \beta(u)$.
	\end{itemize}
	Conversely, for any function $\beta: Q \to R/R_\pm$ and a hermitian $\mp$-form~$\Phi$
	satisfying all the three conditions,
	there is a unique equivalence ($\sim$) class of quadratic $\mp$-forms on~$Q$
	which gives~$\beta$ and~$\Phi$.
\end{lemma}

Note that if $R = \FF$ is a field with the identity involution,
then $R_- = \{0\}$ and a quadratic ($+$)-form is a classical quadratic form 
over~$\FF$.

\begin{proof}
	The forward direction is obvious, 
	where the choice of a representative~$\phi$ is clearly immaterial.

	If the lemma is proved with $Q$ free,
	the full lemma follows because a projective module is a direct summand of a free module,
	and a form may be defined to assume zero on the complementary module.
	So, we assume $Q$ is free with a basis~$\{e_j\}$
	which we assume is totally ordered ($<$) but is not necessarily finite.

	To show the existence in the converse,
	we define a representative quadratic form~$\xi$ by
	letting $\xi(e_j,e_j) \in \beta(e_j,e_j)$ be arbitrary
	and~$\xi(e_j, e_k) = \Phi(e_j,e_k)$ if~$j < k$ but~$\xi(e_j, e_k) = 0$ if~$j > k$.
	The definition of~$\xi$ is extended to all~$Q$ by sesquilinearity.
	We have to check that a function $\tilde \beta(u) = \xi(u,u) + R_\pm$ and $\Phi$
	satisfy the three conditions.
	By definition, $\tilde \beta(e_j) = \beta(e_j)$ for all~$j$
	and $\Phi(e_j,e_j) = \xi(e_j,e_j) \mp \overline{\xi(e_j,e_j)}$.
	Let~$a = \sum_j a_j e_j$ and~$b = \sum_k b_k e_k$ 
	be arbitrary but finite $R$-linear combinations.
	Condition~(i) reads $\xi(a+b,a+b) - \xi(a,a) - \xi(b,b) = \Phi(a,b) + r_\pm$ for some $r_\pm \in R_\pm$,
	which is true because
	\begin{align}
		&\sum_{j,k} \bar a_j \xi(e_j, e_k) b_k + \bar b_j\xi(e_j,e_k)a_k \nonumber\\
		=&
		\sum_{j} (\bar a_j \xi(e_j,e_j) b_j + \bar b_j \xi(e_j,e_j) a_j ) 
		+ \sum_{j < k} (\bar a_j \Phi(e_j,e_k) b_k + \bar b_j \Phi(e_j,e_k) a_k )\nonumber\\
		=&
		\sum_j ( \bar a_j \Phi(e_j,e_j) b_j \pm \bar a_j \overline{\xi(e_j,e_j)} b_j + \bar b_j \xi(e_j,e_j) a_j )\\
		&\qquad
		+ \sum_{j < k} (\bar a_j \Phi(e_j,e_k) b_k + \bar a_k \Phi(e_k,e_j) b_j) - (\bar a_k \Phi(e_k,e_j) b_j \pm \bar b_j \overline{\Phi(e_k,e_j)} a_k )\nonumber\\
		=&
		\sum_{j,k} \bar a_j \Phi(e_j,e_k) b_k + r_\pm .\nonumber
	\end{align}
	Condition~(ii) is immediate from the definition of~$\tilde \beta$.
	Condition~(iii) is checked as
	\begin{align}
		\xi(a,a) &= \sum_{j,k} \bar a_j \xi(e_j,e_k) a_k = \sum_j \bar a_j \xi(e_j,e_j) a_j + \sum_{j<k} \bar a_j \Phi(e_j,e_k) a_k,\nonumber\\
		\overline{\xi(a,a)} &= \mp \sum_j \bar a_j \overline{\xi(e_j,e_j)} a_j \mp \sum_{j<k} \bar a_k \Phi(e_k,e_j) a_j, \\
		\xi(a,a) \mp \overline{\xi(a,a)} &= \sum_j \bar a_j \Phi(e_j,e_j) a_j + \sum_{j \neq k} \bar a_j \Phi(e_j,e_k) a_k.\nonumber
	\end{align}
	It follows that $\beta(u) = \tilde \beta(u)$ for all $u \in Q$ because,
	due to~(ii), this holds for all the basis elements,
	and
	due to~(i), this holds for any finite $R$-linear combination of basis elements.
	We conclude that $\xi$ is a desired quadratic form.
	Finally, for uniqueness,
	we define $\theta(e_j,e_k) = (\phi - \xi)(e_j,e_k)$ for~$j \neq k$ and
	choose $\theta(e_j,e_j)$ such that $(\theta \pm \theta^\dag)(e_j,e_j) = (\phi - \xi)(e_j,e_j) \in R_\pm$ for each~$j$,
	and extend the definition of $\theta$ by sesquilinearity to all~$Q$.
	It follows that $\phi - \xi = \theta \pm \theta^\dag$ and $[\phi] = [\xi]$.
\end{proof}

We say two forms $\phi,\phi'$, quadratic or hermitian, of possibly different dimensions 
are {\bf stably equivalent} or {\bf Witt equivalent} and write $\phi \simeq \phi'$
if there is an invertible $R$-module morphism~$E$
such that
\begin{align}
	E^\dagger(\phi \oplus \xi)E = \phi' \oplus \xi'
\end{align}
where $\xi,\xi'$ are trivial forms of possibly different dimensions.

\begin{proposition}\label{thm:WittGroupOfForms}
	The set~$\qwitt^\mp(R)$ of all stable equivalence classes
	of nonsingular quadratic $\mp$-forms on finitely generated free $R$-modules
	is an abelian group 
	under the direct sum operation 
	with the zero class represented by a trivial form
	and $[-\xi] = -[\xi]$ for all~$[\xi] \in \qwitt^\mp(R)$.
\end{proposition}

\begin{proof}
	We have to find an inverse for every element~$[\phi] \in \qwitt^\mp(R)$.
	Let $\Delta = \phi \mp \phi^\dag$.
	Let $\psi$ be a quadratic form such that $\psi \sim \Delta^{-1} \phi \Delta^{-1}$.
	It follows that $\psi \mp \psi^\dag = \Delta^{-1} (\phi \mp \phi^\dag) \Delta^{-1} = \Delta^{-1}$.
	Note that $\psi^\dag \phi \psi \sim \psi^\dag (\Delta \psi \Delta) \psi$
	and $\psi \phi \psi^\dag \sim \psi(\Delta \psi \Delta)\psi^\dag$,
	but $(\psi \mp \psi^\dag)\Delta \psi = \psi = \psi\Delta(\psi \mp \psi^\dag)$,
	which implies that $\psi^\dag \Delta \psi = \psi \Delta \psi^\dag$.
	So, $\psi \phi \psi^\dag \sim \psi \Delta \psi \Delta\psi^\dag = \psi \Delta \psi^\dag \Delta \psi 
	= \psi^\dag \Delta \psi \Delta \psi \sim \psi^\dag \phi \psi$.
	Using this equivalence, we see that
	\begin{align}
		\begin{pmatrix}
			I & I \\
			\psi^\dag & \pm \psi
		\end{pmatrix}
		\begin{pmatrix}
			\phi & 0 \\ 0 &-\phi
		\end{pmatrix}
		\begin{pmatrix}
			I & \psi \\
			I & \pm \psi^\dag
		\end{pmatrix}
		=
		\begin{pmatrix}
		0 &&& \phi \Delta^{-1} \\
		\Delta^{-\dag} \phi &&& \psi^\dag \phi \psi - \psi \phi \psi^\dag
		\end{pmatrix}
		\sim
		\begin{pmatrix}
		0 & I \\
		0 & 0
		\end{pmatrix}	.
	\end{align}
	Therefore, $[\phi] + [-\phi] = 0 \in \qwitt^\mp(R)$.
\end{proof}

We call $\qwitt^\mp(R)$ the {\bf quadratic $\mp$-Witt group} of $R$.
(If $R=\FF$ was a field and if the involution was the identity,
then $\qwitt^+(\FF)$ is the classical Witt group of the field~$\FF$.)

For a quadratic form~$\phi$ on a projective $R$-module~$Q$,
a direct summand~$A$ of~$Q$
such that there exists some sesquilinear function~$\theta$ 
with which $(\phi + \theta \pm \theta^\dagger)(a, a) = 0$ for all~$a \in A$,
is called a {\bf sublagrangian}.
If $\phi$ is nonsingular, a {\bf lagrangian} is a sublagrangian such that the orthogonal complement is itself.%
\footnote{
	It is important that a lagrangian is a direct summand.
	Since every exact ``stabilizer module''~\cite{Haah2013}
	has its orthogonal complement equal to itself,
	if we do not demand that a lagrangian be a direct summand,
	every exact stabilizer module would be a lagrangian.
	The phenomenology of exact stabilizer modules 
	is much richer than that of direct summand lagrangians.
	In a sense, the classification of Clifford QCA in this paper 
	is the classification of direct summand lagrangians over~$R$.
}
We similarly define sublagrangians and lagrangians with respect to nonsingular hermitian forms
by the same condition where the even orthogonality is replaced by just orthogonality.
Note that, unlike in the classical theory of quadratic forms over fields
where a quadratic form is defined to be a certain function on a vector space,
our quadratic form restricted to a sublagrangian is not identically zero;
it is merely required to become zero on some equivalent ($\sim$) quadratic form.

\begin{lemma}[Thm.\,1.1 of \cite{Ranicki1}]\label{thm:lagrangian-trivialForm}
	For any finitely generated projective module~$Q$ over~$R$ with a nonsingular quadratic $\mp$-form~$\phi$
	(resp. a nonsingular hermitian even $\mp$-form~$\Phi$),
	if there is a lagrangian~$L \subseteq Q$, then there exists a direct complement~$M \cong L$, 
	as $R$-modules,
	which is a lagrangian itself, such that
	\begin{align}
		\phi \sim 
		\begin{Bmatrix} M^* \\ L^* \end{Bmatrix}
		\begin{pmatrix}
		0 & \varphi \\ 0 & 0
		\end{pmatrix}
		\begin{Bmatrix} M \\ L \end{Bmatrix}
	\end{align}
	for some isomorphism~$\varphi$.
	In particular, $\phi$ (resp. $\Phi$) is trivial if and only if there is a lagrangian.
\end{lemma}

\begin{proof}
	Since $L$ is a direct summand there is a direct complement~$N$ such that $Q = N \oplus L$.
	Then,
	\begin{align}
		\phi =	
		\begin{Bmatrix}
		N^* \\ L^*
		\end{Bmatrix}
		\begin{pmatrix}
		a & b \\ c & d \pm d^\dag
		\end{pmatrix}
		\begin{Bmatrix}
		N \\ L
		\end{Bmatrix},
		\quad
		\phi \mp \phi^\dag =
		\begin{Bmatrix}
		N^* \\ L^*
		\end{Bmatrix}
		\begin{pmatrix}
		a \mp a^\dag & b \mp c^\dag \\ c \mp b^\dag & 0
		\end{pmatrix}
		\begin{Bmatrix}
		N \\ L
		\end{Bmatrix}
	\end{align}
	Since $\phi \mp \phi^\dag$ is an isomorphism, $c \mp b^\dag$ is surjective.
	Since $L^\perp = L \oplus \ker(c\mp b^\dag) = L$, $c \mp b^\dag$ is injective.
	So, $L \cong L^*\xleftarrow{\quad f = c \mp b^\dag \quad} N$ is an $R$-module isomorphism.
	Now,
	\begin{align}
		\begin{Bmatrix}
		N^* \\ L^*
		\end{Bmatrix}
		\begin{pmatrix}
		1 & -a f^{-1} \\ 0 & 1
		\end{pmatrix}
		\begin{pmatrix}
		a & b \\ c & d \pm d^\dag
		\end{pmatrix}
		\begin{pmatrix}
		1 & 0 \\ -f^{-\dag} a^\dag & 1
		\end{pmatrix}
		\begin{Bmatrix}
		N \\ L
		\end{Bmatrix}
	\end{align}
	has $(1,1)$-entry $a - b f^{-\dag} a^\dag - a f^{-1} c + a f^{-1} (d \pm d^\dag) f^{-\dag} a^\dag$,
	which is (using $c = f \pm b^\dag$) $\pm$-even.
	Hence, $M = \{ (v, -f^{-\dag} a^\dag v) \in N \oplus L = Q ~|~ v \in N \}$ is the desired lagrangian.

	The even hermitian case is completely analogous.
\end{proof}

We record some general facts:

\begin{lemma}\label{thm:ComplementaryLagrangiansArePaired}
	Let $\lambda : Z \to Z^*$ be a nonsingular sesquilinear form on a module~$Z$.
	Suppose $Z = X + Y$ for some submodules~$X$ and~$Y$.
	If $\lambda(X)(X) = 0$ and $\lambda(Y)(Y) = 0$,
	then $X \cap Y = 0$ so $Z = X \oplus Y$ 
	and both $X \xrightarrow{~\lambda~} Z^* \to Y^*$ and $Y \xrightarrow{~\lambda~} Z^* \to X^*$ 
	are isomorphisms,
	where $Z^* \to Y^*$ and $Z^* \to X^*$ 
	are the canonical projections with respect to the direct sum $Z^* = X^* \oplus Y^*$.
\end{lemma}

This means that two self-orthogonal submodules with respect to a hermitian form,
jointly generating the full module,
are lagrangians with respect to the hermitian form.
Under the conditions of this lemma,
we will say that $Y$ and $Z$ are a {\bf nonsingular pair} with respect to~$\lambda$.

\begin{proof}
	The assumption is that $X^\perp \supseteq X$, $Y^\perp \supseteq Y$, and $(X+Y)^\perp = 0$.
	A general identity $(X + Y)^\perp = X^\perp \cap Y^\perp$ 
	implies that~$X \cap Y \subseteq X^\perp \cap Y^\perp = 0$.
	Now,
	\begin{align}
		\lambda = \begin{Bmatrix}
		X^* \\ Y^*
		\end{Bmatrix}
		\begin{pmatrix}
		0 & \beta \\ \gamma & 0
		\end{pmatrix}
		\begin{Bmatrix}
		X \\ Y
		\end{Bmatrix}
	\end{align}
	must be an isomorphism. This forces both~$\beta$ and~$\gamma$ be isomorphisms.
\end{proof}

\begin{lemma}\label{thm:FormOnQuotientModule}
	Let $S$ be a sublagrangian of a projective module~$Q$
	with a quadratic $\mp$-form~$\phi$.
	If $S \perp Q$, then $u+S \mapsto \phi(u,u)$
	is a well-defined quadratic $\mp$-form on~$Q/S$.

	If, instead, $Q$ is equipped with a hermitian $\mp$-form~$\Delta$,
	and~$S \perp Q$,
	then~$(u+S,v+S) \mapsto u^\dag \Delta v$ is a well-defined hermitian $\mp$-form on~$Q/S$.
\end{lemma}

\begin{proof}
	The hermitian case is obvious.
	For the quadratic case, we check the three conditions of~\ref{thm:quadraticFormByValues}.
	If $u \in Q$ and $s \in S$,
	then $(u+s)+S \mapsto \phi(u,u) + \phi(s,u) + \phi(u, s) + \phi(s, s)$.
	Since $S$ is a sublagrangian, the value~$\phi(s, s)$ is $\pm$-even.
	Also, since $\phi(s, u) = \pm \phi^\dag(s, u)$, 
	we see that $\phi(s, u) + \phi(u, s) = \pm \phi^\dag(s, u) + \phi(u, s)$ 
	is $\pm$-even.
	Hence, a function $q : Q/S \to R/R_\pm$ is well defined with $R_\pm$ defined in~\ref{thm:quadraticFormByValues}.
	It is routine to check the three conditions for~$q$ and~$\Delta = \phi \mp \phi^\dag$.
	Hence, the quadratic form is well defined.
\end{proof}

\begin{lemma}\label{thm:LagrangianDirectSummandFactorsOutFromQuadraticSpace}
	Let $\phi$ be a nonsingular quadratic $\mp$-form on~$R^m$.
	If there is a nonzero sublagrangian,
	then $\phi$ is stably equivalent to a nonsingular quadratic form of a smaller dimension.
\end{lemma}

In view of~\ref{thm:FormOnQuotientModule}, the new form of a smaller dimension
is the induced one on the quotient module.
Note that the orthogonal complement~$S^\perp$ for any sublagrangian~$S$
is a direct summand of~$R^m = Q$
because if $S \oplus S' = Q$, we have $Q^* = S^* \oplus (S')^*$ 
so $S^\perp = \Delta^{-1}(S')^*$ and $Q = S^\perp \oplus \Delta^{-1} S^*$.

\begin{proof}
	It suffices to assume that we have a sublagrangian~$L = R\ell$ 
	generated by one element~$\ell \in R^m$ and that $\phi(\ell, \ell) = 0$.
	By the Quillen--Suslin--Swan theorem~\cite{Suslin1977Stability,Swan1978}, 
	there is a basis of $R^m$ extending~$\{\ell\}$.
	In this basis, $\phi$ has~$0$ on the top left entry,
	and using the freedom to add $\theta\pm \theta^\dagger$ to $\phi$,
	we may assume that $\phi$ is upper triangular.
	Since $\phi$ is nonsingular, the top row must 
	(be unimodular and hence by the QSS theorem)
	be a row of some invertible matrix of dimension~$m-1$,
	which means $\phi$ can be congruently transformed
	to have entry~$1$ at the second entry in the top row
	and the first column of $\phi$ still zero.
	The second diagonal element is then eliminated by further congruence transformation.
	We have singled out a block of $\eta_1$.
\end{proof}

The use of the Quillen--Suslin--Swan theorem is not too important; 
without them, we could contend ourselves 
with nonsingular forms on projective modules. 
See~\cite[Cor.1.2]{Ranicki1}.
However, the QSS theorem is anyway
an important ingredient for our main classification result of this paper.

\begin{proposition}\label{thm:ZeroDcalculation}
	$\qwitt^+(\FF_2) = \qwitt^-(\FF_2) \cong \ZZ/2\ZZ$.
	Let $p$ be an odd prime power.
	$\qwitt^-(\FF_p) = 0$.
	$\qwitt^+(\FF_p) \cong \ZZ/2\ZZ \oplus \ZZ/2\ZZ$ if $p = 1 \bmod 4$.
	$\qwitt^+(\FF_p) \cong \ZZ/4\ZZ$ if $p = 3 \bmod 4$.
\end{proposition}

This is a widely known fact~\cite{Kniga},
but we choose to include full calculation.

\begin{proof}
	By~\ref{thm:LagrangianDirectSummandFactorsOutFromQuadraticSpace}
	it suffices to consider only forms that do not have any nonzero sublagrangians,
	{\it i.e.}, {\bf totally anisotropic} quadratic forms.
	Recall that every symplectic form (nonsingular alternating bilinear) 
	over any field is congruent to $\lambda^-$~\cite{Lang};
	one can prove it by elementary row (and column) operations
	and induction in dimension.
	Since the associated bilinear form determines a quadratic form if $\half \in \FF_p$,
	we conclude that $\qwitt^-(\FF_p) = 0$ if $p$ is odd.
	It remains to calculate~$\qwitt^+$.

	($\FF_2$)
	The associated bilinear form has zero diagonal,
	and is symplectic. Therefore, the quadratic form~$\phi$ 
	may differ from $\eta$ only on the diagonal.
	Any zero diagonal element means a sublagrangian,
	so we may assume that $\phi = I_{2q} + \eta_{q}$ 
	where $I_{2q}$ is the identity matrix.
	But this is $q$-fold direct sum of $\phi_1 = I_2 + \eta_1$,
	and by~\ref{thm:WittGroupOfForms} this is stably equivalent
	to $\phi_1$ if $q$ is odd and to $\eta$ if $q$ is even.
	On the other hand, $\phi_1$ is totally anisotropic 
	as seen by enumerating the values of all four elements in $\FF_2^2$,
	which is basically the Arf invariant.
	Therefore, $\qwitt(\FF_2) \cong \ZZ/2\ZZ$ generated by $\phi_1$.

	($\FF_p$ with $p$ odd)
	Since $\half \in \FF_p$, we consider nonsingular symmetric bilinear forms~$\Delta$.
	Congruent transformations by elementary row operations
	bring $\Delta$ to a diagonal matrix, that does not contain any zero entry on the diagonal.
	Congruences by diagonal matrices show that we only have to consider
	the diagonal elements up to squares~$\{ x^2 \in \FF_p : x \in \FF_p^\times\}$.
	The multiplicative group $\FF_p^\times$ is abelian,
	and if $n \le p-1$ is the smallest positive integer 
	such that $x^n = 1$ for all $x \in \FF_p$,
	then the equation $x^n - 1 = 0$ has at most~$n$ roots
	and is satisfied by all of $p-1$ elements of~$\FF_p^\times$,
	and hence $n = p -1$, and $\FF_p^\times$ is cyclic with a generator, say,~$g$.
	Therefore, the quotient group of~$\FF_p^\times$ modulo squares 
	has only two elements, represented by $1$ and $g$.
	So, we may assume that $\Delta$ is diagonal consisting of $1$ and $g$,
	each repeated some number of times.

	For $p=1 \bmod 4$, $g^{(p-1)/4}$ squares to an multiplicative order~$2$ element,
	which has to be~$-1$.
	If any diagonal element~$a$ of $\Delta$ is repeated,
	then $\Delta$ is congruent to one where both $a$ and $-a$ appear.
	They give a sublagrangian.
	Hence, there are only three possibilities for totally anisotropic~$\Delta$:
	$\diag(1), \diag(g), \diag(1,g)$.
	The dimension one cases are certainly nontrivial in the Witt group
	because its dimension is odd.
	The dimension two case is totally anisotropic 
	since if $x^2 + g y^2 = 0$ had a nonzero
	solution $(x,y)$ then $g \cong -g = (x/y)^2$ would be a square.
	Therefore, $\qwitt^+(\FF_p) \cong \ZZ/2\ZZ \oplus \ZZ/2\ZZ$.

	For $p=3\bmod 4$, we know $-1 = g^k \in \FF_p$ for some $k$, 
	but $k$ cannot be even,
	and $-1$ is not a square.
	Therefore, $g = -h^2$ for some $h \in \FF_p^\times$.
	On the other hand, a two-variable equation $x^2 + y^2 + 1 = 0$ 
	has a solution over~$\FF_p$, {\it i.e.},
	the sum of two squares may assume~$-1$.
	To see this, let $S$ be the set of all nonzero squares.
	If $S + S = \{x^2 + y^2:x,y \in \FF_p^\times\} \subseteq S$, 
	then $n x^2 \in S$ for any~$n$,
	but $p x^2 = 0 \notin S$. So, $S+S \not\subseteq S$
	and $S+S$ must contain at least one and hence all nonsquare nonzero elements,
	in particular~$-1$.
	Let $x,y$ be such a pair that $x^2 + y^2 = -1$.
	Then, a congruent transformation 
	of~$\diag(a,a)$ by~$\begin{pmatrix}x & y \\ -y & x\end{pmatrix}$
	gives~$\diag(-a,-a)$ for any $a \in \FF_p^\times$,
	and $\diag(1,1,1) \cong \diag(1,-1,-1) \cong \diag(1,-1,-h^2) 
	\cong \lambda^+_1 \oplus \diag(g)$
	by~\ref{thm:WittGroupOfForms}.
	Hence, the $+$-Witt group of $\FF_p$ is cyclic, generated by~$\diag(1)$,
	of order that divides~$4$.
	Clearly, this generator is totally anisotropic.
	The two-fold direct sum~$\diag(1,1)$ is also totally anisotropic,
	since, otherwise, the equation~$x^2 + y^2 = 0$ would have a nonzero solution,
	which means $-1 = (x/y)^2$, which we know is false.
	Therefore, $\qwitt^+(\FF_p) \cong \ZZ/4\ZZ$.
\end{proof}

\begin{lemma}\label{thm:hermitianMatrices}
	The set $\switt^\mp(R)$ of all stable equivalence classes of nonsingular even hermitian $\mp$-forms 
	on finitely generated free $R$-modules
	is an abelian group.
	There is a surjective group homomorphism
	\begin{align}
		\sym : \qwitt^\mp(R) \ni [\xi] \mapsto [\xi \mp \xi^\dag] \in \switt^\mp(R),
	\end{align}
	whose kernel is contained in the image of the embedding $\qwitt^\mp(\FF) \hookrightarrow \qwitt^\mp(R)$
	defined by~$[\phi] \mapsto [\phi]$.
\end{lemma}

\begin{proof}
	Every element~$[\xi \mp \xi^\dag]$ has an inverse~$[\xi' \mp (\xi')^\dag]$
	where $[\xi'] + [\xi] = 0$ in~$\qwitt(R)$.
	The morphism~$\sym$ is clearly well defined and surjective.

	Fix some basis for each free module,
	and identify a form with its matrix of values evaluated on all pairs of the basis elements.
	The map~$\qwitt(\FF) \to \qwitt(R)$ 
	is injective because if~$E^\dagger(\xi \oplus \eta) E = \eta_{q'}$ with~$\xi \in \qwitt^\mp(\FF)$ and~$E \in \gl(q';R)$,
	then the morphism~$\ev: R \ni x_i \mapsto 1 \in \FF$ gives a stable equivalence of~$\xi$ to~$\eta$ over~$\FF$.
	Also,~$\ev$ gives a right inverse $\qwitt(R) \to \qwitt(\FF)$ of the embedding.

	Regarding the kernel,
	consider $E^\dag (\xi \mp \xi^\dag) E = \eta \mp \eta^\dag$ and put $\xi' = E^\dag \xi E$.
	Then, $\xi' - \eta = \pm(\xi' - \eta)^\dag = \sigma$.
	Let $\delta$ collect the diagonal ``constant'' terms that are invariant under the involution,
	so~$\delta \in \Mat(\FF)$.
	Let $S$ collect the upper diagonal part of $\sigma - \delta$ 
	together with an exactly half of the diagonal terms, each of which is not invariant under the involution,
	so that $S \pm S^\dag = \sigma - \delta$.
	Then, $\sigma = S \pm S^\dag + \delta$ where $\delta \mp \delta^\dag = 0$.
	So, $\xi'\sim \eta + \delta$.
	The latter $\eta + \delta$ is a nonsingular quadratic $\mp$-form over~$\FF$.
\end{proof}

Note that the statement of~\ref{thm:hermitianMatrices} is valid regardless whether $\FF$ is a field;
we only have used the condition that $\FF$ is elementwise fixed under the involution in the proof of~\ref{thm:hermitianMatrices}.

\begin{proposition}\label{thm:sesquilinearWittGroup}
	If $\half \in R$, then $\ker \sym = 0$ and $\switt^\mp(R) \cong \qwitt^\mp(R)$ as abelian groups.
	If~$2=0 \in R$, then the following sequence is split exact.
	\begin{align}
		0 \to \qwitt(\FF) \hookrightarrow \qwitt(R) \xrightarrow{\sym} \switt(R) \to 0
	\end{align}
\end{proposition}
\begin{proof}
	Suppose $\half \in R$.
	If $\xi$ and $\phi$ are two quadratic $\mp$-forms such that $\xi \mp \xi^\dag = \phi\mp \phi^\dag$,
	then $\xi - \phi = \pm(\xi - \phi)^\dag = \half(\xi - \phi) \pm \half(\xi - \phi)^\dag$,
	and $\xi \sim \phi$ as quadratic forms.

	When $2=0$, every nonsingular even hermitian form over~$\FF$ is symplectic and hence trivial,
	implying that $\qwitt(\FF) \to \qwitt(R) \to \switt(R)$ is a zero map.
	We have shown in~\ref{thm:hermitianMatrices} the exactness.
\end{proof}

\begin{proposition}\label{thm:WittGroupExponent}
	If $n$ is an exponent of $\qwitt^+(\FF)$ (i.e., $n\cdot \qwitt^+(\FF) = 0$),
	then $n \cdot \qwitt^\mp(R) = 0$.
\end{proposition}

This is an analogue of~\ref{thm:UnitaryGroupExponent} below,
and has appeared in~\cite[III.10]{clifqca1}.

\begin{proof}
	If $2=0 \in \FF$, the claim is obvious since~$n$ is even.
	Otherwise, by~\ref{thm:sesquilinearWittGroup} we consider $\switt^\mp(R)$.
	The assumption is that $[I_n] = 0 \in \switt^+(\FF)$,
	{\it i.e.}, $I_n \cong \lambda^+_n$.
	Let $[\Delta] \in \switt^\mp(R)$ where $\Delta$ has dimension~$k$.
	Then,
	$I_n \otimes \Delta \cong \lambda^+_n \otimes \Delta \cong \begin{pmatrix} 0 & \Delta \\ \mp \Delta^\dagger & 0 \end{pmatrix} \cong \lambda^\mp_{nk}.$
\end{proof}

\section{Unitary groups and Clifford circuits}\label{sec:unitary}

\begin{definition}
	Define~$\hp^\mp(Q;R) = \hp^\mp(q;R)$ to be the automorphism group of~$Q \oplus Q^* = R^q \oplus R^{q*}$ such that
	\begin{align}
		U^\dagger \lambda_q^\mp U = \lambda_q^\mp. \label{eq:DefLambdaUnitary}
	\end{align}
	The group~$\hp^\mp(q;R)$ embeds into $\hp^\mp(q+1;R)$ 
	in an obvious manner:
	\begin{align}
		\begin{matrix}
		U = \begin{pmatrix} a & b \\ c & d\end{pmatrix} \in \hp(q;R),\\
		V = \begin{pmatrix}a' & b' \\ c' & d' \end{pmatrix} \in \hp(q';R),
		\end{matrix}\qquad
		U \hat\oplus V = 
		\begin{pmatrix}
		a & 0 & b & 0 \\
		0 & a'& 0 & b'\\
		c & 0 & d & 0 \\
		0 & c'& 0 & d'
		\end{pmatrix} \in \hp(q+q';R).
	\end{align}
	The embedding is $U \mapsto U\hat\oplus I$.
	We take the direct limit~$\hp^\mp(R) = \bigcup_{q} \hp^\mp(q;R)$ to define
	the (infinite) {\bf $\lambda^\mp$-unitary group}.
	Let $\qhp^\mp(q;R)$ be the automorphism group of~$R^q \oplus R^{q*}$
	consisting of $U$ such that
	\begin{align}
		U^\dagger \eta_q U = \eta_q + \theta \pm \theta^\dagger
	\end{align}
	for some $\theta : R^q \oplus R^{q*} \to R^{q*} \oplus R^{q}$ that may depend on~$U$.
	That is, $U^\dag \eta U \sim \eta$.
	Similarly, we define the (infinite) {\bf $\eta^\mp$-unitary group}
	$\qhp^\mp(R) = \bigcup_{q} \qhp^\mp(q;R)$.
\end{definition}

\begin{proposition}\label{thm:qhphp}
	For any $q \ge 1$, $\qhp^\mp(q;R) \le \hp^\mp(q;R)$. 
	The equality holds if~$\tfrac 1 2 \in R$.
\end{proposition}
\begin{proof}
	The inclusion is obvious since $(\eta_q + \theta\pm \theta^\dagger) \mp (\eta_q + \theta\pm \theta^\dagger)^\dagger = \lambda_q^\mp$.

	If $U \in \hp^\mp(q)$, then $U^\dagger \eta U - \eta = \mp \eta^\dagger \pm U^\dagger \eta^\dagger U = \pm (U^\dagger \eta U - \eta)^\dagger$.
	If $\half \in R$, 
	then any $\alpha$ such that $\alpha^\dagger = \pm \alpha$ can be written as
	$\alpha = \half \alpha \pm \half \alpha^\dagger$.
\end{proof}

\begin{definition}
	Define~$\ehp^\mp(q;R)$ to be the subgroup of $\hp^\mp(q;R)$ generated by
	\begin{align}
		\cz(\theta) = \begin{pmatrix}
			1 & 0 \\ \theta & 1
		\end{pmatrix} &\in \hp^\mp(q)  \text{ where } \theta^\dagger = \pm \theta : R^q \to R^{q*}, \nonumber\\
		\cx(\alpha) = \begin{pmatrix}
			\alpha & 0 \\ 0 & \alpha^{-\dagger}
		\end{pmatrix} &\in \hp^\mp(q)  \text{ where } \alpha \in \gl(q;R), \\
		\hada_\mp = \begin{pmatrix}
			0 & \hat 1^{-\dag} \\ \mp \hat 1 & 0
		\end{pmatrix} &\in \hp^\mp(1) \subseteq \hp^\mp(q) \nonumber
	\end{align}
	where $\hat 1 : R^{q} \ni u \mapsto (v \mapsto u^\dag v) \in R^{q*}$ 
	has the identity matrix representation under the standard basis of~$R^{q*}$.
	We also define the (infinite) {\bf elementary $\lambda^\mp$-unitary group} $\ehp^\mp$
	by the direct limit of the inclusions $\ehp^\mp(q;R) \hookrightarrow \ehp^\mp(q+1;R)$.
	Define~$\qehp^\mp(q;R)$ to be the subgroup of $\qhp^\mp(q;R)$ generated by $\cx(\alpha)$ 
	with~$\alpha \in \gl(q;R)$, $\hada_\mp$,
	and 
	\begin{align}
		\qcz(\theta) = \cz(\theta \pm \theta^\dagger)
	\end{align}
	with~$\theta: R^q \to R^{q*}$.
	Similarly, $\qehp^\mp(R) = \bigcup_q \qehp^\mp(q;R)$ is the (infinite)
	{\bf elementary $\eta^\mp$-unitary group}.
\end{definition}
Though we defined the unitary group as a subgroup of the automorphism group of $R^q \oplus R^{q*}$,
it is perfectly fine to think of a unitary as a matrix that satisfies the defining equations.

Let us very briefly review the transcription of Pauli groups to $R$-modules~\cite{Haah2013}.
The group of all finitely supported Pauli operators on~$\ZZ^\dd$ with $q$ qudits $\CC^p$ per site is nonabelian, 
but its commutator subgroup consists of $p$-th root of unity.
Since the single-site Pauli group abelianizes to~$\FF_p^{2q}$,
the abelianization of the Pauli group on the infinite lattice 
is the additive group of~$R^{2q}$.
The translation group action gives an $R$-module structure to~$R^{2q}$
that matches the usual ring action on the free $R$-module~$R^{2q}$.
It is our convention that upper half block of a Laurent polynomial column matrix, 
({\it i.e.}, $R^q \oplus 0 \subseteq R^q \oplus R^q$), 
corresponds to the ``$X$''-part, operator factors $X = \sum_{j \in \FF_p} \ket{j+1}\bra{j}$ and its powers,
and the lower half block $0 \oplus R^q$ to the ``$Z$''-part, 
operator factors $Z = \sum_{j \in \FF_p} e^{2\pi\mathrm{i} j / p} \ket j \bra j$.
For example, 
\begin{align}
	\FF_p[x,\tfrac 1 x, y, \tfrac 1 y, z, \tfrac 1 z]^2 \ni
	\begin{pmatrix}
		1+ 2 xy \\ 1 + z^{-1}
	\end{pmatrix} 
	\quad\Longleftrightarrow\quad
	e^{2\pi\mathrm{i} \cdots / p}~ X_{(1,1,0)}^2 X_{(0,0,0)} Z_{(0,0,0)} Z_{(0,0,-1)}.
\end{align}
Any two Pauli operators~$P,Q$ commute up to a $p$-th root of unity, $PQ = \omega^{\lambda}QP$,
where $\lambda$ is valued in~$\ZZ/p\ZZ = \FF_p$.
It is not difficult to see that $\lambda$ as a function of $P,Q$
is a bilinear function of $[P], [Q] \in R^{2q}$.
Moreover, $QP = \omega^{-\lambda}PQ$ implies that $\lambda$ is alternating.
Since the Pauli group has a trivial center, $\lambda$ must be nonsingular
as a map from $R^{2q}$ to the space of all linear functionals of finite support.
Considering all translations of~$P$ and~$Q$,
we promote $\lambda$ to a nonsingular hermitian ($-$)-form, which is~$\lambda_q^-$.
(In our previous paper~\cite{clifqca1}, we called it antihermitian.)

If a $\CC$-algebra automorphism~$\beta$ maps a Pauli operator to a Pauli operator,
we obtain a map~$U$ from a Laurent polynomial column matrix to another.
If $\beta$ is translation invariant,
then the induced map~$U$ must be $R$-linear so we have an $R$-module automorphism on~$R^{2q}$.
The size of a Pauli operator under~$\beta$ can only be finite,
and since the translation invariant map~$\beta$ is determined by finitely many images,
the locality blowup is uniformly bounded across the whole lattice~$\ZZ^\dd$.
Therefore, $\beta$ is a Clifford QCA.
Since any $\CC$-algebra automorphism must preserve commutation relation,
the induced $R$-automorphism~$U$ must preserve~$\lambda$.
We hence obtain a map from Clifford QCA to $\lambda$-unitaries.
Introducing auxiliary qudits in the lattice on which a QCA acts by identity
amounts to our stabilization.

With $p=2$,
it is well known that the controlled-Not, the Hadamard, and the phase gate
generate the full Clifford group on any finite set of qubits.
Here, the phase gate conjugates Pauli~$X$ to~$Y$ but commutes with~$Z$.
In the infinite system, we consider layers of commuting gates
whose arrangement is translation invariant.
Then, it is routine to verify that a layer of 
nonoverlapping controlled-Not gates, one gate per unit cell in the lattice,
induces a elementary row operation $E$ on the upper half block,
and $E^{-\dag}$ on the lower half block, which is in the form of $\cx(E)$
where $E$ has exactly one off-diagonal monomial.
Similarly, a layer of the Hadamard gates, one gate per unit cell, induces $\hada_-$,
and a layer of possibly overlapping controlled-$Z$ gates, one per unit cell, induces $\cz(\theta)$ 
where $\theta = \theta^\dag$ contains exactly two monomials,
so $\cz(\theta) = \qcz(\phi)$ for a matrix~$\phi$ with exactly one nonzero monomial.
A layer of the phase gate, one gate per unit cell,
induces $\cz(1)$.
For $p > 2$, there are analogous gates that induce elementary $\lambda^-$-unitaries of the described form.
Note that, for example,
$\cx(x_1^n x_3^{-1}) \in \ehp^\mp(1;R)$ 
represents a shift (translation) along positive $x_1$-direction by distance~$n$ 
and along negative $x_3$-direction by distance~$1$.

It is a priori not clear whether these elementary $\lambda^-$-unitaries generate all elementary $\lambda^-$-unitaries 
we have define above.
Indeed, it is a nontrivial application of Suslin's stability theorem~\cite{Suslin1977Stability}
that layers of controlled-Not gates induce all possible $\cx(E)$ with~$E \in \gl(R)$.%
\footnote{
	Without Suslin's stability, we would be stuck with $L$-groups of ``simple'' equivalences.
}
On the other hand, it can be proved easily that controlled-$Z$ induces all possible $\qcz(\phi)$~\cite[IV.10]{nta3}.
By definition, the group of all elementary $\lambda^-$-unitaries
is a composition of finitely many layers of gates.
Since our lattice~$\ZZ^\dd$ is infinite,
this composition is always shallow or 
constant depth if we consider any infinite collection of finite systems with periodic boundary conditions.
All together, our (traditional) definition of elementary $\lambda^-$-unitaries captures precisely
the group of shallow translation invariant Clifford circuits.

Next, we examine relative sizes of unitary group, elementary unitary group, and the commutator subgroup.

\begin{proposition}
	$\qehp^\mp(q) = \ehp^\mp(q)$ if $\half \in R$.
\end{proposition}

\begin{proof}
	Obvious; see the proof of~\ref{thm:qhphp}.
\end{proof}

\begin{proposition}\label{thm:n0elem}
	If $R = \FF$ is a field on which the involution acts by identity,
	then $\ehp^\mp(q) = \hp^\mp(q)$ 
	and $\qehp^\mp(q) = \qhp^\mp(q)$ for all~$q$.
\end{proposition}

\begin{proof}
	We use induction in~$q$.

	(Base case $q=1$.)
	If $\FF$ has characteristic different from~$2$,
	then \eqref{eq:DefLambdaUnitary} implies that
	any element of $\hp^-$ is a matrix with determinant~$1$.
	The elementary matrices are precisely row operations,
	by which any element of $\hp^-$ can be brought to the identity.
	Similarly, \eqref{eq:DefLambdaUnitary} implies that
	any element of $\hp^+(1)$ is either $\cx(\alpha)$ or $\cx(\alpha)\hada_+$.
	If $\FF$ is characteristic~$2$,
	$\qhp(1)$ consists of matrices of determinant~$1$
	where each row has exactly one nonzero entry, 
	and any such matrix has form~$\cx(\alpha)$ or $\cx(\alpha)\hada$.
	$\hp(1)$ consists of matrices of determinant~$1$,
	and the elementary matrices supply all possible row operations.

	(Induction step.)
	If $U$ is a unitary matrix,
	the first column of~$U$ is nonzero.
	We apply various elementary unitary on the left of~$U$
	to make it into the unit column vector.
	Indeed, $\hada_\mp$ on some pair of rows 
	can bring a nonzero element to the upper half block.
	$\cx(\alpha)$ can turn the upper half block to the unit column vector.
	Any remaining nonzero entry in the first column can be eliminated by~$\hada_\mp$
	and $\cx$, and possibly~$\cz$ for~$\hp^-$.
	Once the first column is put into the unit vector,
	the first row can also be transformed in a similar manner 
	by acting on the right of~$U$, turning it into the unit row vector.
	Then, \eqref{eq:DefLambdaUnitary} dictates that $q+1$-th row and $q+1$-th column
	are the unit vectors, completing the induction step.
\end{proof}

\begin{proposition}[\cite{Haah2013}]\label{thm:MinusUnitaryIn1D}
	Over $R = \FF[z,\tfrac 1 z]$ where $\FF$ is any field,
	the $\lambda^-$-unitary group $\hp^-(q;R)$ is elementary, i.e., $\hp^-(q;R) = \ehp^-(q;R)$.
\end{proposition}
\begin{proof}
	The base ring is a Euclidean domain with a degree function~$\deg : R \to \ZZ_{\ge 0}$
	being the maximum $z$-exponent minus the minimum $z$-exponent;
	the elements of degree zero are precisely monomials.

	Let $U \in \hp^-(q;R)$.
	The first column of~$U$ by using~$\cx$ on the left, implementing Euclid's algorithm
	which is nothing but some elementary row operation,
	can be brought to one where the upper half has a unique nonzero entry.
	Then, by letting $\alpha^{-\dag} \in \gl(q-1;R)$ to act on all but the first entry of
	the lower half block,
	the first column of~$U$ is further simplified to have at most three nonzero entries,
	the first component, say $a = a_0 z^i + \cdots + a_1 z^j$, in the upper half, 
	and the first two, say $b = b_0 z^k +\cdots + b_1 z^\ell$ and $c$, in the lower half.
	The equation $U^\dag \lambda^- U = \lambda^-$ implies that $a^\dag b = b^\dag a$
	and in particular
	$a_0 b_1 = a_1 b_0$ and $k - i = j - \ell$.
	Without loss of generality, assume that $\deg a \le \deg b$ so that $j - i \le \ell - k$.
	If $\deg a < \deg b$, we define 
	$d = \frac{b_0}{a_0} z^{k-i} + \frac{b_1}{a_1} z^{\ell-j} = \frac{b_0}{a_0} z^{k-i} + \frac{b_0}{a_0} z^{i-k} = d^\dag$.
	If $\deg a = \deg b$, then $\ell = j$ and $i = k$, and we define
	$d = \frac{b_0}{a_0} = \frac{b_1}{a_1}$.
	Now, $\cz(-d) \in \ehp^-(1;R)$ turns%
	\footnote{
		If $\deg a = \deg b$, then $\cz(-d) \in \hp(q;R) \setminus \qhp(1;R)$.
		Hence, the proof does not work for~$\qhp^-$ against~$\qehp^-$.
	}
	$\begin{pmatrix} a \\ b \end{pmatrix}$
	into $\begin{pmatrix} a \\ b - da \end{pmatrix}$ where $\deg (b-da) < \deg b$.
	Since a degree cannot decrease forever, 
	by~$\hada$ and~$\cz$ we can bring the first column into one with~$b = 0$ or~$a=0$.
	In either case, we will be left with a first column where only two components are nonzero,
	one of which will be eliminated by~$\cx$.
	But the first column of~$U$ is unimodular ({\em i.e.}, the entries generate the unit ideal),
	so the eventual sole nonzero component must be a unit.

	Now, acting by~$\ehp$ on the right of~$U$,
	the first row can be brought to a row vector with a sole nonzero component.
	The equation~$U^\dag \lambda^- U = \lambda^-$ implies that
	an identity unitary is singled out.
	By induction in~$q$,
	we express $U$ by elementary unitaries.
\end{proof}

\begin{remark}\label{rem:oneDimNontrivialQuadraticUnitary}
	The proof of~\ref{thm:MinusUnitaryIn1D} does not go through for the $\eta^-$-unitary group over a ring without~$\half$,
	and it shouldn't.
	Indeed, the following $\eta^-$-unitary~$U_{p=2,\dd=1}$ is a generator of $\qhp(R)/\qehp(R)$ with $R = \FF_2[z,\tfrac 1 z]$,
	which will be shown to be isomorphic to~$\ZZ/2\ZZ$.
	See~\ref{thm:oddBassExact} together with~\ref{thm:ZeroDcalculation} and~\ref{thm:n0elem}.
	\begin{align}
		U_{p=2,\dd=1}
		=
		\begin{pmatrix}
			z + \tfrac 1 z &&& z + 1 + \tfrac 1 z \\
			z + 1+ \tfrac 1 z &&& z + \tfrac 1 z
		\end{pmatrix}
	\end{align}
	Observe that this $\eta$-unitary belongs to~$\ehp(R)$.
	In view of~\S\ref{sec:timereversal} below,
	this means that the Hamiltonian over qubits obtained 
	by applying this Clifford QCA to the trivial Hamitonian~$-\sum_j Z_j$
	where $Z_j$ is the single-qubit Pauli~$Z$ at site~$j$,
	is complex-conjugation invariant (in the basis where Pauli~$Z$ is diagonal).
	Indeed, this Hamiltonian is $-\sum_j Y_{j-1}X_j Y_{j+1}$ whose ground state is the so-called cluster state,
	a representative of the time-reversal symmetry protected topological phase in~$\dd=1$~\cite{Pollmann2009,MPS}.
\end{remark}

\begin{lemma}\label{lem:HHbyZX}
	$\hada_\mp \hat\oplus \hada_\mp$ belongs to the subgroup of $\qehp^\mp(2)$ generated by 
	all~$\qcz$ and $\qcz^\dagger$ and all~$\cx$.
\end{lemma}
\begin{proof}
	Let us use the standard basis $\hat 1(1) \in R^*$ with respect to which $\hat 1$ is represented by the identity matrix.
	\begin{align}
	\hada_\mp \hat\oplus \hada_\mp
	=
	\left(
	\begin{array}{cc|cc}
	1 & 0 & 0 & \pm 1 \\
	0 & 1 & 1 & 0 \\
	\hline
	0 & 0 & 1 & 0 \\
	0 & 0 & 0 & 1 \\
	\end{array}
	\right)\left(
	\begin{array}{cc|cc}
	1 & 0 & 0 & 0 \\
	0 & 1 & 0 & 0 \\
	\hline
	0 & -1 & 1 & 0 \\
	\mp 1 & 0 & 0 & 1 \\
	\end{array}
	\right)\left(
	\begin{array}{cc|cc}
	1 & 0 & 0 & \pm 1 \\
	0 & 1 & 1 & 0 \\
	\hline
	0 & 0 & 1 & 0 \\
	0 & 0 & 0 & 1 \\
	\end{array}
	\right)\left(
	\begin{array}{cc|cc}
	0 & 1 & 0 & 0 \\
	\pm 1 & 0 & 0 & 0 \\
	\hline
	0 & 0 & 0 & 1 \\
	0 & 0 & \pm 1 & 0 \\
	\end{array}
	\right)
	\end{align}
\end{proof}

\begin{lemma}\label{thm:CZ}
	For $U,U' \in \hp^\mp(q)$, 
	if the left half blocks of $U$ and $U'$ are the same,
	then there exists $E \in \ehp^\mp(q)$ such that $U E = U'$.
	If furthermore $U,U' \in \qhp^\mp(q)$, then $E$ may be chosen from~$\qehp^\mp(q)$.
\end{lemma}
\begin{proof}
	It suffices to consider $U' = I$ because we may multiply everything by~$(U')^{-1}$.
	By definition of $\hp^\mp(q)$, we must have $U = \begin{pmatrix}
	I & Y \\ 0 & I
	\end{pmatrix}$ where $Y^\dagger \mp Y = 0$.
	If $U \in \qhp^\mp(q)$, then $Y = \theta \pm \theta^\dagger$ for some~$\theta$.
\end{proof}

\begin{proposition}\label{thm:formation-unitary}
	Suppose $L \subset Q \oplus Q^*$ is a lagrangian 
	of a trivial quadratic $\mp$-form~$\eta$ on~$Q \oplus Q^*$ 
	where $Q$ is a finitely generated projective $R$-module.
	Then, there exists an $\eta^\mp$-unitary~$U$ on~$Q \oplus Q^*$
	such that $U(Q \oplus 0) = L$ and such $U$ is unique up to~$\qehp^\mp(Q)$.
	Similarly, if $L$ is a lagrangian 
	with respect to the trivial hermitian form~$\lambda^\mp$ on $Q \oplus Q^*$,
	then there exists a $\lambda^\mp$-unitary~$U$ such that $U(Q \oplus 0) = L$,
	unique up to~$\ehp^\mp(Q)$.
\end{proposition}

This underlies the notion of ``formations''~\cite{WallSurgeryBook,Ranicki1}.
The case of $\lambda^-$-unitary has appeared in~\cite{nta3}.

\begin{proof}
	By~\ref{thm:lagrangian-trivialForm}, 
	there exists a complementary lagrangian~$L^*$
	such that the form~$\eta$ takes the form of~\eqref{eq:eta}.
	A basis change~$U : Q \oplus Q^* \to L \oplus L^*$ is then a desired $\eta^\mp$-unitary.
	If $V$ is any other such unitary,
	the unitary~$V^{-1} U$ maps $Q\oplus 0$ onto itself, giving an automorphism~$\alpha : Q \to Q$.
	Then, $W = \cx(\alpha^{-1})V^{-1} U$ restricts to the identity on~$Q \oplus 0$,
	and the unitarity forces $W$ restricted on~$0 \oplus Q^*$ to be the identity as well.
	By~\ref{thm:CZ}, we conclude that $W$ is elementary.
	The $\lambda^\mp$-unitary case is completely analogous.
\end{proof}

\begin{lemma}\label{thm:TRCisInvU}
	Let $U = \begin{pmatrix} A & B \\ C & D \end{pmatrix} \in \hp^\mp(q)$ and
	define $U' = \begin{pmatrix} A & -B \\ -C & D \end{pmatrix} \in \hp^\mp(q)$.
	Then, $U \hat\oplus U' \in \ehp^\mp(2q)$ and $U \hat\oplus U^{-1} \hat\oplus I \in \ehp^\mp(3q)$.
	If $U \in \qhp^\mp(q)$, then $U' \in \qhp^\mp(2q)$ and $U \hat\oplus U' \in \qehp^\mp(2q)$ 
	and $U \hat\oplus U^{-1} \hat\oplus I \in \qehp^\mp(3q)$.
\end{lemma}

The following calculation is virtually the same as the proof of~\cite[II.2]{clifqca1};
however, we repeat it here to handle $\eta$-unitary groups and $\mp$-signs.
A similar calculation was performed by Wall~\cite[Lemma~6.2]{WallSurgeryBook}.

\begin{proof}
	It is straightforward to check that $U' \in \qhp^\mp(q)$ if $U \in \qhp^\mp(q)$.

	To show the rest of the claim,
	we will use transformations of the form~$\hp^\mp(q) \ni \bullet \mapsto \cx(\alpha) \bullet \cx(\beta) \in \hp^\mp(q)$
	where $\alpha,\beta \in \sgl(2q;R)$
	and $\bullet \mapsto \cz(\theta)\bullet$  
	with $\theta = \pm \theta \in \Mat(2q;R)$ or $\bullet \mapsto \qcz(\theta')\bullet$,
	until the left half block is put in a trivial form.
	We work in a basis where $\hat 1$ is the identity matrix.
	\begin{align}
	&\begin{pmatrix}
	A & 0 \\
	0 & A \\
	-C & 0 \\
	0 & C
	\end{pmatrix}
	\xrightarrow{\alpha = \begin{pmatrix}I & 0 \\ I & I \end{pmatrix} = \beta^{-1}}
	\begin{pmatrix}
	A & 0 \\
	0 & A \\
	0 & -C \\
	-C & C
	\end{pmatrix}
	\xrightarrow{I \hat\oplus (-\hada^{\hat\oplus q})}
	\begin{pmatrix}
	A & 0 \\
	C & -C \\
	0 & -C \\
	0 & \pm A
	\end{pmatrix}
	\xrightarrow{\alpha = U^{-1}, \beta = I}
	\begin{pmatrix}
	I & K'\\
	0 & K\\
	0 & 0 \\
	0 & \pm I
	\end{pmatrix} 
	\nonumber\\&
	\xrightarrow{\alpha = I,\beta = \begin{pmatrix}I & -K' \\ 0  & I \end{pmatrix}}
	\begin{pmatrix}
	I & 0\\
	0 & K\\
	0 & 0 \\
	0 & \pm I
	\end{pmatrix} 
	\xrightarrow{I_{2q} \hat\oplus(\pm \hada^{\hat\oplus q})}
	\begin{pmatrix}
	I & 0\\
	0 & I\\
	0 & 0 \\
	0 & -K
	\end{pmatrix} 
	\xrightarrow{\cz(0 \oplus K)\text{ or } \qcz(0 \oplus \theta)}
	\begin{pmatrix}
	I & 0\\
	0 & I\\
	0 & 0 \\
	0 & 0
	\end{pmatrix} 
	\end{align}
	Here, in the third arrow, the lower block is multiplied on the left 
	by $(U^{-1})^{-\dagger} = U^\dagger$ and the result is due to~\eqref{eq:DefLambdaUnitary}.
	Because all matrices here are the left half block of some matrix of~$\hp^\mp(q)$,
	the matrix~$K$ must satisfy $K = \pm K^\dagger$.
	If $U \in \qhp^\mp(q)$, then $K = \theta \pm \theta^\dagger$ 
	for some $\theta\in \Mat(q;R)$.
	Since the left half block is trivialized,
	we conclude the first part of the proof by~\ref{thm:CZ}.

	Using the first part of the claim,
	we have $\ehp^\mp(3q) \ni I \hat\oplus U' \hat\oplus U = ((U^{-1} \hat\oplus I) E) \hat\oplus U$.
	Hence, $U^{-1} \hat\oplus I \hat\oplus U \in \ehp(3q)$.
	Similarly, if $U \in \qhp^\mp(q)$, then $U^{-1} \hat\oplus I \hat\oplus U \in \qehp(3q)$
\end{proof}

\begin{proposition}[\cite{WallSurgeryBook}]\label{thm:circuit-contains-commutators}
	The elementary unitary groups contain
	the commutator subgroups of the unitary groups:
	$[\hp^\mp,\hp^\mp] \le \ehp^\mp$ and~$[\qhp^\mp,\qhp^\mp] \le \qehp^\mp$.
	In particular, both $\umod^\mp = \hp^\mp/\ehp^\mp$ and $\qumod^\mp = \qhp^\mp / \qehp^\mp$ are abelian groups.
\end{proposition}

\begin{proof}
	Suppose $A,B \in \hp^\mp(q)$.
	Then, from~\ref{thm:TRCisInvU},
	we see 
	\begin{align}
	(ABA^{-1} B^{-1}) \hat\oplus I = (AB \hat\oplus (AB)^{-1})(A^{-1} \hat\oplus A)(B^{-1} \hat\oplus B) \in \ehp^\mp(3q).
	\end{align}
	Therefore, $[\hp^\mp,\hp^\mp] \le \ehp^\mp$.
	The proof for $\eta$-unitary groups is completely analogous.
\end{proof}

\begin{proposition}\label{thm:UnitaryGroupExponent}
	If $n$ is an exponent of~$\qwitt^+(\FF)$ (i.e., $n \cdot \qwitt^+(\FF) = 0$),
	then $n \cdot \umod^\mp = 0$ and $n \cdot \qumod^\mp = 0$.
\end{proposition}

This is analogous to~\ref{thm:WittGroupExponent}.

\begin{proof}
	If $2=0 \in \FF$, then the smallest exponent of $\qwitt^+(\FF)$ is~$2$ 
	and the claim is obvious by~\ref{thm:TRCisInvU}.
	So, assume $\half \in \FF$.
	Since any trivial form has dimension even,
	an exponent must also be even, say $n = 2m$.
	By the Witt cancellation theorem, 
	the assumption is that $E^T I_m E = -I_m$ for some~$E \in \gl(m;\FF)$.
	The $m$-fold direct sum of a unitary~$U$ can be written as $U \otimes I_m$.
	Then, $\cx(I \otimes E^{-1}) (U \otimes I_m) \cx(I \otimes E)$ is equal to~$(U\otimes I_m)'$ 
	in the notation of~\ref{thm:TRCisInvU}.
	Therefore, $U \otimes I_{2m}$ is elementary by~\ref{thm:TRCisInvU}.
\end{proof}

We wonder how large the group of elementary unitaries is above the commutator subgroup.
Until the end of this section, we just think of any unitary as a matrix acting on~$R^{2q}$ using the standard basis.
Recall that $\sgl(q;R)$ is the special linear group;
$\mathsf{S}$ does not stand for ``simple'' as in Wall~\cite{WallSurgeryBook} or Ranicki~\cite{Ranicki1}.

\begin{lemma}\label{thm:suslin}
	For any $\alpha \in \sgl(q;R)$ with $q\ge 3$,
	we have $\cx(\alpha) \in [\qhp^\mp(q;R),\qhp^\mp(q;R)]$.
	For any $\theta \in \Mat(q;R)$, 
	we have $\qcz(\theta) \hat\oplus I \in [\qhp^\mp(2q),\qhp^\mp(2q)]$.
\end{lemma}

The first statement is a combination of Whitehead's lemma~\cite{MilnorWhiteheadTorsion} and Suslin's stability 
theorem~\cite{Suslin1977Stability}.
The second statement appears in the proof of~\cite[Thm.\,6.3]{WallSurgeryBook}.

\begin{proof}
	Let $E_{ij} \in \Mat(q;R)$ 
	be the matrix with entry $1$ at the $(i,j)$ position
	and zeroes elsewhere.
	Whitehead's lemma says that any matrix of form~$I+a E_{ij}$ for $a \in R$
	belongs to $[\sgl(q;R),\sgl(q;R)]$ if $q \ge 3$;
	proof: $(I+a E_{ij})(I+E_{jk})(I-a E_{ij})(I-E_{jk}) = I + a E_{ik}$ 
	where $i\neq j\neq k\neq i$.
	Since $p$ is a prime, 
	Suslin's stability theorem~\cite{Suslin1977Stability}
	says that $\sgl(q)$
	is generated by matrices of the form $I+ a E_{ij}$ if $q \ge 3$.
	Therefore, $\alpha$ belongs to the commutator subgroup,
	and $\cx(\alpha) \in [\qhp^\mp(q;R),\qhp^\mp(q;R)]$.

	Let $\beta = \begin{pmatrix}0 & \theta \\ 0 & 0 \end{pmatrix} \in \Mat(2q)$.
	Let $\alpha = \begin{pmatrix} I & 0 \\ I & I \end{pmatrix} \in \gl(2q;R)$.
	Then, a direct calculation shows that 
	$\qcz(\beta)\cx(\alpha)\qcz(-\beta) \cx(\alpha^{-1}) 
	= \qcz(\theta) \hat\oplus I$.
\end{proof}

\begin{lemma}\label{thm:elemModCommToF}
	Inclusions induce abelian group isomorphisms
	\begin{align}
		\sgl(\FF_p) \cap \ehp^\mp(\FF_p)/[\hp^\mp(\FF_p),\hp^\mp(\FF_p)] &\cong \sgl(R) \cap \ehp^\mp(R)/[\hp^\mp(R),\hp^\mp(R)],\\
		\sgl(\FF_p) \cap \qehp^\mp(\FF_p)/[\qhp^\mp(\FF_p),\qhp^\mp(\FF_p)] &\cong \sgl(R) \cap \qehp^\mp(R)/[\qhp^\mp(R),\qhp^\mp(R)].\nonumber
	\end{align}
	The group is generated by at most three elements~$\hada_\mp,\cx(a),\cz(1) \in \ehp(1;\FF_p)$ 
	where $a$ is any fixed nonsquare element of~$\FF_p$.
\end{lemma}

\begin{proof}
	Writing $\cx(\alpha)$ as the product of~$\cx(\det \alpha) \in \qehp(1)$
	and $\cx((\det \alpha)^{-1}) \cx(\alpha)$,
	we see that any element of $\qehp(R)$ is a product of $\cx(\alpha)$s 
	with $\det \alpha = 1$,
	$\qcz(\theta)$s,
	and $\hada$ up to commutators.
	By~\ref{thm:suslin} we see that 
	$\qehp = \langle \cx(a),\hada, [\qhp,\qhp]~|~0 \neq a \in \FF_p \rangle$.
	In the hermitian case, we have generators~$\cz(\theta)$ with $\theta$ diagonal.
	It is easy to see that $\theta$ can further be restricted to diagonal matrices over~$\FF_p$;
	we have done this in the proof of~\ref{thm:hermitianMatrices}.
	We still have an equation $\bar \theta = \pm \theta$,
	which can only be met with a nonzero~$\theta$ for~$\hp^-$.
	Moreover, a commutator $\cx(a)\hada_\mp\cx(1/a)\hada_\mp^{-1}$ is equal to~$\cx(a^2)$.

	Hence, the inclusion $\varphi : \ehp(\FF_p) \hookrightarrow \ehp(R)$ 
	induces a surjective map~$\bar\varphi$ on the quotient groups.
	The quadratic case is only a restriction of~$\varphi$.
	If $U$ over~$\FF_p$ belongs to the commutator subgroup over~$R$,
	then we obtain a commutator that gives $U$ by setting every variable~$x_i$ to~$1 \in R$.
	Hence, the map~$\bar\varphi$ is injective.
\end{proof}

\begin{proposition}\label{thm:circuit/commutator}
	We have the following.
	\begin{itemize}
	\item If $p=2$, then $(\sgl(R) \cap \qehp(R)) / [\qhp(R),\qhp(R)]$ has order~$2$ exactly, generated by~$\hada$.
	\item If $p=2$, then $(\sgl(R) \cap \ehp(R)) / [\hp(R),\hp(R)]$ has order~$\le 2$ generated by either~$\hada$ or~$\cz(1)$.
	\item If $p > 2$, then $(\sgl(R) \cap \qehp^-(R)) / [\qhp^-(R),\qhp^-(R)] = 0$.
	\item If $p > 2$, then $(\sgl(R) \cap \qehp^+(R)) / [\qhp^+(R),\qhp^+(R)]$ has order~$\le 2$ generated by~$\cx(a)$ with a nonsquare $a \in \FF_p$.
	\end{itemize}
\end{proposition}

\begin{proof}
	By~\ref{thm:elemModCommToF}, we examine $\hada$, $\cx(a)$, and $\cz(1)$.

	(The quadratic case with $p=2$)
	$\cz(1)$ is absent from~$\qhp(\FF_2)$.
	We only have to prove that $\hada \notin [\qehp(q;\FF),\qehp(q;\FF)]$ for any $q \ge 1$.
	Given $U \in \qehp(q;\FF)$, consider a composed map
	\begin{align}
		\phi : \qehp(q;\FF) \ni U = \begin{pmatrix}
		A & B \\ C & D
		\end{pmatrix}
		\mapsto
		U^\dag \eta U = \begin{pmatrix}
		A^\dag C & A^\dag D \\ B^\dag C & B^\dag D
		\end{pmatrix}
		\mapsto
		\Tr (B^\dagger C) \in \FF.
	\end{align}
	We claim that this map is a group homomorphism.
	First, a direct calculation shows that 
	$\phi(U\widetilde\cz(\theta)) - \phi(U) =\Tr (B^\dagger D (\theta + \theta^\dagger))$,
	but $B^\dagger D = D^\dagger B$ because of~\eqref{eq:DefLambdaUnitary}.
	Since $\Tr(J) = \Tr(J^\dagger)$ for any $J \in \Mat(q;\FF)$,
	we see that $\phi(U \qcz(\theta)) = \phi(U)$.
	Similarly, $\phi(U \qcz(\theta)^\dagger) = \phi(U)$.
	On the other hand, $\phi(\qcz(\theta)) = 0$
	and $\phi(\qcz(\theta)^\dagger) = 0$.
	Second, $\phi(U \cx(\alpha)) = \Tr(\alpha^{-1} B^\dagger C \alpha) = \Tr(B^\dagger C) = \phi(U)$ 
	and $\phi(\cx(\alpha)) = \Tr(0) = 0$.
	Finally, $\hada^\dagger U^\dagger \eta U \hada$ has lower left block~$C'$
	whose diagonal is different from that of $U^\dagger \eta U$ by at most one entry,
	which is swapped with the corresponding entry in the upper right block~$B'$.
	But, $(B')^\dagger + C' = I$ because of~\eqref{eq:DefLambdaUnitary}.
	Therefore, $\phi(U \hada) = \phi(U) + 1 = \phi(U) + \phi(\hada)$. 
	Since $\qehp$ is generated by these elements,
	we have proved that $\phi$ is a group homomorphism.

	Any commutator involves an even number of~$\hada$,
	so $\phi([\qehp(q;\FF),\qehp(q;\FF)]) = 0$ for any~\mbox{$q \ge 1$}.
	Since $\phi(\hada) = 1$, the proof is complete.

	(The hermitian case with $p=2$)
	$\cx(a)$ drops because there is no nonsquare.
	An identity $(\cz(1) \hada)^3 = I$ shows that $\cz(1) \hada$ is a commutator.

	(The ($-$)-case with $p > 2$)
	Since $\cz(\theta) = \qcz(\half\theta)$ is a commutator by~\ref{thm:suslin},
	both $\cz(\theta)$ and $\cz(\theta)^\dag$ belong to the commutator subgroup for any $\theta \in \FF_p$.
	The defining equation~$U^\dag \lambda_1^- U = \lambda_1^-$ is satisfied for all $U \in \sgl(2;\FF_p)$
	and $\sgl(2;\FF_p)$ is generated by the elementary row operations which are~$\cz(\theta)$ and~$\cz(\theta)^\dag$.
	Indeed,
	$
	\cx(a) = \cz(-1/a) \cz(a)^\dag \cz(-1/a) \hada_-^{-1}
	$ 
	and
	$ \hada_- = \cz(1)^\dag \cz(-1) \cz(1)^\dag $.

	(The ($+$)-case with $p > 2$)
	Since $\det \hada_+ = -1$ and $\bar 1 \neq - 1$, 
	we are left with $\cx(a)$ where $a$ is a nonsquare.
\end{proof}

\section{Time-reversal symmetry over qubits}\label{sec:timereversal}

In this section, we will set $R = \FF_2[x_1,\tfrac 1 {x_1},\ldots,x_\dd,\tfrac 1 {x_\dd}]$ in which~$2=0$.

Recall that the multiplicative group of all finitely supported Pauli operators on $\ZZ^\dd$, 
each of which is a tensor product of ``clock shift'' operator~$\sum_{j \in \FF_p} \ket{j+1}\bra{j}\in \Mat(p;\CC)$ and its Fourier transform~$\sum_j e^{2\pi i j / p} \ket{j}\bra{j} \in \Mat(p;\CC)$,
abelianizes to the additive group of Laurent polynomial ``vectors'' (column matrices)
in~$\dd$~variables with coefficients in~$\FF_p$.
Our results do not depend on particular generating set of the latter;
the only exception so far was 
when we associated specific local Clifford $\CC$-unitaries
to elementary $\lambda$-unitaries,
where the ``clock shift'' operators corresponded upper half blocks
and their Fourier transforms lower half blocks.

Here, this standard basis we used before is particularly convenient.
With $p=2$, the ``shift'' operator is Pauli $X = \begin{pmatrix} 0 & 1 \\ 1 & 0 \end{pmatrix}$
and its Fourier transform is Pauli $Z = \begin{pmatrix} 1 & 0 \\ 0 & -1 \end{pmatrix}$,
both of which are real matrices, invariant under complex conjugation.
Hence, the reality of a hermitian tensor product of Pauli matrices 
is determined by the mod~2 number of Pauli $Y = iXZ$ tensor components.
A Pauli $Y$ tensor component exists 
whenever the Laurent polynomial ``vector''~$v$ of the Pauli operator
has the same terms in the upper half $v_\text{up}$ ($X$ part) 
and the lower half $v_\text{down}$ ($Z$ part).
Observe that the ``constant term'' $\coe(v_\text{up}^\dag v_\text{down})$,
that does not involve any variable, of $v_\text{up}^\dag v_\text{down}$
is precisely the mod~2 number of Pauli $Y$ tensor components.

Therefore, a $\CC$-hermitian Pauli operator on~$\ZZ^\dd$ that gives a Laurent polynomial ``vector''~$v$
is real if and only if $\coe(v_\text{up}^\dag v_\text{down}) = 0$ or
\begin{align}
	\coe (v^\dag \eta v ) = 0 \label{eq:PauliReality}
\end{align}
where $\coe : R \to \FF_2$ is an $\FF_2$-linear map that reads off the term
in which the exponent of every variable is zero.
The map~$\coe$ is not a ring homomorphism.
Introduce a ring homomorphism $\ev: R \to \FF_2$ 
by which every variable~$x_j$ is mapped to~$1$.
The homomorphism~$\ev$ counts the mod~2 number of nonzero terms in a Laurent polynomial.
Hence, for any $r \in R$, it holds that 
\begin{align}
	\ev (r) = \coe(\overline{r} r).\label{eq:evcoe}
\end{align}

\begin{proposition}\label{thm:timereversal}
	Let $U \in \hp(R)$ where $0 = 2 \in R$.
	Then, $U$ belongs to the $\eta$-unitary group~$\qhp(R)$
	if and only if the corresponding Clifford QCA maps every real operator to a real operator.
	Every class of~$\umod(R)$ has a representative that belongs to the $\eta$-unitary group~$\qhp(R)$,
	{\em i.e.},
	the map~$\qumod(R) \to \umod(R)$ induced by the inclusion~$\qhp(R) \hookrightarrow \hp(R)$ is surjective.
\end{proposition}

The idea below is excerpted from the proof of~\cite[Thm.\,IV.4]{nta3}.

\begin{proof}
	Since Pauli $X$ and $Z$ generate the real algebra of operators,
	a Clifford QCA maps every real operator to a real operator
	if and only if
	the images of all Pauli $X$ and $Z$ are real operators.
	These images are nothing but the columns of $U \in \hp(R)$ up to a sign $\pm \in \RR$.
	Therefore, by~\eqref{eq:PauliReality},
	the images of Pauli $X$ and $Z$ under a Clifford QCA represented by~$U \in \hp(q;R)$,
	are all real if and only if
	\begin{align}
		\coe\diag \left[U^\dag \begin{pmatrix} 0 & I \\ 0 & 0 \end{pmatrix} U \right] = 0. \label{eq:reality}
	\end{align}
	This equation precisely means that the diagonal entries of $U^\dag \eta U$ are even hermitian forms,
	which is to say that $U$ is $\eta$-unitary.
	This proves the first statement.

	To prove the second statement,
	let $[V] \in \umod(R)$.
	Since $\hp(\FF_2) = \ehp(\FF_2)$ and $(\ev V)^\dagger \lambda (\ev V) = \lambda$,
	we see that $[V] = [(\ev V)^{-1} V]$ and $\ev( (\ev V)^{-1} V ) = I$.
	Thus, we may assume that 
	\begin{align}
		\ev V = I \in \hp(q;\FF_2).\label{eq:evV=1}
	\end{align}
	Given $V = \begin{pmatrix} \sigma & \tau \end{pmatrix}$ satisfying~\eqref{eq:evV=1},
	we define a diagonal matrix~$P = \coe\diag(\sigma^\dag \eta \sigma) = P^\dag$.
	Then,
	\begin{align}
		\coe\diag[\sigma^\dag \cz(P)^\dag \eta \cz(P) \sigma]
		&=
		\coe\diag\left[\sigma^\dag \begin{pmatrix} P & I \\ 0 & 0 \end{pmatrix} \sigma\right]
		=
		\coe[\sum_j \overline{\sigma_{ji}} P_{jj} \sigma_{ji}]_i + [P_{ii}]_i \nonumber\\
		&=
		\left(\sum_j P_{jj} \coe[\overline{\sigma_{ji}}\sigma_{ji}]_i \right) + [P_{ii}]_i & \text{($\coe$ is $\FF_2$-linear)} \\
		&=	
		[P_{ii}]_i + \sum_j P_{jj} [\ev(\sigma_{ji})]_i & \text{\eqref{eq:evcoe}}\nonumber\\
		&= 
		[P_{ii}]_i + \sum_j P_{jj} [\delta_{ji}]_i  = 0. & \text{\eqref{eq:evV=1}}\nonumber
	\end{align}
	Thus, we may assume that a representative~$V$ of~$[V]$ has the left half block~$\sigma$
	with~\mbox{$\coe\diag(\sigma^\dag \eta \sigma) = 0$}.

	Now, we examine the right half block~$\tau$ of~$V$.
	Let $F = \coe\diag(\tau^\dag \eta \tau) = F^\dag$ be a diagonal matrix.
	Consider~$V \cz(F)^\dag = \begin{pmatrix}
	\sigma &&& \tau + \sigma F
	\end{pmatrix}$.
	Since $F\sigma^\dag \eta \sigma F$ is even and $\tau^\dag \lambda \sigma = I$,
	we see
	\begin{align}
		\coe\diag[(\tau + \sigma F)^\dag \eta (\tau + \sigma F)]
		&=
		F + \coe\diag[\tau^\dag \eta \sigma F + F \sigma^\dag \eta \tau + F \sigma^\dag \eta \sigma F]\\
		&=
		F + \coe\diag[\tau^\dag \eta \sigma F + \tau^\dag \eta^\dag \sigma F]\nonumber\\
		&=F + \coe\diag[\tau^\dag \lambda \sigma F] = 0.\nonumber
	\end{align}
	This completes the proof.
\end{proof}

Hence, every $\dd$-dimensional Clifford QCA over qubits, $\CC^2$,
is equivalent to one that is time reversal symmetric.
(See also~\ref{rem:oneDimNontrivialQuadraticUnitary}.)
This conclusion will actually be implied by our main classification result
since every class of $\lambda$-unitaries corresponds to a hermitian Witt class in one dimension lower,
which in turn corresponds down to a $\eta$-unitary class.
However, our proof of the main result will use~\ref{thm:timereversal} 
in the proof of~\ref{thm:oddBassExactNonquadratic} below.

\section{Descent and Ascent}\label{sec:updown}

We reserve a variable~$z$ to be
independent from any other variables~$x_1,\ldots,x_\dd$.
We denote by~$\Rz$ 
a commutative Laurent extension~$R[z,\tfrac 1 z] = \FF_p[x_1,\tfrac 1 {x_1},\ldots,x_\dd,\tfrac 1 {x_\dd},z,\tfrac 1 z]$,
which is a Laurent polynomial ring in $\dd+1$ variables.
The $\FF_p$-linear involution on~$R$ extends to~$\Rz$ by~$z \mapsto \tfrac 1 z$.
The goal of this section is to prove a theorem of Novikov~\cite{Novikov1} and Ranicki~\cite{Ranicki2}
which establishes abelian group isomorphisms
\begin{align}
	\qumod^\mp(\Rz) &\cong \qumod^\mp(R) \oplus \qwitt^\mp(R),\label{eq:descentIsomorphisms}\\
	\qwitt^\mp(\Rz) &\cong \qwitt^\mp(R) \oplus \qumod^\pm(R).	\nonumber
\end{align}
These are conveniently combined by
\begin{definition}[\cite{Ranicki1}]\label{defn:vtheory}
	For $n \in \ZZ/4\ZZ$ and $R = \FF_p[x_1,\tfrac 1 {x_1},\ldots,x_\dd,\tfrac 1 {x_\dd}]$ we define
	\begin{align}
		\vtheory_n (\dd,p) = \left\{\begin{matrix}
			\qwitt^+(R) & (n=0), &&&	\qumod^+(R) & (n=1),\\
			\qwitt^-(R) & (n=2), &&&	\qumod^-(R) & (n=3).
		\end{matrix}\right. 
	\end{align}
\end{definition}

Then,~\eqref{eq:descentIsomorphisms} is cast into

\begin{align}
	\vtheory_n(\dd,p) \cong \vtheory_n(\dd-1,p) \oplus \vtheory_{n-1}(\dd-1,p).
\end{align}
It obviously follows that $\vtheory_n(\dd,p)$ is a direct sum of some number of copies of~$\vtheory_k(0,p)$ for various~$k$.
More precisely, $\vtheory_n(\dd,p) 
	\cong 
	\vtheory_0(0,p)^{\oplus m_0} 
	\oplus \vtheory_1(0,p)^{\oplus m_1} 
	\oplus \vtheory_2(0,p)^{\oplus m_2} 
	\oplus \vtheory_3(0,p)^{\oplus m_3}$
where $ m_0 + m_1 x + m_2 x^2 + m_3 x^3 = x^n(x^{-1} + 1)^\dd \in \ZZ[x]/(x^4-1)$ if $p$ is odd,
and
$\vtheory_n(\dd,p) \cong \vtheory_0(0,p)^{\oplus m_0} \oplus \vtheory_1(0,p)^{\oplus m_1}$ where $m_0 + m_1 x = x^n(x^{-1}+1) \in \ZZ[x]/(x^2 -1)$ if $p=2$.
We know from~\ref{thm:n0elem} and~\ref{thm:ZeroDcalculation} that $\vtheory_n(0,p)$ vanishes except for $n=0$,
where $\vtheory_0(0,p)$ is isomorphic to, depending on $p$, $\ZZ/2\ZZ$, $\ZZ/4\ZZ$, or $\ZZ/2\ZZ \oplus \ZZ/2\ZZ$.
So, all the groups~$\vtheory_n(\dd,p)$ are determined.

However, our ultimate object~$\clifqca(\dd,p)$ is not $\vtheory_n(\dd,p)$.
The most relevant group is~$\vtheory_3(\dd,p)$, 
but our equivalence relation for unitaries is stronger (with fewer equivalence classes)
and a Clifford QCA is not necessarily an $\eta$-unitary but is only a $\lambda^-$-unitary.
The latter distinction is immaterial for $p$ odd by~\ref{thm:qhphp} but important for $p=2$.
We will address these in~\S\ref{sec:cg}.

We will mostly follow Ranicki~\cite{Ranicki2}
and all results in this section are special cases of those there,
except one minor difference in the going-up morphism~$\bass^\uparrow_{2\iota+1}$
from a unitary group to a Witt group,
where our morphism lands in an even hermitian Witt group~$\switt^\pm$,
while Ranicki's lands in a quadratic Witt group~$\qwitt^\pm$.
See~\ref{thm:WittGroupOfForms} and~\ref{thm:hermitianMatrices} to recall the definitions.
We will show in~\ref{thm:tildeBass} that Ranicki's map (after a minor correction) is well defined,
but our logical flow is slightly different from Ranicki's.
Also, it appears that the proof of~\cite[Lem.\,3.4]{Ranicki2} is not complete,
so we present our amended calculation in~\ref{thm:up-down-and-down-up}.
Our exposition will be self-contained,
except the Quillen--Suslin--Swan theorem~\cite{Suslin1977Stability,Swan1978}.

\subsection{Modular bases}\label{sec:modularBase}

For a finitely generated $\Rz$-module~$Q$, 
an $R$-module~$A$ that is set-theoretically contained in~$Q$ 
is said to be a {\bf modular $R$-base}
of~$Q$ if \mbox{$Q = \bigoplus_{k = -\infty}^\infty (z^k A)$}, {\it i.e.},
every $v \in Q$ is uniquely expressed as $v=\sum_k z^k a_k$
where all but finitely many~$a_k \in A$ are zero~\cite{Ranicki2}.
We introduce notations associated with a modular $R$-base~$A$ of~$Q$:
\begin{align}
	A^+ &= \projz^{\ge 0}_A Q = \bigoplus_{k=0}^\infty z^k A  \subset Q, &
	A^- &= \projz^{<0}_A Q = \bigoplus_{k=-\infty}^{-1} z^k A  \subset Q. 
\end{align}
Here, $\projz^{\cdots}_A$ is an $R$-linear, not $\Rz$-linear, 
projection on the the designated range of $z$-exponents.
It depends on the chosen modular $R$-base~$A$ in the subscript.
By abuse of notation, we sometimes write~$u \in \projz^{\cdots}_A$ to mean~$u \in \projz^{\cdots}_A Q$,
{\it i.e.}, a projector also stands for its image.

A modular $R$-base $A$ of an $\Rz$-module~$Q$ 
is isomorphic to the $R$-module~$Q/(z-1)Q$
as seen by the projection $Q \ni \sum_k z^k a_k \mapsto \sum_k a_k \in A \subset Q$.
It follows that any two modular $R$-bases are isomorphic as $R$-modules.
Since a finitely generated free $\Rz$-module~$\Rz^m$ has a modular $R$-base~$R^m$,
which is finitely generated free over~$R$,
every modular $R$-base of~$\Rz^m$ is finitely generated free over~$R$.

Suppose we have two modular $R$-bases~$A,B$ of~$\Rz^m$.
Each $R$-generator of~$A$ belongs to 
some finite direct sum of copies~$z^k B$ of~$B$,
and therefore there exists an integer~$n \ge 0$ such that
\begin{align}
	z^n A^+ \subseteq B^+,
\end{align}
with which we define
\begin{align}
	\bd_n(A,B) = (z^n A^-) \cap B^+. \label{eq:bdAB}
\end{align}
In other words, $B^+ = \bd_n(A,B) \oplus z^n A^+$;
the unique expansion of an element of~$B^+$ with respect to the $R$-base~$A$ 
gives projections onto the direct summands.
It follows that
\begin{align}
	\bd_{n+1}(A,B) = B \oplus z \bd_n(A,B),\label{eq:bdABoneMore1}
\end{align}
so a different choice of~$n$ in $\bd_n$ only adds copies of~$B$ to $\bd_n$.
Similarly, if $n$ is sufficiently large that we can define~$\bd_n(A,B)$ and also that $z^n A^- \supseteq B^-$, 
then $z^n A^- = \bd_n(A,B) \oplus B^-$,
so 
\begin{align}
	\bd_{n+1}(A,B) = z^n A \oplus \bd_n(A,B). \label{eq:bdABoneMore2}
\end{align}
If $C$ is another modular $R$-base of $\Rz^m$, 
then there exists an integer $n'\ge 0$ so large that $z^{n'} B^+ \subseteq C^+$.
We also have $z^{n'+n} A^+ \subseteq C^+$, and $\bd_{n+n'}(A,C)$ is defined.
It follows that
\begin{align}
	C^+ 
	&= \bd_{n'}(B,C) \oplus z^{n'} B^+ \nonumber\\
	&= \bd_{n'}(B,C) \oplus z^{n'} (\bd_n(A,B) \oplus z^{n} A^+)\nonumber\\
	&= \bd_{n'+n}(A,C) \oplus z^{n'+n} A^+ ,\nonumber\\
	\bd_{n'+n}(A,C)
	&=z^{n'} \bd_n(A,B) \oplus \bd_{n'}(B,C). \label{eq:bass-Sum}
\end{align}
If we set $A=C$, then 
\begin{align}
	\bd_{n'+n}(A,A) = \bigoplus_{k=0}^{n'+n-1} z^k A = z^{n'} \bd_n(A,B) \oplus \bd_{n'}(B,A)
	\label{eq:bdn-projective}
\end{align}
which implies that $\bd_n(A,B)$ is a direct summand of a finitely generated free $R$-module,
and hence is also finitely generated free due to the Quillen--Suslin--Swan theorem.

\subsection{Duals}

The Laurent extension~$\Rz$ has a canonical modular $R$-base~$R$;
every Laurent polynomial~$r \in \Rz$ 
has a unique expansion $r = \sum_k z^k r_k$ with~$r_k \in R$.
For any $r \in \Rz$, we write~$[r]_0$ or $[r]_{z^0}$
to mean the $z$-degree zero term~$r_0$ in this expansion.

Given a modular $R$-base~$A$ of a finitely generated $\Rz$-module~$Q$,
let 
\begin{align} 
	Q^{*c} = \{ f \in \Hom_R(Q,R)~|~f(z^j A) = 0 \text{ for all but finitely many }j\}
\end{align}
be the module of all $R$-linear functionals of {\bf compact support}.
This notion of compact support is defined only if there is a modular $R$-base,
but since any modular $R$-base is contained in a finite direct sum of any other modular $R$-base,
the set~$Q^{*c}$ is independent of a modular $R$-base.
Note that $Q^{*c}$ is an $\Rz$-module by $(z \cdot f)(\circ) = f( \bar z\, \circ)$.

\begin{lemma}\label{thm:CompactDualIsDual}
	For any finitely generated $\Rz$-module~$Q$ that has a modular $R$-base~$A$,
	the dual $\Rz$-module~$Q^*$ is canonically isomorphic to~$Q^{*c}$ 
	by~$Q^* \ni \hat f \mapsto f = [\hat f(\circ)]_0 \in Q^{*c}$.
\end{lemma}

\begin{proof}
	The map lands in~$Q^{*c}$ because for any $u \in Q$,
	$[\hat f(z^k u)]_0$ is eventually zero as $|k| \to \infty$
	and there are only finitely many generators of~$A$.
	The map is $\Rz$-linear since $z \cdot \hat f = \hat f( \bar z \,\circ) 
	\mapsto [\hat f(\bar z \,\circ)]_0 = z \cdot [\hat f(\circ)]_0$.
	Given $f \in Q^{*c}$, let $F(u) = \sum_j z^j f(z^{-j} u)$ for all~$u \in Q$,
	which is always a finite sum because~$f$ is compactly supported.
	Since $F(z u) = \sum_j z^j f(z^{-j+1} u) = \sum_j z^{j+1} f(z^{-j} u) = z F(u)$,
	we see that~$F$ is a $\Rz$-linear functional on~$Q$.
	Now, $[F(u)]_0 = f(u)$ and $\sum_j z^j [\hat f(z^{-j} u)]_0 = \hat f(u)$,
	implying that the map in the lemma is invertible.
\end{proof}

An obvious $R$-linear functional of compact support is defined by $f \in A^*$ and an integer~$k$
as~$Q \supset z^j A \cong A \xrightarrow{f} R$ if~$j = k$ and $z^j A \to 0$ if~$j \neq k$.
We may write the set of all such ones as~$z^k A^*$.
This notation makes a perfect sense if we identify~$f \in A^*$ with~$f([\circ]_0) : Q \to R$
since, then, $z^k \cdot f = f([ z^{-k}\, \circ]_0)$.
Any other $R$-linear functional of compact support must be a finite sum of these $R$-linear functionals,
and we see that $Q^{*c} = \bigoplus_j z^j A^*$.
By~\ref{thm:CompactDualIsDual}, 
we have $Q^* \cong \bigoplus_j z^j A^*$ as $\Rz$-modules.
This means that
\begin{align}
	\text{``}A^*\text{''} = \left\{ \sum_{j\in \ZZ} z^j f([z^{-j}\circ]_0) \in \Hom_\Rz(Q,\Rz) ~\middle|~ f \in \Hom_R(A,R) \right\}
\end{align}
is a modular $R$-base of~$Q^*$.
\emph{This choice of a modular $R$-base of~$Q^*$ is our convention and is applied always below.}
The associated projection maps~$\projz^{\cdots}_{\text{``}A^*\text{''}}$
will be denoted by~$\projz^{\cdots}_{A^*}$.

Note that this choice of an $R$-base makes the adjoint 
of a projection~$\projz_A^{\cdots}$ convenient.
Since $\projz_A^{\cdots}$ is a just $R$-linear map on~$Q$,
we should consider $R$-duals of compact support.
A general element of $Q^{*c}$ is written 
as~$\sum_{j\in \ZZ} z^j \cdot f_j \in \bigoplus_{j \in \ZZ} z^j A^*$ 
with~$f_j([\circ]_0) : Q \to R$.
Then, for~$J \subseteq \ZZ$ and~$\circ = \sum_{i\in\ZZ} z^i a_i \in Q$ with~$a_i \in A$,
we have
\begin{align}
	(\projz_A^J)^\dag \sum_{j\in \ZZ} z^j \cdot f_j
	&=
	\sum_{j\in\ZZ} (\projz_A^{J})^\dag f_j([z^{-j} \,\circ]_0)
	=
	\sum_{j\in\ZZ} f_j([z^{-j} \projz_A^J \,\circ]_0)
	=
	\sum_{j\in\ZZ} f_j([ \sum_{i \in J} z^{i-j} a_i]_0)
	\\&	
	=
	\sum_{j \in J} f_j(a_j)
	=
	\sum_{j \in J} f_j([z^{-j} \,\circ]_0)
	=
	\sum_{j \in J} z^j \cdot f_j
	=
	\projz_{A^*}^J \sum_{j \in \ZZ} z^j \cdot f_j.
	\nonumber
\end{align}
Hence, for any $J \subseteq \ZZ$,
\begin{align}
	(\projz_A^J : Q \to Q)^\dag = \projz_{A^*}^J : Q^* \to Q^*. \label{eq:projzDagger}
\end{align}

\begin{remark}\label{rem:DeltaOverR}
	A nonsingular hermitian form~$\Delta$ on a finitely generated $\Rz$-module~$Q$ 
	is by definition an $\Rz$-isomorphism between $Q$ and its $\Rz$-dual~$Q^*$.
	If $Q$ has a modular $R$-base,
	we have shown in~\ref{thm:CompactDualIsDual} that $Q^* \cong Q^{*c}$ as~$\Rz$-modules,
	and therefore we may think of~$\Delta$ as an isomorphism from~$Q$ to~$Q^{*c}$.
	Since the functionals of~$Q^{*c}$ are valued in~$R$,
	the module~$Q^{*c}$ lives in the category of $R$-modules.
\end{remark}

\subsection{Boundary forms}

\begin{lemma}[Lemma~2.1 of~\cite{Ranicki2}]\label{thm:boundaryXi-nonsingular}
	Let $\Phi$ (\emph{resp.} $\phi$) be a nonsingular hermitian (\emph{resp.} quadratic) 
	$\mp$-form on a finitely generated $\Rz$-module~$Q$ that has a modular $R$-base~$A$.
	Suppose that a finitely generated $R$-module~$C$
	is a direct summand of~$Q = C \oplus D$
	such that $[\Phi(c,d)]_0 = 0$ (\emph{resp.} $[(\phi \mp \phi^\dag)(c,d)]_0 =0$)
	for all $c\in C$ and $d \in D$.
	Then, a hermitian $\mp$-form on the $R$-module~$C$ defined by~$[\Phi(a,b)]_0$ 
	(\emph{resp.} a quadratic $\mp$-form $[\phi(a,b)]_0$)
	for $a,b \in C$ is nonsingular.
	If $\Phi$ is even, so is the induced form on~$C$.
\end{lemma}

The proof was omitted in~\cite{Ranicki2}.
We denote by $[ \Phi |_C ]_0$ or $[\phi|_C]_0$ the induced nonsingular form on~$C$.

\begin{proof}
	The quadratic case follows from a proof for the hermitian case
	since nonsingularity and orthogonality are defined through the associated hermitian form.
	We have to show that 
	$C \ni c \mapsto (x \mapsto [\Phi(c,x)]_0) \in \Hom_R(C,R)$
	is an isomorphism.
	Let $\pi_C : C \oplus D \to C$ be the projection.

	For any nonzero~$c \in C$, there exists $x \in Q$ 
	such that $\Phi(c,x) = \sum_k [\Phi(c,x)]_k z^k \in \Rz$ is nonzero,
	and therefore $= [\Phi(c, \pi_C(z^k x))]_0 = [\Phi(c, z^k x)]_0 = [z^{k} \Phi(c,x)]_0$ 
	is nonzero for some~$k$.
	Hence, the map $C \ni c \mapsto [\Phi(c, \circ)]_0 \in \Hom_R(C,R)$ is injective.

	To show its surjectivity,
	let $g \in \Hom_R(C,R)$.
	We extend $g$ to $f \in \Hom_R(C \oplus D,R)$ by letting $f = 0$ on~$D$.
	We claim that $f$ is compactly supported with respect to~$\bigoplus_j z^j A$.
	To see this, we observe that any element of~$C$ 
	is a finite sum~$\sum_j a_j z^j$ for some~$a_j \in A$.
	Since $C$ is finitely generated, 
	we must have $C \subseteq \bigoplus_{j = -n}^n z^j A$ for some sufficiently large~$n$.
	Hence, for any $x \in C \oplus D$, 
	there are only finitely many~$j$ such that $\pi_C(z^j x)$ is nonzero.
	Hence, $f(z^j x) = f(\pi_C(z^j x))$ is nonzero only for finitely many~$j$.
	By~\ref{thm:CompactDualIsDual}, we find an $\Rz$-linear functional~$F$ such that~$[F(\circ)]_0 = f$,
	which restricts to~$g$ on~$C$.
	Since $\Phi$ is nonsingular, $F = \Phi(y, \circ)$ for some~$y \in Q$.
	Therefore, $\Hom_R(C,R) \ni [\Phi(\pi_C y, \circ)]_0 = g(\circ)$.
\end{proof}

\subsection{Going down from unitaries to forms}

We will apply the constructions of~$\bd_n$ and of the induced form in~\ref{thm:boundaryXi-nonsingular} 
to~$\hp^\mp$ and~$\qhp^\mp$.
The domain of a $\lambda$- or $\eta$-unitary over~$\Rz$ is $\Rz^q \oplus (\Rz^q)^*$.
Let $A = R^q \oplus R^{q*}$ be a modular $R$-base, so $\Rz^q \oplus \Rz^{q*} = \bigoplus_j z^j A$.

\begin{proposition}[\S2 of~\cite{Ranicki2}]\label{thm:boundaryXi-from-QCA}
	For any unitary~$U \in \hp^\mp(q;\Rz)$, 
	choose $n$ so large
	that $z^n(U A)^+ \subseteq A^+$.
	The induced form~$[\lambda|_{\bd_n(UA,A)}]_0$ 
	is a nonsingular even hermitian $\mp$-form over~$R$
	whose Witt class
	depends only on the class of~$U$ in~$\umod^\mp(\Rz)$.
	If $U \in \qhp^\mp(q;\Rz)$, then $[\eta|_{\bd_n(UA,A)}]_0$ is nonsingular quadratic $\mp$-form
	whose Witt class depends only on the class of~$U$ in~$\qumod^\mp(\Rz)$.
\end{proposition}

We call the result a {\bf boundary $\mp$-form} of~$U$.
For $[U] \in \qumod^\mp$, we denote the boundary form by~$\bass^\downarrow_{2\iota+1}([U])$
where $\iota = 0$ is for $\bass^\downarrow$ from $\qumod^+$ and $\iota=1$ is for $\bass^\downarrow$ from $\qumod^-$.
Note that the boundary $\mp$-form of~$U \in \hp^\mp \setminus \qhp^\mp$
is \emph{not} defined as a quadratic $\mp$-form;
it is merely an even form.
Of course, if $\half \in R$, 
the quadratic Witt group is the same as the even hermitian Witt group.

The even ($-$)-hermitian case has appeared in~\cite{clifqca1} as ``boundary antihermitian form.''
As explained in~\cite{clifqca1} 
this is a version of the boundary algebra~\cite{GNVW,FreedmanHastings2019QCA} tailored
to translation invariant Clifford QCA.
The statement of~\ref{thm:boundaryXi-from-QCA} is slightly different from that in~\cite{Ranicki2}
because our starting point of~$\umod^\mp$ is different;
\cite{Ranicki2} uses ``formations,'' pairs of lagrangians of a trivial form,
and says that it is the same as the unitary group in a sense.
The notion of formations was introduced by Wall~\cite{WallSurgeryBook} 
and developed by Novikov~\cite{Novikov1} and Ranicki~\cite{Ranicki1}.

\begin{proof}
	The quadratic case will parallel or be subsumed by the hermitian case.

	We first show that $C = \bd_n(U A,A)$ satisfies the assumption of~\ref{thm:boundaryXi-nonsingular}.
	By definition, $\Rz^{2q} = A^- \oplus \bd_n(U A,A) \oplus z^n(UA)^+$,
	so we put $D = A^- \oplus z^n (UA)^+$.
	By \eqref{eq:bdn-projective}, we know $C$ is a finitely generated projective $R$-module, 
	which is free by the Quillen--Suslin--Swan theorem~\cite{Suslin1977Stability,Swan1978}.
	Let~$c \in C$. 
	If~$d \in A^-$, then~$\lambda(c,d) \in \Rz$ has strictly negative $z$-exponents since~$c \in A^+$.
	So,~$[\lambda(c,d)]_0 = 0$.
	If~$d \in z^n(U A)^+ = z^n U A^+$, 
	then there are $c' \in A^-$ and $d' \in A^+$ such that~$c = z^n U c'$ and~$d = z^n U d'$
	because~$c \in z^n(UA)^- = z^n UA^-$.
	This implies that~$\lambda(c,d) = \lambda(c', d')$ 
	has strictly positive $z$-exponents, 
	and hence~$[\lambda(c,d)]_0=0$.
	Therefore,~\ref{thm:boundaryXi-nonsingular} implies that the induced form~$[\lambda|_C]_0$ on~$C$ is nonsingular.
	The form~$[\lambda|_C]_0$ is even because~$\lambda^\mp = \eta \mp \eta^\dagger$ is even.

	Next, we show that the Witt class of the boundary form is unchanged by
	arbitrary choices of~$n$, stabilization, or~$\ehp^\mp$.
	A different choice of~$n$, say~$n \to n+ m > n$, results in change~$C \to z^m C \oplus A^{\oplus m}$,
	where on the extra submodule~$A^{\oplus m}$ the induced form is an orthogonal summand~$\lambda_m$ 
	(or~$\eta_m$ in the case of~$\eta$-unitaries),
	which gives the zero Witt class.
	Since a direct sum~$U \hat\oplus V$ gives a direct sum of forms
	and the identity unitary gives a trivial form, any stabilization is fine.

	It remains to show that~$\bass^\downarrow(U) = [\lambda|_{\bd_n(U A,A)}]_0$
	is Witt equivalent to~$\bass^\downarrow(EU)$ for any elementary unitary~$E$.
	From~\eqref{eq:bass-Sum}, we know 
	$\bd_{n+m}(EUA,A) = z^m \bd_{n} (EUA,EA) \oplus \bd_m (EA,A)$ for sufficiently large~$n,m$,
	where the two summands are orthogonal with respect to~$[\lambda]_0$ since~$z^m E A^+ \perp z^m E A^-$.
	The first summand~$\bd_{n} (EUA,EA)$ is equal to~$E \bd_n(UA,A) = E C$,
	giving a form $EC \times EC \ni (Ea, Eb) \mapsto [\lambda(Ea,Eb)]_0 = [\lambda(a,b)]_0$,
	which is obviously equivalent to~$[\lambda|_C]_0$.

	Therefore, it suffices to show that~$\bass^\downarrow(E)$ is Witt equivalent to a trivial form 
	for every generator~$E$ of~$\ehp$ (or~$\qehp$).
	The generator~$\hada_\mp$ does not involve any variable, in particular~$z$,
	so~$\bd_n(\hada A,A) = \bd_n(A, A) = A^{\oplus n}$,
	on which the induced form is trivial.
	For $E = \cx(\alpha)$, we have $\bd_n(EA,A) = \bd_n(\alpha R^q \oplus \alpha^{-\dag} R^{q*} , R^q \oplus R^{q*} )
	= \bd_n(\alpha R^q, R^q) \oplus \bd_n(\alpha^{-\dag} R^{q*}, R^{q*} )$
	where each direct summand is a lagrangian.
	So,~\ref{thm:lagrangian-trivialForm} implies that the induced form is trivial.
	For $E = \cz(\theta)$ (where $\theta$ is $\pm$-even for the quadratic case),
	let us find what $\bd_n(\cz(\theta) A, A) = (z^n \cz(\theta) A^-) \cap A^+$ is.
	Let $a' \oplus b' \in z^n A^- \subset \Rz^q \oplus \Rz^{q*}$ be arbitrary.
	To have $\cz(\theta)(a' \oplus b') \in A^+$,
	we must have $\projz^{<0} (a' \oplus (b' + \theta a')) = 0$.
	This means that $a' \in \projz^{\ge 0}$ and $\projz^{<0}b' = - \projz^{<0}\theta a'$.
	The latter condition is solved by setting~$b' = b - \projz^{<0}\theta a'$ for $b \in \projz^{\ge 0} \cap \projz^{<n}$.
	Then, $\cz(\theta)(a' \oplus b') = a' \oplus (b - \projz^{<0}\theta a' + \theta a') = a' \oplus (b + \projz^{\ge 0} \theta a')$.
	So,
	\begin{align}
	\bd_n( \cz(\theta) A, A ) 
	&\subseteq
		\left\{ 
			a \oplus (b + \projz^{\ge 0} \theta a) \in \projz^{\ge 0}_A ~\middle|~ a \oplus b \in \projz^{\ge 0}_A \cap \projz^{<n}_A 
		\right\}.\label{eq:bdnZtheta}
	\end{align}
	Conversely, any element~$a \oplus (b + \projz^{\ge 0} \theta a) \in \projz^{\ge 0}_A$ 
	on the right-hand side belongs to the left-hand side
	because it is the image of~$a \oplus (b - \projz^{<0} \theta a) \in \projz^{<n}_A$ under~$\cz(\theta)$.
	Hence, the inclusion in~\eqref{eq:bdnZtheta} is in fact an equality.
	An $R$-submodule~$\{0 \oplus b \in \projz^{\ge 0}_A \cap \projz^{<n}_A \}$ is a direct summand of
	the boundary module~$\bd_n( \cz(\theta) A, A )$ and is a lagrangian.
	Again,~\ref{thm:lagrangian-trivialForm} implies that the induced form is trivial.
	We have confirmed that $\bass^\downarrow(E)$ is trivial for all generators~$E$ of the elementary unitary group.
\end{proof}

An alternate characterization of the boundary module and form is as follows.

\begin{lemma}[\S2 of \cite{Ranicki2}]\label{thm:separatorOnly}
	Let $U \in \hp^\mp(q;\Rz)$.
	Let $S_0$ be any modular $R$-base of a lagrangian $U(\Rz^{q} \oplus 0)$.
	For a sufficiently large~$n$ we have a direct summand~$S = z^n (S_0)^+$ of the $R$-module~$A^+$.
	Then, the induced form on~$S^\perp / S$ 
	where $\perp$ is taken within $A^+$ with respect to~$[\lambda_q]_0$,
	is Witt equivalent to a boundary $\mp$-form that is even hermitian (or quadratic if $U \in \qhp^\mp(q;\Rz)$).
\end{lemma}

This is a tailored version of the statement that
the boundary algebra of a general QCA is the commutant of the set~$S$ 
of all separator elements~\cite{nta3} in an infinite half space, modulo~$S$.

\begin{proof}
	$U^{-1} S_0$ is a modular $R$-base of a lagrangian~$\Rz^{q} \oplus 0$.
	Let $\alpha$ be an $\Rz$-automorphism of~$\Rz^{q}$
	that brings~$R^q \oplus 0$ onto~$U^{-1} S_0$.
	So, $S_0 = U(\alpha R^q \oplus 0)$.
	Let $S_0^* = U(0 \oplus \alpha^{-\dag} R^{q*})$ be a modular $R$-base of the lagrangian $U (0 \oplus \Rz^{q*})$
	that is paired with $U(\Rz^q \oplus 0)$.
	Assume that $n$ is so large that $z^n (S_0^*)^+ \subseteq A^+$.
	That is, $S_0 \oplus S_0^* = U \cx(\alpha) A$ is a modular $R$-base of~$\Rz^{q}\oplus \Rz^{q*}$,
	with which the induced form~$[\lambda_q|_{\bd_n(U\cx(\alpha)A,A)}]_0$ is
	Witt equivalent to the boundary form by~\ref{thm:boundaryXi-from-QCA}
	since $\cx(\alpha)$ is elementary.

	Since $\lambda_q$ restricted to~$S_0 \oplus S_0^*$ is trivial,
	the decomposition~$z^n(S_0 \oplus S_0^*)^+ \oplus z^n(S_0 \oplus S_0^*)^-$ is orthogonal with respect to~$[\lambda_q]_0$.
	It follows that $S^\perp = S \oplus \bd_n(U\cx(\alpha)A,A)$.
	Moreover, $S^\perp \oplus z^n (S_0^*)^+ \oplus A^- = \Rz^{q} \oplus \Rz^{q*}$ as $R$-modules,
	implying that $S^\perp$ is projective.
	Hence, \ref{thm:FormOnQuotientModule} is applicable and the induced form on 
	$S^\perp /S \cong \bd_n(U\cx(\alpha)A,A)$
	is Witt equivalent to the boundary form.
\end{proof}

\subsection{Going up from forms to unitaries}

\begin{proposition}\label{thm:form-to-up-unitary}
	Let $[\xi] \in \qwitt^\mp(R)$ (resp. $[\Delta] \in \switt^\mp(R)$).
	Let $T : Q \oplus Q^* \to Q \oplus Q$ be an isomorphism of $R$-modules such that 
	$T^\dagger \diag(\xi,-\xi) T \sim \eta$ (resp. $T^\dagger \diag(\Delta,-\Delta) T = \lambda$).
	Let $U = T^{-1} \diag(zI,I) T$ be an automorphism of~$\bigoplus_j z^j (Q \oplus Q^*)$ over~$\Rz$.
	Then, the map $\bass^\uparrow_{2\iota} : \qwitt^\mp(R) \ni [\xi] \mapsto [U] \in \qumod^\mp(\Rz)$ 
	(resp. $\bass^\uparrow_{2\iota} : \switt^\mp(R) \ni [\Delta] \mapsto [U] \in \umod^\mp(\Rz)$)
	is a well-defined group homomorphism.
\end{proposition}

The notation is that $\bass^\uparrow_0$ is from $\qwitt^+$ 
and $\bass^\uparrow_2$ is from $\qwitt^-$.
This is an instance of Eilenberg's swindle construction.
The quadratic case appears in~\cite[\S 2]{Ranicki2}
and the even hermitian case appears in~\cite[III.\,15]{clifqca1}.

\begin{proof}
	Such $T$ exists due to~\ref{thm:WittGroupOfForms} (resp.~\ref{thm:hermitianMatrices}).
	The automorphism~$U$ belongs to~$\qhp^\mp$ (resp.~$\hp^\mp$) 
	because~$\xi$ (resp.~$\Delta$) commutes with~$z$.
	A different choice of~$T$ amounts to $T \to T V$ for some unitary $V$.
	This changes $U$ to~$V^{-1} U V$ which is equal to~$U$ modulo elementary unitaries
	because the quotient group is abelian by~\ref{thm:circuit-contains-commutators}.

	Adding a trivial form to~$\xi$ does not change the Witt class, 
	and the corresponding~$U$ changes as $U \to U \hat\oplus \cx(zI)$.
	A congruent transformation $\xi \to F^\dag \xi F$ (resp. $\Delta \to F^\dag \Delta F$) 
	amounts to $T \to \diag(F,F)T$,
	but $\diag(F,F)$ commutes with~$\diag(zI,I)$, resulting in no change in~$U$.
	Therefore, $\bass^\uparrow$ is well defined.

	$\bass^\uparrow$ is a group homomorphism since 
	$[U \hat\oplus V] = [(U \hat\oplus V)(V \hat\oplus V^{-1})] = [(UV) \hat\oplus I]$ by~\ref{thm:TRCisInvU}.
\end{proof}

\begin{proposition}\label{thm:fromForm-up-and-down-is-identity}
	Both compositions in the following are the identities.
	\begin{align}
		\qwitt^\mp(R) \xrightarrow{\quad\bass^\uparrow_{2\iota}\quad} \qumod^\mp(\Rz) 
		\xrightarrow{\quad\bass^\downarrow_{2\iota+1}\quad} \qwitt^\mp(R) \\
		\switt^\mp(R) \xrightarrow{\quad\bass^\uparrow_{2\iota}\quad} \umod^\mp(\Rz) 
		\xrightarrow{\quad\bass^\downarrow_{2\iota+1}\quad} \switt^\mp(R)\nonumber
	\end{align} 
	In particular, $\bass^\downarrow_{2\iota+1}$ is surjective for both cases.
\end{proposition}

The quadratic case appears in~\cite[\S 2]{Ranicki2}
and the even hermitian case appears in~\cite[III.\,15]{clifqca1}.

\begin{proof}
	Let $[\xi] \in \qwitt^\mp(R)$ where $\xi$ is on~$Q$.
	Let $A = Q\oplus Q^*$ be a modular $R$-base of the domain~$\bigoplus_j z^j (Q \oplus Q^*)$
	of a unitary~$U = T^{-1}\diag(zI,I)T$.
	The modular $R$-base~$B = U A$ is contained in~$A^+$,
	so $B^+ \subseteq A^+$.
	We have a decomposition
	\begin{align}
		z^k B &= z^{k+1} T^{-1} \diag(I,0) T A \oplus z^{k} T^{-1} \diag(0,I) T A & \text{for any $k \in \ZZ$},\\
		B^- &= T^{-1} \diag(I,0) T A \oplus A^-, \nonumber
	\end{align}
	so, in particular, $\bd_0(B,A) = T^{-1} \diag(I,0) T A$.
	It follows that $[\eta|_{\bd_0(B,A)}]_0 \sim \xi$ is the boundary form,
	completing the proof.
	The even hermitian case is completely analogous.
\end{proof}

\subsection{Kernel of the descent maps from unitary groups}

Let $\epsilon : \Rz \to R$ be the ring homomorphism 
such that $z \mapsto 1$ and $\Rz \supset R \ni r \mapsto r \in R$.
Naturally, we have $\epsilon : \qhp^\mp(\Rz) \to \qhp^\mp(R)$ by replacing every $z$ with $1$.
The map $\epsilon$ has an right inverse $\bar\epsilon$ that embeds $R$ into $\Rz$.
The right inverse $\bar\epsilon$ makes any $R$-module an $\Rz$-module
and induces
$\bar\epsilon: \qhp^\mp(R) \ni U \mapsto U \in \qhp^\mp(\Rz)$ and
$\bar\epsilon: \qumod^\mp(R) \ni [U] \mapsto [U] \in \qumod^\mp(\Rz)$.

\begin{lemma}\label{thm:piMinMmeansR}
	Let $B$ be a modular $R$-base of a finitely generated $\Rz$-module~$M = \bigoplus_j z^j B$.
	Suppose $\projz^{\ge 0}_{B} M \subseteq M$.
	Then $M = \bigoplus_{j=-\infty}^\infty z^j M_0 = \bar\epsilon M_0$
	where $M_0 = \projz^{0}_{B} M$.
\end{lemma}

That is, if $B = R^m$ (the standard modular $R$-base of $\Rz^m$), 
then $M$ is generated over~$\Rz$ by elements that do not involve~$z$.

\begin{proof}
	We drop subscript~$B$ in this proof.
	Note that $z^k \projz^{\ge 0} = \projz^{\ge k} z^k$ as $R$-linear maps.
	So, $\projz^{\ge k} M = \projz^{\ge k} z^k M = z^k \projz^{\ge 0} M \subseteq z^k M = M$ for any $k \in \ZZ$.
	Also, $\projz^{< k} M = (1 - \projz^{\ge k})M \subseteq M$.
	Hence, $\projz^k M = \projz^{<k+1} \projz^{\ge k} M \subseteq \projz^{<k+1} M \subseteq M$.
	Since $1 = \sum_k \projz^k$, we conclude that $M = \sum_k \projz^k M = \sum_k z^k \projz^0 M$.
\end{proof}

Our goal here is to prove

\begin{proposition}[\S 2 of \cite{Ranicki2}]\label{thm:oddBassExact}
	The following short sequence is split exact.
	\begin{align}
		0 \xrightarrow{\quad} \qumod^\mp(R) \xrightarrow{\quad \bar\epsilon \quad} \qumod^\mp(\Rz) 
		\xrightarrow{\quad \bass^\downarrow_{2\iota+1} \quad} \qwitt^\mp(R) \xrightarrow{\quad} 0
	\end{align}
\end{proposition}

In~\cite{clifqca1} we have shown that, for $p$ odd,
the kernel of~$\bass^\downarrow_3$ ``blends'' into shift QCA
($p=2$ is treated there with $\lambda$-unitary groups).
The blending will be discussed at the end of this paper.
This proposition is a stronger result.

\begin{proof}
	The map $\bar\epsilon$ is injective since if $\bar\epsilon(T) = S$ for some $S \in \qehp^\mp(\Rz)$,
	then $T = \epsilon \bar\epsilon T = \epsilon S \in \qehp^\mp(R)$.
	The map~$\bass^\downarrow$ is surjective by~\ref{thm:fromForm-up-and-down-is-identity}.
	If a unitary~$U$ does not involve~$z$, 
	then the boundary module in the construction of $\bass^\downarrow$
	can be taken zero, and the resulting form is trivial.
	So, the composition of the two maps is zero.
	It remains to prove that $\ker \bass^\downarrow \subseteq \im \bar \epsilon$.

	Let $U \in \qhp(q;\Rz)$ be such that $\bass^\downarrow_{2\iota+1}([U]) = 0$.
	We will construct an elementary unitary $H \in \qehp^\mp(\Rz)$ 
	such that $V = H (U \hat\oplus I)$ (where the dimension of~$I$ will be set as we proceed)
	does not have any $z$ variable,
	{\em i.e.}, in fact, $V \in \qhp^\mp(R) \subseteq \qhp^\mp(\Rz)$.

	Let $B_0 = R^q \oplus R^{q*}$ be a modular $R$-base of~$Q = \Rz^{q} \oplus \Rz^{q*}$;
	$Q$ is given a trivial quadratic $\mp$-form~$\eta_q$.
	The unitary~$U$ gives another modular $R$-base~$U B_0$ of~$Q$.
	The boundary form~$\bass^\downarrow(U)$ is constructed by 
	choosing a sufficiently large integer $n \ge 0$ such that 
	\begin{align}
		C &= \bd_n(U B_0, B_0) = (\projz^{<n}_{U B_0} Q) \cap (\projz^{\ge 0}_{B_0} Q)\\
		Q &= (\projz^{~\ge n}_{U B_0}) \oplus C \oplus (\projz^{< 0}_{B_0})\nonumber
	\end{align}
	is defined.
	Let $P_0$ be any $R$-module isomorphic to~$C$ with an isomorphism~$\alpha_0: P_0 \to C$.
	A boundary form~$\phi = \bass^\downarrow(U)$ 
	is a quadratic $\mp$-form on $P_0$ inherited from~$\eta_q$ on~$Q$:
	\begin{align}
		\phi(a, b) = [ \eta_q(\alpha_0(a), \alpha_0(b)) ]_0.
	\end{align}
	Since~$\phi$ is trivial by the assumption that $\bass^\downarrow_{2\iota+1}([U])=0$,
	the module~$P_0$ decomposes as $P_0 = L_0 \oplus L_0^*$
	where $L_0, L_0^*$ are lagrangians.

	We introduce an $\Rz$-module 
	\begin{align}
		P = \bigoplus_{j = -\infty}^{\infty} z^j P_0
	\end{align}
	which has a modular $R$-base $P_0$.
	Naturally, $L_0$ and $L_0^*$ generate $L = \bigoplus_j z^j L_0$ and $L^* = \bigoplus_j z^j L_0^*$,
	respectively.
	We stabilize $Q$ with $P \oplus P^* = L \oplus L^* \oplus L^* \oplus L$ 
	where $P \oplus P^*$ has its own, completely independent from~$\phi$,
	trivial quadratic $\mp$-form
	\begin{align}
		\eta_P =
		\begin{Bmatrix}
		P^* \\ P
		\end{Bmatrix}
		\begin{pmatrix}
		0 & I \\ 0 & 0
		\end{pmatrix}
		\begin{Bmatrix}
		P \\ P^*
		\end{Bmatrix}.
	\end{align}

	Now we construct an elementary $\eta$-unitary~$H$,
	which will involve a number of maps as follows.
	Since~$P_0\xrightarrow{\alpha_0} C$ are isomorphic as $R$-modules,
	we have a natural $\Rz$-linear map~$\alpha : P \to Q$ 
	that extends~$\alpha_0$ by~$\Rz$-linearity.
	\begin{align}
		\alpha &=
		\begin{Bmatrix}
		\projz^{~\ge n}_{U B_0}\\
		C\\
		\projz^{<0}_{B_0}
		\end{Bmatrix}
		\begin{pmatrix}
		\star &&& 0 &&& 0 \\
		\star &&& \alpha_0 &&& \star\\
		0 &&& 0 &&& \star
		\end{pmatrix}
		\begin{Bmatrix}
		\projz^{>0}_{P_0}\\
		P_0\\
		\projz^{<0}_{P_0}
		\end{Bmatrix}
		:
		P \to Q \label{eq:alpha}
	\end{align}
	While we do not know in general the blocks denoted by~$\star$,
	we have determined four zero blocks from the $\Rz$-linearity.
	In the opposite direction, we define
	\begin{align}
		\beta &=
		\begin{Bmatrix}
		\projz^{>0}_{P_0}\\
		P_0\\
		\projz^{<0}_{P_0}
		\end{Bmatrix}
		\begin{pmatrix}
		0 & 0 & 0 \\
		0 & \alpha_0^{-1} & 0\\
		0 & 0 & 0
		\end{pmatrix}
		\begin{Bmatrix}
		\projz^{~\ge n}_{U B_0}\\
		C\\
		\projz^{<0}_{B_0}
		\end{Bmatrix}
		:
		Q \to P\, .
	\end{align}
	In addition, we define various ``$z$-shifting'' maps.
	\begin{align}
		\mu &=
		\begin{Bmatrix}
		L\\
		L^*
		\end{Bmatrix}
		\begin{pmatrix}
		0 & \pm z\\
		\tfrac 1 z & 0
		\end{pmatrix}
		\begin{Bmatrix}
		L^*\\
		L
		\end{Bmatrix} : P^* \to P\, ,
		&
		\zeta &=
		\begin{Bmatrix}
		L\\
		L^*
		\end{Bmatrix}
		\begin{pmatrix}
		I & 0\\
		0 & z
		\end{pmatrix}
		\begin{Bmatrix}
		L\\
		L^*
		\end{Bmatrix} : P \to P\, ,
		\nonumber\\
		\zeta^{-1} &=
		\begin{Bmatrix}
		L\\
		L^*
		\end{Bmatrix}
		\begin{pmatrix}
		I & 0\\
		0 & \tfrac 1 z
		\end{pmatrix}
		\begin{Bmatrix}
		L\\
		L^*
		\end{Bmatrix} : P \to P\, ,
		&
		\zeta^\dag &=
		\begin{Bmatrix}
		L^*\\
		L
		\end{Bmatrix}
		\begin{pmatrix}
		I & 0\\
		0 & \tfrac 1 z
		\end{pmatrix}
		\begin{Bmatrix}
		L^*\\
		L
		\end{Bmatrix} : P^* \to P^*\, ,
		\\
		\mu\zeta^\dag &=
		\begin{Bmatrix}
		L\\
		L^*
		\end{Bmatrix}
		\begin{pmatrix}
		0 & \pm I\\
		\tfrac 1 z & 0
		\end{pmatrix}
		\begin{Bmatrix}
		L^*\\
		L
		\end{Bmatrix}: P^* \to P \, .\nonumber
		\label{eq:eta-zeta}
	\end{align}
	Also, we define $[\eta]_+$, $[\eta]_-$, and $[\eta]_0$, all from $P$ to $P^*$, by
	\begin{align}
		&\alpha^\dag \eta \alpha
		=
		[\eta]_+ + [\eta]_- + \underbrace{[\eta]_0}_{=\bar\epsilon \phi},\\
		&[\eta]_+ P_0 \subseteq \projz^{>0}_{P_0^*} P^*,
		\qquad
		[\eta]_- P_0 \subseteq \projz^{<0}_{P_0^*} P^*,
		\qquad
		[\eta]_0 P_0 \subseteq P_0^*.\nonumber
	\end{align}
	Finally,
	\begin{align}
		\theta &= [\eta]_+ \mp [\eta]_-^\dag + \psi,\\
		\psi &=
		\begin{Bmatrix}
		L^*\\
		L
		\end{Bmatrix}
		\begin{pmatrix}
		0 & I \\
		0 & 0
		\end{pmatrix}
		\begin{Bmatrix}
		L\\
		L^*
		\end{Bmatrix} : P \to P^*,\nonumber
	\end{align}
	where $\psi$ is not necessarily equal but equivalent as a quadratic $\mp$-form to $[\eta]_0$,
	{\it i.e.}, $\psi - [\eta]_0$ is an even hermitian $\pm$-form.
	Hence, $\alpha^\dag \eta \alpha - \theta = [\eta]_- \pm [\eta]_-^\dag + [\eta]_0 - \psi$
	is an even hermitian $\pm$-form.

	Consider
	\begin{align}
		H=
		\begin{Bmatrix}
		Q\\
		P\\
		P^*
		\end{Bmatrix}
		\begin{pmatrix}
		I & 0 & 0\\
		0 & I & \mu\\
		0 & 0 & I
		\end{pmatrix}
		\begin{Bmatrix}
		Q\\
		P\\
		P^*
		\end{Bmatrix}
		\underbrace{
		\begin{pmatrix}
		I &&& -\alpha \zeta &&& 0 \\
		0 &&& I &&& 0\\
		\zeta^\dag \alpha^\dag(\eta \mp \eta^\dag) &&& -\zeta^\dag \theta \zeta &&& I
		\end{pmatrix}
		}_{J}
		\begin{Bmatrix}
		Q\\
		P\\
		P^*
		\end{Bmatrix}.
		\end{align}
		Let us show that this (rather complicated) transformation 
		is a stably elementary $\eta$-unitary on $Q \oplus (P \oplus P^*)$.
		Indeed, the first matrix involving~$\mu$ is nothing but~$\cz(\mu^\dag)^\dag$ 
		with $\mu$ even $\pm$-hermitian.
		For the second matrix~$J$, we see
		\begin{align}
		&\underbrace{
		\begin{pmatrix}
		I & 0 & \mp(\eta \mp \eta^\dag)\alpha\zeta \\
		-\zeta^\dag \alpha^\dag & I & -\zeta^\dag \theta^\dag \zeta\\
		0 & 0 & I
		\end{pmatrix}
		}_{J^\dag}
		\begin{pmatrix}
		\eta & 0 & 0\\
		0 & 0 & I \\
		0 & 0 & 0
		\end{pmatrix}
		\underbrace{
		\begin{pmatrix}
		I & -\alpha \zeta & 0 \\
		0 & I & 0\\
		\zeta^\dag \alpha^\dag(\eta \mp \eta^\dag) & -\zeta^\dag \theta \zeta & I
		\end{pmatrix}
		}_{J}\\
		&\qquad =
		\begin{pmatrix}
		\eta & -\eta\alpha\zeta & 0 \\
		\mp \zeta^\dag \alpha^\dag \eta^\dag & \zeta^\dag(\alpha^\dag \eta \alpha - \theta) \zeta & I\\
		0 & 0 & 0
		\end{pmatrix}\nonumber
	\end{align}
	which is equivalent as a quadratic $\mp$-form to~$\eta \oplus \eta_P$ at the middle,
	proving that~$H$ is an~$\eta$-unitary.
	The matrix~$J$ is stably elementary 
	because it preserves a lagrangian $\{(q,0,t,q) \in Q \oplus P \oplus P^* \oplus Q ~|~ q \in Q, t \in P^*\}$
	of $Q \oplus (P \oplus P^*) \oplus Q$ 
	that is equipped with a trivial quadratic form~$\eta \oplus \eta_P \oplus (-\eta)$;
	it is the sum of two lagrangians $\{(q,0,0,q)\}$ and $\{(0,0,t,0)\}$.
	See~\ref{thm:formation-unitary}.

	Now, we claim that 
	\begin{align}
		\projz^{\ge 0}_{B_0 \oplus P_0 \oplus P_0^*} H(U A \oplus P \oplus 0) \subseteq H(U A \oplus P \oplus 0),
		\quad\text{ where } A = \Rz^q \oplus 0 \subset Q
		\label{eq:pospre}
	\end{align}
	is a trivial lagrangian.
	By~\ref{thm:piMinMmeansR}, we see that
	with respect to an ordered basis of $Q \oplus P \oplus P^*$ 
	such that the basis elements of $A \oplus P \oplus 0$ come first,
	the matrix $H(U \hat\oplus I)$ has the left half of the block not involving~$z$.
	By~\ref{thm:CZ} (see also~\ref{thm:formation-unitary}), 
	$H(U\oplus I)$ is equivalent to a unitary that does not involve~$z$ up to elementary ones,
	so $[H(U \hat\oplus I)] \in \bar\epsilon(\qumod^\mp(R))$,
	completing the proof.

	The claim~\eqref{eq:pospre} is proved by direct calculation.
	Below we will use a projection $\projz^{~\ge 0}_{B_0 \oplus (P_0 \oplus P_0^*)}$.
	It will be convenient to spell out components
	\begin{align}
		H(s \in U A \subset Q) 
		&=
		\begin{pmatrix}
		s \\
		\mu \zeta^\dag \alpha^\dag(\eta \mp \eta^\dag)s\\
		\zeta^\dag \alpha^\dag(\eta \mp \eta^\dag)s
		\end{pmatrix},
		&
		H(t \in P)
		&=
		\begin{pmatrix}
		-\alpha \zeta t\\
		t - \mu\zeta^\dag \theta \zeta t\\
		-\zeta^\dag \theta \zeta t
		\end{pmatrix}.
		\end{align}
		We consider these components sector by sector according to
		\begin{align}
		U A \subset Q  &= \projz^{\ge n}_{U B_0} \oplus C \oplus \projz^{<0}_{B_0} ,\\
		P &= \projz^{\ge 0}_{P_0} \oplus \projz^{<0}_{P_0}.\nonumber
	\end{align}

	\noindent
	(i) Suppose $s \in \projz^{~\ge n}_{U B_0} \subset Q$.
	Obviously, $\projz^{\ge 0} s = s$.
	Since $\eta \mp \eta^\dag$ has $z$-degree zero and preserves~$\projz^{~\ge n}_{U B_0}$,
	we have $(\eta \mp \eta^\dag)s \in \projz^{~\ge n}_{U B_0}$,
	Then, $\alpha^\dag$ maps $(\eta \mp \eta^\dag)s$ into $\projz^{>0}_{P_0}$ as seen by~\eqref{eq:alpha}.
	Since both $\mu \zeta^\dag$ and $\zeta^\dag$ may decrease the $z$-degree by at most one,
	we conclude that
	\begin{align}
		\projz^{\ge 0} H(s) = H(s) \quad\text{ if }\quad s \in \projz^{~\ge n}_{U B_0}.
	\end{align}
	(ii) Suppose $s \in C = \projz^{< n}_{U B_0} \cap \projz^{\ge 0}_{B_0} \subset Q$.
	We have $s = \alpha \beta s$.
	We claim that
	\begin{align}
		\projz^{\ge 0} H(s) = H(-\zeta^{-1}\beta s) 
		=\begin{pmatrix}
		\alpha \beta s\\
		-\zeta^{-1}\beta s + \mu\zeta^\dag \theta \beta s\\
		\zeta^\dag \theta \beta s
		\end{pmatrix}.
	\end{align}
	The first components on both sides read $\projz^{\ge 0} s = s = \alpha \beta s$.
	For the second and third components,
	we observe 
	$\projz^{\ge 0} \alpha^\dag (\eta \mp \eta^\dag) s 
	= \projz^{\ge 0} \alpha^\dag (\eta \mp \eta^\dag) \alpha \beta s
	= \projz^{\ge 0} (\theta \mp \theta^\dag) \beta s$.
	Note that $\projz^{\ge 0} ([\eta]_+^\dag \mp [\eta]_-) \beta s = 0$
	because $\beta s$ has $z$-degree zero
	and is mapped into 
	the strictly negative sector by $[\eta]_+^\dag \mp [\eta]_-$.
	Hence, $\projz^{\ge 0} \alpha^\dag (\eta \mp \eta^\dag) s 
	= 
	\projz^{\ge 0} (\theta \mp \theta^\dag) \beta s 
	=
	\projz^{\ge 0} (\theta \mp \psi^\dag) \beta s$.
	Now, putting $\beta s = \ell + \ell^* \in L_0 \oplus L_0^* = P_0 \subset P$,
	we have
	\begin{align}
		\projz^{\ge 0} 
		\begin{pmatrix} 
		\mu \zeta^\dag \alpha^\dag(\eta \mp \eta^\dag)s\\
		\zeta^\dag \alpha^\dag(\eta \mp \eta^\dag)s
		\end{pmatrix}
		&=
		\projz^{\ge 0}
		\begin{pmatrix} 
		\mu \zeta^\dag \\
		\zeta^\dag 
		\end{pmatrix} \projz^{\ge 0}\alpha^\dag(\eta \mp \eta^\dag)s
		\qquad\text{(since $\mu\zeta^\dag$, $\zeta^\dag$ do not increase $z$-degrees)}\nonumber\\
		&=
		\projz^{\ge 0}
		\begin{pmatrix} 
		\mu \zeta^\dag \\
		\zeta^\dag 
		\end{pmatrix}
		([\eta]_+ \beta s \mp [\eta]_-^\dag \beta s + \ell^* \mp \ell )
		\qquad(\ell^* = \psi\beta s, ~\ell = \psi^\dag \beta s)\nonumber\\
		&=
		\projz^{\ge 0}
		\begin{pmatrix} 
		\mu \zeta^\dag ([\eta]_+ \mp [\eta]_-^\dag) \beta s + \tfrac 1 z \ell^* - \ell \\
		\zeta^\dag ([\eta]_+ \mp [\eta]_-^\dag) \beta s + \ell^* \mp \tfrac 1 z \ell
		\end{pmatrix}\\
		&=
		\begin{pmatrix} 
		\mu \zeta^\dag ([\eta]_+ \mp [\eta]_-^\dag) \beta s - \ell \\
		\zeta^\dag ([\eta]_+ \mp [\eta]_-^\dag) \beta s + \ell^* 
		\end{pmatrix}\nonumber\\
		&=
		\begin{pmatrix} 
		\mu \zeta^\dag ([\eta]_+ \mp [\eta]_-^\dag + \psi) \beta s\\
		\zeta^\dag ([\eta]_+ \mp [\eta]_-^\dag + \psi) \beta s
		\end{pmatrix}
		-
		\begin{pmatrix}
		\ell + \tfrac 1 z \ell^*\\
		0
		\end{pmatrix}\nonumber\\
		&=\begin{pmatrix}
		\mu\zeta^\dag\theta\beta s - \zeta^{-1}\beta s\\
		\zeta^\dag \theta \beta s
		\end{pmatrix}.\nonumber
	\end{align}
	(iii) Suppose $s \in \projz^{< 0}_{B_0} \subset Q$.
	The map $\eta \mp \eta^\dag$ preserves $\projz^{<0}_{B_0}$
	that is taken by $\alpha^\dag$ to $\projz^{<0}_{P_0} \subset P$.
	Then, both $\mu\zeta^\dag$ and $\zeta^\dag$ can only decrease the $z$-degree.
	Therefore,
	\begin{align}
		\projz^{\ge 0} H(s) = 0 \text{ if } s \in \projz^{< 0}_{B_0}.
	\end{align}
	(iv) Suppose $t \in \projz^{\ge 0}_{P_0} \subset P$.
	The map $\zeta$ can only increase the $z$-degree
	and $\alpha$ maps the nonnegative sector into the nonnegative sector.
	In addition,
	\begin{align}
		\mu \zeta^\dag \psi \zeta &= 
		\begin{Bmatrix}
		L \\ L^*
		\end{Bmatrix}
		\begin{pmatrix}
		0 & 0 \\ 0 & I
		\end{pmatrix}
		\begin{Bmatrix}
		L \\ L^*
		\end{Bmatrix},
		&
		\zeta^\dag \psi \zeta &= 
		\begin{Bmatrix}
		L^* \\ L
		\end{Bmatrix}
		\begin{pmatrix}
		0 & z \\ 0 & 0
		\end{pmatrix}
		\begin{Bmatrix}
		L \\ L^*
		\end{Bmatrix}
	\end{align}
	are $z$-degree nondecreasing.
	Therefore,
	\begin{align}
		\projz^{\ge 0} H(t) = H(t) \text{ if } t \in \projz^{\ge 0}_{P_0}.
	\end{align}
	(v) Suppose $t \in \projz^{<0}_{P_0} \subset P$.
	We claim that
	\begin{align}
		\projz^{\ge 0} H(t) = 
		- \begin{pmatrix}
		\projz^{\ge 0} \alpha \zeta t\\
		\projz^{\ge 0} \mu\zeta^\dag\theta\zeta t\\
		\projz^{\ge 0} \zeta^\dag\theta\zeta t
		\end{pmatrix}
		=
		-\begin{pmatrix}
		\alpha \beta \alpha \zeta t\\
		\mu \zeta^\dag\theta\beta\alpha\zeta t - \zeta^{-1} \beta \alpha \zeta t\\
		\zeta^\dag \theta \beta \alpha \zeta t
		\end{pmatrix}
		=H( \zeta^{-1}\beta \alpha \zeta t ).\label{eq:tneg}
	\end{align}
	Since $\zeta t \in \projz^{\le 0}$, 
	we must have $\alpha \zeta t = c + b \in C \oplus \projz^{<0} \subset Q$.
	Projecting, we see $\projz^{\ge 0} \alpha \zeta t = c$,
	but $c = \alpha \beta (c + b) = \alpha \beta \alpha \zeta t$.
	This confirms the claim for the first component.
	To examine the second and third components,
	we observe that 
	\begin{align}
		0 \to \ker\left( \beta \alpha |_{\projz^{\le 0}_{P_0}} \right) 
		\xrightarrow{\qquad}
		\projz^{\le 0}_{P_0}
		\xrightarrow{\quad \beta\alpha \quad}
		P_0
		\to
		0
	\end{align}
	is split exact since $\beta \alpha|_{P_0} = I$.
	So, it suffices to assume that either~$\zeta t \in P_0$ or~$\zeta t \in \ker (\beta\alpha)$.
	(v.i)~Suppose $p = \zeta t = \beta \alpha \zeta t \in P_0 \subset P$.
	Since $t$ is strictly negative in $z$-degree,
	for $p=\zeta t$ to have $z$-degree zero we must have $p \in L^* \subset P_0 \subset P$.
	\begin{align}
		\projz^{\ge 0}
		\begin{pmatrix}
		\mu\zeta^\dag\theta p\\
		\zeta^\dag \theta p
		\end{pmatrix}
		&=
		\projz^{\ge 0}
		\begin{pmatrix}
		\mu\zeta^\dag([\eta]_+ \mp [\eta]_-^\dag)p + \underbrace{\mu \zeta^\dag \psi p}_{=\tfrac 1 z p = \zeta^{-1} p}\\
		\zeta^\dag ([\eta]_+ \mp [\eta]_-^\dag)p + \underbrace{\zeta^\dag \psi p}_{= p}
		\end{pmatrix}\\ \nonumber
		&=
		\begin{pmatrix}
		\mu\zeta^\dag([\eta]_+ \mp [\eta]_-^\dag)p\\
		\zeta^\dag ([\eta]_+ \mp [\eta]_-^\dag +\psi )p
		\end{pmatrix}
		=
		\begin{pmatrix}
		\mu\zeta^\dag\theta p - \zeta^{-1} p\\
		\zeta^\dag \theta p
		\end{pmatrix}
	\end{align}
	which is the claim.
	(v.ii)~Suppose $\beta \alpha \zeta t = 0$,
	so $\alpha \zeta t \in \projz^{<0}_{B_0}$.
	Since $\mu \zeta^\dag$ and $\zeta^\dag$ is $z$-degree nonincreasing,
	we have $\projz^{\ge 0} \mu\zeta^\dag = \projz^{\ge 0} \mu\zeta^\dag \projz^{\ge 0}$ and
	$\projz^{\ge 0} \zeta^\dag = \projz^{\ge 0}\zeta^\dag \projz^{\ge 0}$.
	So, we examine $\projz^{\ge 0} \theta \zeta t$.
	Since~$\psi^\dag \zeta$ preserves $z$-degree,
	we see that $\projz^{\ge 0} \psi^\dag \zeta t = 0$
	and $\projz^{\ge 0} \theta^\dagger \zeta t = 0$.
	Hence, $\projz^{\ge 0} \theta \zeta t = \projz^{\ge 0} (\theta \mp \theta^\dag) \zeta t =
	\projz^{\ge 0} \alpha^\dag(\eta \mp \eta^\dag) \alpha \zeta t$.
	Since $\eta \mp \eta^\dag$ preserves $\projz^{<0}_{B_0}$,
	we see $(\eta \mp \eta^\dag) \alpha \zeta t \in \projz^{<0}$,
	which is mapped by $\alpha^\dag$ to $\projz^{<0}_{P_0} \subset P$ as seen by~\eqref{eq:alpha}.
	Therefore, $\projz^{\ge 0} \theta \zeta t = 0$.
	This completes the proof of~\eqref{eq:tneg} and hence~\eqref{eq:pospre}.
\end{proof}

\subsection{Going down from forms to unitaries}

Let $\Delta$ be a nonsingular even hermitian $\mp$-form over~$\Rz$ of dimension~$t$.
Let $B$ be a modular $R$-base of~$\Rz^t$.
In this subsection, all $z$-degree projections are with respect to $B\oplus B^*$ 
unless specified otherwise.

Let $n\ge 0$ be so large an integer that
\begin{align}
	\Delta B &\subseteq \bigoplus_{k=-n}^n z^k B^*, &
	\Delta^{-1} B^* &\subseteq \bigoplus_{k=-n}^n z^k B.
\end{align}
We thus obtain finitely generated free $R$-modules (see~\ref{sec:modularBase}),
\begin{align}
	\bd_n(\Delta^{-1}B^*, B) &= z^n \Delta^{-1} (B^*)^- \cap B^+ \subseteq \bigoplus_{k=0}^{2n-1} z^k B, \label{eq:bdnABBA}\\
	\bd_n(B, \Delta^{-1}B^*) &= z^n B^- \cap \Delta^{-1} (B^*)^+ \subseteq \bigoplus_{k=-n}^{n-1} z^k B. \nonumber
\end{align}
Following~\cite[\S3]{Ranicki2}, we define~$L(\Delta,B,n)$ and~$L^*(\Delta,B,n)$ by
\begin{align}
\left( \underbrace{\bigoplus_{k=0}^{n-1} z^k B}_M \right)\oplus \left( \underbrace{\bigoplus_{k=0}^{n-1} z^k B^*}_{M^*} \right)
\supseteq
\begin{cases}
	L(\Delta,B,n) &= \left\{ \begin{pmatrix}
	 \projz^{<n} u\\
	 \projz^{\ge 0}\Delta u
	\end{pmatrix}
	\middle|~ u \in \bd_n(\Delta^{-1}B^*,B)\right\},
	\\
	 L^*(\Delta,B,n) &= 
	\left\{ \begin{pmatrix}
	 -\projz^{\ge 0} g \\
	 \projz^{\ge 0} \Delta \projz^{<0} g
	\end{pmatrix} \middle|~ 
	g \in \bd_n(B,\Delta^{-1}B^*) \right\}.
\end{cases} 	 \label{eq:Ldef}
\end{align}
Due to the $z$-exponent bound in~\eqref{eq:bdnABBA}, the containment of~\eqref{eq:Ldef} holds.
We are going to show that $L(\Delta,B,n)$ is a lagrangian of some trivial quadratic $\pm$-form,
which defines a unitary for that form.
This is the descent map we are going to define 
from even hermitian $\mp$-forms down to $\eta^\pm$-unitaries.

\begin{lemma}\label{thm:Lcomplement}
	There are $R$-module isomorphisms 
	$L(\Delta,B,n) \cong \bd_n(\Delta^{-1}B^*,B)$ and $L^*(\Delta,B,n) \cong \bd_n(B,\Delta^{-1}B^*)$.
	In addition, $L(\Delta,B,n) \oplus L^*(\Delta,B,n) = M \oplus M^*$.
\end{lemma}

This appears in the proof of~\cite[Lem. 3.1]{Ranicki2}.

\begin{proof}
	Define a map of $R$-modules
	\begin{align}
		F&=
		\begin{Bmatrix}
		M \\ M^*
		\end{Bmatrix}
		\begin{pmatrix}
		\projz^{<n} &&& - \projz^{\ge 0} \\ 
		\projz^{\ge 0} \Delta &&& \projz^{\ge 0} \Delta \projz^{<0}
		\end{pmatrix} 
		\begin{Bmatrix}
		\bd_n(\Delta^{-1}B^*,B) \\ \bd_n(B,\Delta^{-1}B^*) 
		\end{Bmatrix}.
		\label{eq:pixi}
	\end{align}
	Here, the domain of~$F$ is the (external) direct sum of two $R$-modules displayed in the right curly braces.
	The image of~$F$ is precisely $L(\Delta,n) + L^*(\Delta,n)$,
	where $L(\Delta,n)$ is the image of~$\bd_n(\Delta^{-1}B^*,B)$
	and $L^*(\Delta,n)$ is the image of~$\bd_n(B,\Delta^{-1}B^*)$.
	We will show that $F$ is an $R$-module isomorphism,
	which will in turn shows the claimed isomorphisms.
	Then, since the domain of~$F$ is a direct sum of two modules,
	their images cannot intersect along a nonzero submodule,
	which implies that the sum~$L(\Delta,n) + L^*(\Delta,n)$ is direct.
	This will complete the proof.

	To show that $F$ is injective,
	let $u = x + m \in \bd_n(\Delta^{-1}B^*,B)$ 
	with~$x = \projz^{\ge n} u \in z^n B^+$ and~$m = \projz^{< n} u \in M$,
	and $g = m' + y \in \bd_n(B,\Delta^{-1}B^*)$ 
	with~$m' = \projz^{\ge 0} g \in M$ and~$y = \projz^{< 0} g \in B^-$.
	If~\mbox{$F( u, g ) = 0$},
	then $m=m'$ by the first row of~$F$ and $\projz^{\ge 0} \Delta( x + m + y ) = 0$,
	implying $\Delta(x+m+y) \in (B^*)^-$.
	But $\Delta x \in z^n \Delta B^+ \subseteq (B^*)^+$ and $\Delta(m+y) = \Delta g \in (B^*)^+$,
	so $\Delta(x+m+y) \in (B^*)^+ \cap (B^*)^-$, or~$x+m+y = 0$.
	Since $x,m,y$ have different $z$-exponents, we must have $x = m = y = 0$.

	To show that $F$ is surjective, we are going to show that 
	\begin{align}
		G&=
		\begin{Bmatrix}
		\bigoplus_k z^k B \\
		\bigoplus_k z^k B
		\end{Bmatrix}
		\begin{pmatrix}
		\projz^{\ge 0} \Delta^{-1} \projz^{<0} \Delta &&& \projz^{\ge 0} \Delta^{-1} \\
		- \Delta^{-1} \projz^{\ge 0} \Delta &&&  \projz^{<n} \Delta^{-1}
		\end{pmatrix}
		\begin{Bmatrix}
		M \\ M^*
		\end{Bmatrix}\label{eq:matrixG}
		\end{align}
		is a right inverse of~$F$.
		The image of~$G$ lies in $\bd_n(\Delta^{-1}B^*,B) \oplus \bd_n(B,\Delta^{-1}B^*)$:
		for any~$m\in M$ and $m' \in M^*$,
		\begin{align}
		z^n \Delta^{-1} (B^*)^- \ni 
		\Delta^{-1}( \Delta I \Delta^{-1}\projz^{<0} \Delta m - \Delta \projz^{<0} \Delta^{-1} \projz^{<0} \Delta m ) 
		&= 
		\underline{ \projz^{\ge 0}\Delta^{-1} \projz^{<0} \Delta} m 
		\in B^+ \nonumber\\
		z^n \Delta^{-1} (B^*)^- \ni
		\Delta^{-1}(\Delta I \Delta^{-1}m' - \Delta \projz^{<0} \Delta^{-1}m')
		&=
		\underline{\projz^{\ge 0} \Delta^{-1}} m' 
		\in B^+\\
		z^n B^- \ni
		\Delta^{-1} I \Delta m - \Delta^{-1} \projz^{<0} \Delta m  
		&= 
		\underline{\Delta^{-1} \projz^{\ge 0} \Delta} m 
		\in \Delta^{-1} (B^*)^+\nonumber\\
		\Delta^{-1} (B^*)^+
		\ni
		\Delta^{-1}(\Delta I \Delta^{-1} m' - \Delta \projz^{\ge n} \Delta^{-1} m')
		&=
		\underline{\projz^{<n} \Delta^{-1}} m'  
		\in z^n B^- \nonumber
	\end{align}
	where the underlined operators appear verbatim in~\eqref{eq:matrixG}.
	The following calculation shows that $G$ is a right inverse of~$F$.
	Let $\nu = \projz^{\ge 0}$.
	\begin{align}
		&F\circ G \nonumber \\
		&= 
		\begin{pmatrix}
		\projz^{\ge 0}\projz^{<n} \Delta^{-1} \projz^{< 0} \Delta + \projz^{\ge 0} \Delta^{-1} \projz^{\ge 0} \Delta &&&
		\projz^{\ge 0}\projz^{<n} \Delta^{-1} - \projz^{\ge 0}\projz^{<n} \Delta^{-1} \\
		\projz^{\ge 0} \Delta \projz^{\ge 0} \Delta^{-1} \projz^{<0} \Delta - \projz^{\ge 0} \Delta \projz^{<0} \Delta^{-1} \projz^{\ge 0} \Delta &&&
		\projz^{\ge 0} \Delta \projz^{\ge 0} \Delta^{-1} + \projz^{\ge 0} \Delta \projz^{<0}\Delta^{-1}
		\end{pmatrix}\nonumber\\
		&=
		\begin{pmatrix}
		\nu \Delta^{-1} (1-\nu) \Delta + \nu \Delta^{-1} \nu \Delta &&&
		0 \\
		\nu \Delta \nu \Delta^{-1} (1-\nu) \Delta - \nu \Delta (1-\nu) \Delta^{-1} \nu \Delta &&&
		\nu \Delta \nu \Delta^{-1} + \nu \Delta (1-\nu)\Delta^{-1}
		\end{pmatrix} \\
		&=\begin{pmatrix}
		\nu &&& 0 \\ \nu \Delta \nu - \nu \Delta &&& \nu
		\end{pmatrix}\nonumber
	\end{align}
	where in the top left entry, we used~$\projz^{<n}\Delta^{-1}\projz^{<0} = \Delta^{-1}\projz^{<0}$.
	Acting on~$M \oplus M^*$, the projection~$\nu$ does nothing.
	So, the diagonal terms act by identities, and the offdiagonal term is zero.
	Therefore, $F$ is an isomorphism.
\end{proof}

Let $\eta_t$ be a trivial quadratic $\pm$-form 
on~$\bigoplus_j z^j (B \oplus B^*)$.
Note the sign here; $\Delta$ is a $\mp$-form.
By~\ref{thm:boundaryXi-nonsingular},
$M \oplus M^*$ is equipped with a nonsingular quadratic $\pm$-form~$[\eta_t|_{M\oplus M^*}]_0$,
which is trivial.

\begin{lemma}\label{thm:etaOnLLstar}
	With respect to the nonsingular quadratic $\pm$-form~$[\eta_t|_{M\oplus M^*}]_0$ on~$M \oplus M^*$, 
	$R$-modules~$L(\Delta,B,n)$ and $L^*(\Delta,B,n)$ are a nonsingular pair of lagrangians.
	Therefore, there exists $U \in \qhp^\pm(tn;R)$ such that $U(M \oplus 0) = L(\Delta,B,n)$.
\end{lemma}

\begin{proof}
	By~\ref{thm:Lcomplement}, the $R$-submodule $L = L(\Delta,B,n)$ is a direct summand of~$M \oplus M^*$
	with a direct complement $L' = L^*(\Delta,B,n)$.
	We write a general element of~$L$ as $L(u) = (\projz^{< n} u, \projz^{\ge 0}\Delta u)$
	where~$u \in \bd_n(\Delta^{-1}B^*,B)$,
	and likewise a general element of $L'$ 
	as $L'(g) = (-\projz^{\ge 0} g, \projz^{\ge 0} \Delta \projz^{<0} g)$ 
	where~$g \in \bd_n(B,\Delta^{-1}B^*)$.

	Let~\mbox{$u,v \in \bd_n(\Delta^{-1}B^*, B)$}.
	By definition, $\projz^{\ge 0} v = v$ and $\projz^{< n} \Delta u = \Delta u$.
	Recalling~\eqref{eq:projzDagger} and 
	using~$\projz^{\ge 0} \projz^{< n} = \projz^{< n} \projz^{\ge 0}$,
	we have
	\begin{align}
		[\Delta(u,v)]_0 &= [(\projz^{< n}_{B^*} \Delta u)(\projz^{\ge 0}_B v) ]_0
					= 
					[((\projz^{\ge 0}_B)^\dag \Delta u)((\projz^{< n}_{B^*})^\dag v) ]_0\label{eq:LEtaL}\\
					&=
					[(\projz^{\ge 0}_{B^*} \Delta u ) (\projz^{<n}_B v) ]_0
					=
					[\eta_t( L(u), L(v) )]_0. \nonumber
	\end{align}
	In particular, $[\eta_t|_L]_0$ is $\mp$-even.
	Therefore, $L$ is a sublagrangian.

	Let $g,h \in \bd_n(B,\Delta^{-1}B^*)$.
	Then, $\projz^{\ge 0}\Delta g = \Delta g$.
	Using~\eqref{eq:projzDagger}, we have
	\begin{align}
		[\eta_t(L'(g),L'(h))]_0 
		&=
		-[(\projz^{\ge 0}_{B^*} \Delta \projz^{< 0}_B g)(\projz^{\ge 0}_B h)]_0
		=
		-[((\projz^{\ge 0}_B)^\dag \projz^{\ge 0}_{B^*} \Delta (1 - \projz^{\ge 0}_B) g)(h)]_0
		\nonumber\\
		&=
		[-(\projz^{\ge 0}_{B^*} \Delta g)(h) + (\projz^{\ge 0}_{B^*} \Delta \projz^{\ge 0}_B g)(h)]_0
		=
		[-(\Delta g)(h) + (\projz^{\ge 0}_{B^*} \Delta \projz^{\ge 0}_B g)(h)]_0 \nonumber\\ 
		&=
		[(\projz^{\ge 0}\Delta\projz^{\ge 0} - \Delta)(g, h)]_0, \label{eq:LstarEtaLstar}
	\end{align}
	which is $\mp$-even.
	Therefore, $L'$ is a sublagrangian.

	Since $\lambda = [(\eta_t \pm \eta_t^\dag)|_{L \oplus L'}]_0 = [(\eta_t \pm \eta_t^\dag)|_{M \oplus M^*}]_0$ 
	is an isomorphism from a module to its dual,
	we have by~\ref{thm:ComplementaryLagrangiansArePaired} that 
	$\lambda$ restricts to isomorphisms~$L \cong (L')^*$ and~$L' \cong L^*$.
	Hence, the orthogonal complements of~$L$ and~$L'$ are themselves.
	Therefore, both~$L$ and~$L'$ are lagrangians.
	The unitary~$U$ exists by~\ref{thm:formation-unitary}.
\end{proof}

For completeness, we note that for any $u \in \bd_n(\Delta^{-1} B^*, B)$ 
and~$g \in \bd_n(B, \Delta^{-1}B^*)$,
\begin{align} 
	[\eta_t(L(u), L'(g))]_0 &= 
	-[(\projz^{\ge 0} \Delta u)(\projz^{\ge 0} g)]_0 
	= [(-\projz^{\ge 0}\Delta)(u, g) ]_0,
	\label{eq:LstarEtaL}\\
	[\eta_t(L'(g), L(u))]_0 &= [(\projz^{\ge 0} \Delta \projz^{<0} g)( \projz^{\ge 0} u)]_0 
	=[(\Delta \projz^{<0})(g,u)]_0. \nonumber
\end{align}
A similar calculation will appear in the proof in the following.

\begin{proposition}[Lem.\,3.1~\cite{Ranicki2}]\label{thm:FromFormsDownToUnitary}
	The construction from~$\Delta$ to~$L(\Delta,B,n)$ by~\ref{thm:etaOnLLstar} 
	induces a well-defined abelian group homomorphism
	\begin{align}
		\bass^\downarrow_{2\iota}: \switt^\mp(\Rz) \ni [\Delta] 
		\mapsto [U \text{ such that } U (M\oplus 0) = L(\Delta,B,n)] \in \qumod^\pm(R).
	\end{align}
\end{proposition}

The notation is that $\bass^\downarrow_2$ ($\iota = 1$) is for $\switt^-$ 
and $\bass^\downarrow_0$ ($\iota=0$) for~$\switt^+$.

\begin{proof}
	We have to show that the map is independent of~$n$ that is just large enough.
	Also, we have to show that Witt equivalent forms are mapped to the same class.

	The change $L(\Delta,B,n)\to L(\Delta,B,n+1)$ might not be too clear as the $z$-degree projector involves~$n$,
	but that of~$L^*(\Delta,B,n)$ is easy to handle.
	By~\eqref{eq:bdABoneMore2}, we have $\bd_{n+1}(B,\Delta^{-1}B^*) = z^n B \oplus \bd_{n}(B,\Delta^{-1}B^*)$.
	By the definition of~$L^*$, we see $L^*(\Delta,B,n+1) = (z^n B \oplus 0) \oplus L^*(\Delta,B,n)
	\subseteq \bigoplus_{k=0}^n z^k (B\oplus B^*)$.
	But \ref{thm:formation-unitary} says that $[U] \in \qumod^\pm$ 
	is determined by one of the lagrangians in the image, given a fixed pair of lagrangians in the domain.
	Therefore, the map~$\bass^\downarrow_{2\iota}$ is independent of~$n$ as long as it is sufficiently large.

	If $\Delta$ is replaced by~$\Delta \oplus \Delta'$,
	then, because we have shown the independence on~$n$,
	we may choose~$n$ to be sufficiently large that it works for both~$\Delta$ and~$\Delta'$.
	Every module or form for $\Delta \oplus \Delta'$ is the direct sum of those for~$\Delta$ and~$\Delta'$.
	Hence, $\bass^\downarrow_{2\iota}$ preserves the group operation.
	In particular, if $\Delta = \lambda^\mp$, then we may choose $n=0$ so that both~$L$ and~$L^*$ vanish.
	Hence, the images under~$\bass^\downarrow_{2\iota}$ of~$\Delta$ 
	and of $\Delta \oplus \lambda^\mp$ are the same.

	Finally, suppose~$\Delta$ is replaced with a congruent form~$E^\dagger \Delta E$ 
	for some~$E\in \gl(t;\Rz)$.
	Choose~$n$ so large that $E B \subseteq \bigoplus_{j=-n}^{n} z^j B$ 
	and $E^{-1} B \subseteq \bigoplus_{j=-n}^{n} z^j B$.
	Choose $n' \ge 2n$.
	Then,
	\begin{align}
		E B^+
		&=
		\bd_{n'}(B, EB) \oplus z^{n'} B^+
		\nonumber\\
		B^+ 
		&= 
		\bd_n(\Delta^{-1} B^*,B) \oplus z^n \Delta^{-1} (B^*)^+
		\\
		\Delta^{-1}(B^*)^+ 
		&=
		\Delta^{-1}\bd_{n'}(E^{-\dagger}B^*,B^*) \oplus z^{n'} \Delta^{-1} E^{-\dagger}(B^*)^+.
		\nonumber
	\end{align}
	On the other hand,
	$
		E B^+ 
		= 
		\bd_{n+2n'} (\Delta^{-1} E^{-\dagger} B^*, E B) \oplus z^{n+2n'} \Delta^{-1} E^{-\dagger} (B^*)^+
	$, so
	\begin{align}
		\bd_{n+2n'}(\Delta^{-1} E^{-\dagger} B^*, E B) 
		&= z^{n+n'} \Delta^{-1} \bd_{n'}(E^{-\dagger}B^*, B^*)\nonumber\\
		&\qquad\oplus z^{n'} \bd_{n}(\Delta^{-1} B^*, B) \\
		&\qquad \oplus \bd_{n'} (B,EB).\nonumber
	\end{align}
	It follows that
	\begin{align}
		L(E^\dagger \Delta E,B, n+2n') 
		&=
		\begin{pmatrix}
		\projz^{<n+2n'} \\
		\projz^{\ge 0} E^\dagger \Delta E
		\end{pmatrix}
		\bd_{n+2n'}(E^{-1} \Delta^{-1} E^{-\dagger} B^*, B)
		\nonumber\\
		&=
		\left\{
		\begin{pmatrix}
		\projz^{<n+2n'}E^{-1} \\
		\projz^{\ge 0} E^\dagger \Delta
		\end{pmatrix} (a + z^{n'} b + z^{n+n'} \Delta^{-1} c^*)
		\,\middle|\,
		\begin{array}{rl}
		a &\in \bd_{n'}(B,EB),\\
		b &\in \bd_n(\Delta^{-1}B^*,B),\\
		c^* &\in \bd_{n'}(E^{-\dagger}B^*,B^*)
		\end{array}
		\right\}.\label{eq:abcstar}
	\end{align}
	Considering the bounds on the $z$-exponents and that $n' \ge 2n$,
	some projections become redundant:
	\begin{align}
		L(E^\dagger \Delta E, B, n+2n') &=
		\left\{
			\begin{pmatrix}
		E^{-1} a \\
		\projz^{\ge 0} E^\dagger \Delta a
		\end{pmatrix} +
		z^{n'}
		\begin{pmatrix}
		E^{-1} b\\
		E^\dagger \Delta b
		\end{pmatrix} +
		z^{n+n'}
		\begin{pmatrix}
		\projz^{<n'}E^{-1} \Delta^{-1}c^*\\
		E^\dagger c^*
		\end{pmatrix}
		\right\}\, .\label{eq:LEdaggerDeltaEBNTwoNPrime}
	\end{align}
	Write $\BB = B \oplus B^*$.
	Define $\hat H$ and~$\hat L$, both within $\bigoplus_{j=0}^{n+2n'-1} z^j \BB$, by
	\begin{align}
		\hat H &= \left\{
			\begin{pmatrix}
		0 \\
		E^{\dagger} a^*
		\end{pmatrix} 
		\right\}
		\oplus
		z^{n'}\cx(E^{-1})L^*(\Delta,B,n)
		\oplus
		\left\{
		z^{n+n'}
		\begin{pmatrix}
		E^{-1} c\\
		0
		\end{pmatrix}
		\right\}
		,\label{eq:hatHL}\\
		\hat L &=
		\left\{
			\begin{pmatrix}
		E^{-1} a \\
		0
		\end{pmatrix} 
		\right\}
		\oplus
		z^{n'} \cx(E^{-1})L(\Delta,B,n)
		\oplus
		\left\{
		z^{n+n'}
		\begin{pmatrix}
		0\\
		E^\dagger c^*
		\end{pmatrix}
		\right\}\nonumber
	\end{align}
	where~$a \in \bd_{n'}(B,EB)$ and~$ c^* \in \bd_{n'}(E^{-\dagger}B^*,B^*)$ are as in~\eqref{eq:abcstar},
	but $a^* \in \bd_{n'}(B^*,E^{-\dagger} B^*)$ and $c \in \bd_{n'}(EB,B)$ are duals.
	The sums in~\eqref{eq:hatHL} are direct because, if we apply $\cx(E)$, 
	the three summands belong to different $z$-degree sectors.

	The parent $R$-module $\bigoplus_{j=0}^{2n'-1} z^j \BB$ has an obvious quadratic $\pm$-form $[\eta|_{\cdots}]_0$,
	and decomposes as
	\begin{equation}
		\bd_{n'}(\cx(E^{-1}) \BB,  \BB)
		\oplus
		z^{n'} \cx(E^{-1}) \bd_n(\BB, \BB)
		\oplus
		z^{n+n'} \bd_{n'}(\BB, \cx(E^{-1})\BB), \label{eq:parentDecomposition}
	\end{equation}
	which is an orthogonal sum as seen by the $z$-degrees.
	Hence, the first direct summand of~$\hat H$ is a lagrangian within~$\bd_{n'}(\cx(E^{-1})\BB, \BB)$,
	and
	the third direct summand of~$\hat H$ is a lagrangian within~$z^{n+n'}\bd_{n'}(B \oplus B^*,\cx(E^{-1})\BB)$.
	Since the second direct summand of~$\hat H$ is a lagrangian by~\ref{thm:etaOnLLstar},
	all three direct summands are lagrangians in disjoint parent modules.
	Therefore, $\hat H$ is a lagrangian of~$\bigoplus_{j=0}^{n+2n'-1} z^j \BB$.
	A parallel argument shows that $\hat L$ is also a lagrangian.

	We claim that $\hat H$ and $L(E^\dagger \Delta E, B, n+2n')$ in~\eqref{eq:LEdaggerDeltaEBNTwoNPrime}
	make a nonsingular pair 
	with respect to $[(\eta \pm \eta^\dag)|_{(\bigoplus_{j=0}^{n+2n'-1} z^j \BB)}]_0$.
	To this end, we have to show that 
	$S = \hat H + L(E^\dagger \Delta E, B, n+2n')$ contains $\bigoplus_{j=0}^{2n'-1} z^j \BB$.
	Once this is shown, the claim follows by~\ref{thm:ComplementaryLagrangiansArePaired}
	since we know that the two summands of~$S$ are lagrangians.

	By examining the range of~$a,a^*,c,c^*$ we see that the first
	and the third direct summands in~\eqref{eq:parentDecomposition} are contained in~$S$.
	Mapping the parent module by~$\cx(E)$, 
	the middle direct summand in~\eqref{eq:parentDecomposition} becomes $z^{n'}\bd_n(\BB,\BB)$,
	a submodule defined by a $z$-degree interval~$[n',n+n')$, 
	outside of which we know at least $z$-degree intervals $[n' - n,n')$ and $[n+n', 2n+n')$ 
	are fully covered by $\cx(E)S$.
	Hence, it suffices to show that 
	$\projz^{\ge n'}_\BB \projz^{< n+n'}_\BB \cx(E) S \supseteq \projz^{\ge n'}_\BB \projz^{< n+n'}_\BB$.
	But $\projz^{\ge n'}_\BB \projz^{< n+n'}_\BB \cx(E) L(E^\dagger \Delta E, B, n+2n') = z^{n'} L(\Delta,B,n)$
	by the definition of~$L(\Delta,B,n)$,
	while $\projz^{\ge n'}_\BB \projz^{< n+n'}_\BB \cx(E) \hat H = z^{n'} L^*(\Delta,B,n)$,
	so \ref{thm:Lcomplement} gives the result.

	In addition, it is evident that $\hat H$ and $\hat L$ make a nonsingular pair.

	Applying~\ref{thm:formation-unitary} twice, 
	we see that the unitary~$U$ we are investigating can be defined by~$\hat L$
	instead of~$L(E^\dagger \Delta E, B, n+2n')$.
	Furthermore, without altering the class~$[U]$, we can instead examine a lagrangian $\cx(E)\hat L$
	which reads
	\begin{align}
		\cx(E)\hat L &=
		\left\{
			\begin{pmatrix}
		a \\
		0
		\end{pmatrix} 
		\right\}
		\oplus
		z^{n'} L(\Delta,B,n)
		\oplus
		\left\{
		\begin{pmatrix}
		0\\
		c^*
		\end{pmatrix}
		\right\}
		\subseteq
		\bigoplus_{j=0}^{n+2n'-1}
		z^j
		\begin{pmatrix}
		E B\\
		E^{-\dagger} B^*
		\end{pmatrix} \, .
	\end{align}
	The parent $R$-module admits a decomposition
	\begin{align}
		\bigoplus_{j=0}^{n+2n'-1}
		z^j
		\begin{pmatrix}
		E B\\
		E^{-\dagger} B^*
		\end{pmatrix}
		=
		\begin{pmatrix}
		\bd_{n'}(B,EB)\\
		\bd_{n'}(B^*,E^{-\dagger}B^*)
		\end{pmatrix}
		\oplus z^{n'} 
		\begin{pmatrix}
		\bd_n(B,B)\\
		\bd_n(B^*,B^*)
		\end{pmatrix}
		\oplus 
		z^{n+n'} 
		\begin{pmatrix}
		\bd_{n'}(EB,B)\\
		\bd_{n'}(E^{-\dagger}B^*,B^*)
		\end{pmatrix}.
	\end{align}
	Hence, the unitary that maps a lagrangian~$\bigoplus_{j=0}^{n+2n'-1} z^j EB$ to the lagrangian~$\cx(E)\hat L$
	can be taken to be block-diagonal,
	where the first block is the identity, 
	the second is~$U$ which maps~$\bd_n(B,B)$ to~$L(\Delta,B,n)$,
	and the third is some elementary unitary
	which maps $\bd_{n'}(EB,B)$ to $\bd_{n'}(E^{-\dag}B^*,B^*)$.
\end{proof}

\subsection{Going up from unitaries to forms}

For any $U \in \qhp^{\mp}(Q;R)$, the product~$U^\dagger \eta U$ has matrix entries
\begin{align}
	U^\dagger \eta U = 
	\begin{Bmatrix} Q^* \\ Q \end{Bmatrix}
	\begin{pmatrix}
		\xi \pm \xi^\dag &&& \gamma \\
		\delta &&& \rho \pm \rho^\dag
	\end{pmatrix}
	\begin{Bmatrix} Q \\ Q^* \end{Bmatrix}
	\qquad (\gamma \mp \delta^\dagger = I). \label{eq:xigammadeltarho}
\end{align}
Using
\begin{align}
	U &= \begin{pmatrix} a & b \\ c & d \end{pmatrix},&
	U^{-1} &= \lambda^{-1} U^\dag \lambda = \begin{pmatrix}
	d^\dag & \mp b^\dag \\ \mp c^\dag & a^\dag
	\end{pmatrix}, \label{eq:UUinv}
\end{align}
\begin{align}
	\begin{pmatrix}
	a^\dagger c \mp c^\dagger a &&& a^\dagger d \mp c^\dagger b \\
	b^\dagger c \mp d^\dagger a &&& b^\dagger d \mp d^\dagger b
	\end{pmatrix}
	&=
	\begin{pmatrix}
	0 & 1 \\ \mp 1 & 0
	\end{pmatrix},
	&
	\begin{pmatrix}
	ba^\dagger \mp ab^\dagger &&& bc^\dagger \mp ad^\dagger \\ 
	da^\dagger \mp cb^\dagger &&& dc^\dagger \mp cd^\dag
	\end{pmatrix}
	&=
	\begin{pmatrix}
	0 & \mp 1 \\ 1 & 0
	\end{pmatrix},\label{eq:unitarity}
\end{align}
we may write
\begin{align}
	U^\dagger \eta U  = \begin{pmatrix}
	a^\dag c &&& a^\dag d \\ b^\dag c &&& b^\dag d
	\end{pmatrix}.
\end{align}

\begin{proposition}\label{thm:BassUpFromUnitary}
	For any finitely generated free $R$-module~$Q$,
	let $Q_z = \bigoplus_j z^j Q$ be an $\Rz$-module.
	For any $U \in \qhp^\mp(Q;R)$, we use~\eqref{eq:xigammadeltarho} to define
	\begin{align}
		\bass^\uparrow (U) = 
		\begin{Bmatrix} Q_z^* \\ Q_z \end{Bmatrix}
		\begin{pmatrix} 
			\xi\pm \xi^\dag & -z \gamma \pm \delta^\dag \\ 
			\delta\mp \tfrac 1 z\gamma^\dag & (2-z-\tfrac 1 z)(\rho \pm \rho^\dag )
		\end{pmatrix}
		\begin{Bmatrix}	Q_z \\ Q_z^*	\end{Bmatrix} 
		\label{eq:DefBassUpFromUnitary}
	\end{align}
	where $z$ is an independent variable.
	Then, $\bass^\uparrow(U)$ is a nonsingular even hermitian $\pm$-form over~$\Rz$.
	The induced map 
	\begin{align}
		\bass^\uparrow_{2\iota+1} : \qumod^\mp(R) \ni [U] \mapsto [\bass^\uparrow(U)] \in \switt^\pm(\Rz)
	\end{align}
	is well defined. Here, $\iota = 0$ is for $\qumod^+$ and $\iota =1$ is for $\qumod^-$.
\end{proposition}

Ranicki~\cite[after Lemma~3.2]{Ranicki2} defines a map similar to~$\bass^\uparrow(U)$ by
\begin{align}
	\tilde \bass^\uparrow(U) = \begin{pmatrix} 
		\xi &&& -z \gamma \\ 
		\delta &&& (1-z)(\rho\pm \rho^\dag)
	\end{pmatrix}.\label{eq:tildeBass}
\end{align}
As it stands, this is ill-defined because $\xi$ is not determined by~$U$.
It turns out that 
even though $\tilde \bass^\uparrow(U)$ is ill-defined,
its Witt class~$[\tilde \bass^\uparrow(U)]$ is well-defined,
which we will show in~\ref{thm:tildeBass} at the end of this section.

\begin{proof}
	To show that $\bass^\uparrow(U)$ is nonsingular,
	we observe the following expression which can be checked by straightforward calculation.
	\begin{align}
		\bass^\uparrow(U)
		&=
		\begin{pmatrix}
		a^\dag c &&& -z a^\dag d \pm c^\dag b \\
		b^\dag c \mp \tfrac 1 z d^\dag a &&& (2-z-\tfrac 1 z)b^\dag d
		\end{pmatrix}
		\\&= 
		\begin{pmatrix}
		\frac{z}{z - 1}I & 0 \\ 0 & I
		\end{pmatrix}
		U^\dagger
		\begin{pmatrix}
		0 & I \\ \mp \tfrac 1 z I & 0
		\end{pmatrix}
		U
		\begin{pmatrix}
		I & 0 \\ 0 & (1-z)I
		\end{pmatrix}\nonumber \label{eq:fivematrices}
	\end{align}
	Although $\frac 1 {z-1} \notin \Rz$, by construction we know $z -1 \in \Rz$ is a non-zero-divisor,
	and the formula makes sense in the one-element localization~$\Rz_{z-1}$ and the final result is over~$\Rz$.
	Since $U$ has a unit determinant, 
	we see that $\bass^\uparrow(U)$ has a unit determinant as well.
	It is easy to check that $\bass^\uparrow(U)^\dag = \pm \bass^\uparrow(U)$.

	Next, we show that $\bass^\uparrow(E U)$ is Witt equivalent to $\bass^\uparrow(U)$ for any elementary~$E$.
	If $E$ is an embedding~$\qhp^\mp(q;R) \to \qhp^\mp(q+1;R)$,
	then the result is to direct-sum a trivial quadratic $\pm$-form~$\begin{pmatrix} 0 & -z \\ 0 & 0 \end{pmatrix}$.
	If~$E = \cx(\alpha)$, then $U^\dagger E^\dagger \eta E U = U^\dagger \eta U$, so $\bass^\uparrow(U) = \bass^\uparrow(EU)$.
	We will show that $[\bass^\uparrow(E U)] = [\bass^\uparrow(U)]$ when~$E = \cz(\mu)$ or~$E = \cz(\mu)^\dagger$.
	Then, $[\bass^\uparrow(\hada_\mp U)] = [\bass^\uparrow(\hada_\mp U \hat\oplus I)]$,
	which is equal by~\ref{lem:HHbyZX} to 
	$[\bass^\uparrow((\hada_\mp^{-1} \hat\oplus \hada_\mp^{-1})(\hada_\mp U \hat\oplus I))]
	= [\bass^\uparrow(U \hat\oplus \hada_\mp^{-1})] = 
	[\bass^\uparrow(U) \oplus \bass^\uparrow(\hada_\mp^{-1})]$,
	but $\bass^\uparrow(\hada_\mp) = -\lambda^\pm_1$.
	Therefore, we will conclude that
	$[\bass^\uparrow(EU)] = [\bass^\uparrow(U)]$ for all elementary~$E$.

	For the remaining case of~$E=\cz(\mu)$ and~$E=\cz(\mu)^\dag$,
	we have
	\begin{align}
		&\begin{matrix}
		(I + T^\dag)\bass^\uparrow(U) (I + T) = \bass^\uparrow(\qcz(\mu)U),\\
		(I + S^\dag)\bass^\uparrow(U) (I + S) = \bass^\uparrow(\qcz(\mu)^\dagger U),
		\end{matrix}
		&
		\qcz(\mu) = \begin{pmatrix} I && 0 \\ \mu \pm \mu^\dag && I \end{pmatrix}.
	\end{align}
	where 
	\begin{align}
		T &= \begin{pmatrix}
		\mp(1-\tfrac 1 z) b^\dag \\ 
		-\tfrac 1 z a^\dag
		\end{pmatrix}
		\mu
		\begin{pmatrix}
		a && (1-z) b
		\end{pmatrix} ,
		&
		S &= \begin{pmatrix}
		\pm(1-z) d^\dag \\
		-c^\dag
		\end{pmatrix}
		\mu
		\begin{pmatrix}
		c && (1-z) d
		\end{pmatrix}.
	\end{align}
	(This is easy to write as a proof, but a better calculation is below.)
\end{proof}

The congruence in the last part of the proof was found 
by the following observation. 
Write $\bass^\uparrow(U) = A U^\dag Z U C$ where $A,Z,C$ are written in~\eqref{eq:fivematrices}.
Then,
\begin{align}
	\bass^\uparrow(\qcz(\mu) U) - \bass^\uparrow(U) = (1 - \tfrac 1 z) A U^\dag \diag(\mu \pm \mu^\dag,0) U C.
\end{align}
Reading the ``first order terms'' in~$\mu$, we find 
\begin{align}
	T 
	= (1 - \tfrac 1 z) \bass^\uparrow(U)^{-1} A U^\dag \diag(\mu,0) U C 
	= (1 - \tfrac 1 z) C^{-1} U^{-1} Z^{-1} \diag(\mu,0) U C,
\end{align}
which turns out to be correct.

\begin{remark}\label{rem:alt-bass-up-from-unitary}
	It follows from~\ref{thm:BassUpFromUnitary} that 
	the following matrix is Witt equivalent to~$\bass^\uparrow(U)$.
	\begin{align}
		\begin{pmatrix}
		0 & I \\ \pm \tfrac 1 z I & 0
		\end{pmatrix}
		\bass^\uparrow(U \hada_\mp^{\hat\oplus q}) \begin{pmatrix}
		0 & \pm z I \\ I & 0
		\end{pmatrix}
		&=
		\begin{pmatrix}
		(2-z-\tfrac 1 z) (\xi \pm \xi^\dag) &&& -z \gamma \pm \delta^\dag\\
		\delta \mp \tfrac 1 z \gamma^\dag &&& \rho \pm \rho^\dag
		\end{pmatrix}\\
		&=
		\begin{pmatrix}
		(1-z) (\xi \pm \xi^\dag) & -z \gamma\\
		\delta & \rho
		\end{pmatrix}
		\pm
		\begin{pmatrix}
		(1-z) (\xi \pm \xi^\dag) & -z \gamma\\
		\delta & \rho
		\end{pmatrix}^\dag
		\nonumber
	\end{align}
\end{remark}

\begin{proposition}[Lem.\,3.3 of~\cite{Ranicki2}]\label{thm:up-and-down-from-unitary-is-identity}
	The composition
	\begin{align}
		\qumod^\mp(R) 
		\xrightarrow{\quad\bass^\uparrow_{2\iota + 1}\quad} \switt^\pm(\Rz)
		\xrightarrow{\quad\bass^\downarrow_{2\iota+2}\quad} \qumod^\mp(R)
	\end{align}
	is the identity. 
	In particular, $\bass^\downarrow_{2\iota+2}$ is surjective.
\end{proposition}

\begin{proof}
	Let $U \in \qhp^\mp(Q;R)$ where $Q$ is a free $R$-module of rank~$q$.
	Let $\Delta = \bass^\uparrow(U) : Q_z \oplus Q_z^* \to Q_z^* \oplus Q_z$.
	The $\Rz$-module $Q_z \oplus Q_z^*$ has a modular $R$-base $B = Q \oplus Q^*$.
	Since $\Delta^{-1} B^* \subseteq \tfrac 1 z B \oplus B \oplus z B$,
	we have~$z\Delta^{-1} (B^*)^+ \subseteq B^+$ and~$\bd_1(\Delta^{-1}B^*, B) = z \Delta^{-1} (B^*)^- \cap B^+$ is defined.
	If $\begin{pmatrix}u\\ v\end{pmatrix} \in \bd_1(\Delta^{-1}B^*, B) \subseteq B^+$, 
	then $u$ and $v$ do not have any negative $z$-exponents.
	The elements~$u$ and~$v$ cannot have any term with $z^2$ or higher degree because otherwise 
	$\Delta\begin{pmatrix}u\\v\end{pmatrix} \notin z B^-$.
	So, the most general element we consider is of form $\begin{pmatrix}u_0 + u_1 z \\ v_0 + v_1 z\end{pmatrix}$
	where $\begin{pmatrix} u_0 \\ v_0\end{pmatrix}, \begin{pmatrix}u_1 \\ v_1\end{pmatrix} \in B$.
	Using~\eqref{eq:fivematrices}, we have
	\begin{align}
		\Delta\begin{pmatrix}	u \\ v\end{pmatrix}	 
		&= 
		\begin{pmatrix}
		a^\dag c u - z a^\dag d v \pm c^\dag b v\\
		b^\dag c u \mp \tfrac 1 z d^\dag a u + (2-z-\tfrac 1 z)b^\dag d v
		\end{pmatrix}.
	\end{align}
	The requirement that $\Delta\begin{pmatrix}u_0 + u_1 z\\ v_0 + v_1 z\end{pmatrix} \in z B^-$
	is equivalent to
	\begin{align}
		z^2 &:
		\begin{pmatrix}
		-z a^\dag d v_1 z\\
		-z b^\dag d v_1 z 
		\end{pmatrix} = 0,&
		z^1 &:
		\begin{pmatrix}
		a^\dag c u_1 z - z a^\dag d v_0 \pm c^\dag b v_1 z\\
		b^\dag c u_1 z + 2b^\dag d v_1 z - z b^\dag d v_0
		\end{pmatrix} = 0
		\\
		\Longleftrightarrow \qquad&\begin{cases}
		a^\dag d v_1 = 0\\
		b^\dag d v_1 = b^\dag(c u_1 - dv_0) = 0\\
		\pm c^\dag b v_1 = a^\dag(d v_0 - c u_1)
		\end{cases} \, .\nonumber
	\end{align}
	Since $a^\dag d \mp c^\dag b = 1$ by~\eqref{eq:unitarity},
	we have $v_1 = - a^\dag(d v_0 - c u_1)$.
	Then, $0 = a^\dag d v_1 = -a^\dag d a^\dag (d v_0 - c u_1) = -a^\dag ( 1 \pm cb^\dag)(d v_0 - c u_1) = v_1$.
	This implies that $d a^\dag(d v_0 - c u_1) = 0$. 
	Adding $\mp c (b^\dag(d v_0 - c u_1)) = 0$,
	we see that $d v_0 - c u_1 = 0$.
	Noticing $u_0$ is unconstrained,
	we see that the most general element of $z \Delta^{-1} B^- \cap B^+$ can be written as
	\begin{align}
		\begin{pmatrix} u_0' \\ 0 \end{pmatrix} + 
		\begin{pmatrix}
		u_1' (1-z) \\ v_0
		\end{pmatrix} \qquad \text{ where }\quad c u_1'  + dv_0 = 0,
	\end{align}
	and this expression is unique.
	The constraint on $u_1',v_0$ is equivalent to $\begin{pmatrix}
	u_1'\\
	v_0
	\end{pmatrix} = U^{-1} \begin{pmatrix}
	y\\ 0
	\end{pmatrix}$ for some $y \in Q$.
	Therefore,
	\begin{align}
		\bd_{n=1}(\Delta^{-1} B^*, B ) &= 
		\begin{pmatrix} Q \\0 \end{pmatrix} \oplus 
		\begin{pmatrix} (1-z) I & 0 \\ 0 & I \end{pmatrix} U^{-1} \begin{pmatrix}Q \\0 \end{pmatrix} \, .
	\end{align}

	Next, we calculate $L(\Delta,B,n=1)$. 
	Since it is by~\ref{thm:Lcomplement} an embedding of~$\bd_1(\Delta^{-1}B,B)$ into~$B \oplus B^*$,
	we may consider the direct summands of~$\bd_1(\Delta^{-1}B,B)$ separately.
	Below, we use brackets~$[\cdots]$ to signify that the dimension (rank) of a bracketed column matrix
	is twice as large as that of a parenthesized one.
	\begin{align}
		\begin{bmatrix}
		\projz^{<1}\\
		\projz^{\ge 0} \Delta
		\end{bmatrix} \begin{pmatrix}Q \\ 0\end{pmatrix}
		&=
		\left\{
		\begin{bmatrix}
		u_0\\
		0\\
		a^\dag c u_0\\
		b^\dag c u_0
		\end{bmatrix} ~\middle|~
		u_0 \in Q
		\right\}
		\\
		\begin{bmatrix}
		\projz^{<1}\\
		\projz^{\ge 0} \Delta
		\end{bmatrix}
		\begin{pmatrix} (1-z)I_{q} && 0 \\ 0 && I_{q} \end{pmatrix} U^{-1}\begin{pmatrix}Q \\ 0\end{pmatrix}
		&=
		\begin{bmatrix}
		U^{-1}\\
		\pm U^\dagger \eta^\dagger
		\end{bmatrix} \begin{pmatrix}Q \\ 0\end{pmatrix} 
		\nonumber
	\end{align}
	where in the second line we used~\eqref{eq:fivematrices}.
	Now that we have identified~$L(\Delta,B,n)$,
	we examine what it takes to bring~$L(\Delta,B,n)$ to~$B \oplus 0 \subseteq B \oplus B^*$.
	\begin{align}
		\begin{Bmatrix}
		Q \\ Q^* \\ Q^* \\ Q
		\end{Bmatrix}
		&\begin{bmatrix}
		I_q & d^\dag \\
		0 & \mp c^\dag\\
		a^\dag c &  \pm c^\dag\\
		b^\dag c & \pm d^\dag
		\end{bmatrix}
		\xrightarrow{\cz\begin{pmatrix}-a^\dag c & \mp c^\dag b \\ -b^\dag c & -b^\dag d \end{pmatrix}}
		\begin{bmatrix}
		I_q & d^\dag \\
		0 & \mp c^\dag\\
		0 & 0\\
		0 & \pm d^\dag
		\end{bmatrix}
		\xrightarrow{\cz\begin{pmatrix} 0 & \mp I \\ -I & 0 \end{pmatrix}}
		\begin{bmatrix}
		I_q & 0 \\
		0 & \mp c^\dag\\
		0 & 0\\
		0 & \pm d^\dag
		\end{bmatrix}\\
		&\xrightarrow{I \hat\oplus \hada_\mp}
		\begin{bmatrix}
		I_q & 0 \\
		0 & \pm \hat 1 d^\dag\\
		0 & 0\\
		0 & \hat 1^{-\dag} c^\dag
		\end{bmatrix}
		\xrightarrow{\cx\begin{pmatrix}0 & \pm \hat 1 ^{-1} \\ \hat 1 & 0 \end{pmatrix}}
		\begin{bmatrix}
		0 & d^\dag \\
		I_q & 0\\
		0 & \pm c^\dag\\
		0 & 0
		\end{bmatrix}
		\xrightarrow{\begin{pmatrix} a & - b\\ -c & d \end{pmatrix} \hat\oplus I}
		\begin{bmatrix}
		0 & I_q \\
		I_q & 0\\
		0 & 0\\
		0 & 0
		\end{bmatrix}
		\nonumber
	\end{align}
	where in the first arrow we use~\eqref{eq:unitarity} 
	and in the third arrow we use $\hat 1 : Q \ni u \mapsto (v \mapsto u^\dag v) \in Q^*$.
	Hence, the unitary under which $B \oplus 0 \to L(\Delta,B,n)$ is stably equivalent to~$(U')^{-1}$,
	where $U'$ is the unitary that appears in the last arrow.
	But $(U')^{-1}$ is stably equivalent to~$U$ by~\ref{thm:TRCisInvU}.
	By~\ref{thm:formation-unitary}, the map~$\bass^\downarrow \circ \bass^\uparrow$ is determined
	and is the identity as claimed.
\end{proof}

\subsection{Kernel of the descent maps from Witt groups}

In this subsection, we use~$y$ to denote a variable independent from any other variables~$x_1,\ldots,x_\dd$, and~$z$.
So, $\Ryz$ denotes the Laurent polynomial ring with~$\dd + 2$ variables.
Recall from~\ref{thm:hermitianMatrices} the group homomorphism
\begin{align}
\sym:\qwitt^\mp \ni [\xi] \mapsto [\xi \mp \xi^\dag] \in \switt^\mp
\end{align}
for any base ring.

\begin{lemma}\label{thm:up-down-and-down-up}
	The following diagram anticommutes.\footnote{
		The unitary group at the top can be relaxed to the $\lambda$-unitary group, 
		but the lower one should be kept as the $\eta$-unitary group.
		This is to ensure the evenness on the rightmost Witt group.
	}
	\begin{align}
		\xymatrix{
			& \qumod^\mp(\Ryz)\ar@{->}[rd]^{\sym \circ \bass^\downarrow_{2\iota+1}(y)} & \\
			\qwitt^\mp(\Ry)\ar@{->}[rd]_{\bass_{2\iota}^\downarrow(y) \circ \sym}\ar@{->}[ur]^{\bass^\uparrow_{2\iota}(z)} & 
			& \switt^\mp(\Rz)\\
			& \qumod^\pm(R)\ar@{->}[ur]_{\bass^\uparrow_{2\iota -1}(z)} &
		} \label{eq:diagram}
	\end{align}
\end{lemma}

We note an intuitive reason for the skew-commutativity in~\ref{rem:skewcommutativity} below.
This lemma is stated by Novikov~\cite[Thm.\,6.4]{Novikov1}
and the proof is said to be easy~\cite{Novikov2}.
Shortly after, 
this lemma is cited by Ranicki~\cite[Lem.\,3.4]{Ranicki2} with a lot more detail,
but it appears to the present author that some amendment is still needed.

\begin{proof}
	Let $[\phi] \in \qwitt^\mp(\Ry)$ and $\Delta = \phi \mp \phi^\dag$ where the form~$\phi$ is 
	on~$Q_y = \bigoplus_{j\in \ZZ} y^j Q$;
	the $R$-submodule~$Q$ is a modular $R$-base of an $\Ry$-module~$Q_y$.
	The module~$Q$ is finitely generated free over~$R$.
	Let $\psi : Q_y^* \to Q_y$ be any quadratic $\mp$-form such that $\psi \sim \Delta^{-1} \phi \Delta^{-1}$
	so $\psi \mp \psi^\dag = \Delta^{-1}$.
	There is an integer~$n \ge 0$ such that
	\begin{align}
		(\phi \mp \phi^\dag) Q \subseteq \bigoplus_{j=-n}^n y^j Q^*,
		\qquad
		(\psi \mp \psi^\dag) Q^* \subseteq \bigoplus_{j=-n}^n y^j Q,
		\qquad 
		\psi Q^* \subseteq \bigoplus_{j=0}^n y^j Q,
	\end{align}
	where the last requirement is ensured by, if necessary,
	adding some even $\pm$-form.

	We first follow the upper two arrows in the diagram.
	Let $T = \begin{pmatrix} 1 & \psi \\ 1 & \pm \psi^\dag \end{pmatrix}$ be as in~\ref{thm:WittGroupOfForms}.
	Hence, the~$\eta$-unitary~$\bass^\uparrow_{2\iota}(z)([\phi])$ 
	is represented by~$U = T^{-1} \diag(zI,I)T : Q_{y,z} \oplus Q_{y,z}^* \to Q_{y,z} \oplus Q_{y,z}^*$ over~$\Ryz$
	where $Q_{y,z} = \bigoplus_{j,k \in \ZZ} y^j z^k Q$.

	Next, we have to determine the boundary module $\bd_1^{(y)}(U (Q_{z}\oplus Q_{z}^*), Q_{z}\oplus Q_{z}^*)$
	where $Q_z = \bigoplus_{k \in \ZZ} z^k Q$ is our choice of a modular $\Rz$-base of $Q_{y,z}$,
	so $Q_z^+ = \bigoplus_{j \ge 0} y^j Q_z$.
	It will be more convenient to work with~$T \bd_1^{(y)}(U (Q_{z}\oplus Q_{z}^*), Q_{z}\oplus Q_{z}^*) = 
	\bd_1^{(y)}(\diag(z,1) T (Q_{z}\oplus Q_{z}^*), T (Q_{z}\oplus Q_{z}^*))$.
	We will use the alternative characterization of the boundary module~$S^\perp/S$ of~\ref{thm:separatorOnly}
	where
	\begin{align}
		S  = \begin{pmatrix}zI & 0 \\ 0 & I\end{pmatrix} T \begin{pmatrix} y^n \Delta^{-1} (Q_z^*)^+ \\ 0 \end{pmatrix} 
		= \left\{ \begin{pmatrix} z\Delta^{-1} c \\ \Delta^{-1} c \end{pmatrix} \in Q_{y,z} \oplus Q_{y,z}
		~\middle|~ c \in \projy^{\ge n}_{Q_z^*} \right\}.
	\end{align}
	We are going to show that a map
	\begin{align}
		f = 
		\begin{Bmatrix} 
			Q_{y,z} \\ Q_{y,z}
		\end{Bmatrix}
		\begin{pmatrix}
		I &&& 0 \\
		\Delta^{-1} (1 - \nu + \tfrac 1 z \nu) \Delta &&& \Delta^{-1}
		\end{pmatrix}		
		\begin{Bmatrix} 
			\bd_n^{(y)}(\Delta^{-1} Q_z^*, Q_z) \\
			\bd^{(y)}_n(\Delta Q_z, Q_z^*)
		\end{Bmatrix}
		& &
		(\nu = \projy^{\ge 0})
	\end{align}
	induces an $\Rz$-module isomorphism
	\begin{align}
		\bar f : 
		\bd_n^{(y)}(\Delta^{-1} Q_z^*, Q_z) \oplus \bd^{(y)}_n(\Delta Q_z, Q_z^*)
		\xrightarrow{\quad f \quad}
		S^\perp
		\to
		S^\perp / S 
	\end{align}
	where the orthogonal complement is with respect to 
	the hermitian $\mp$-form~$[(\Delta \oplus -\Delta)|_{T(Q_z \oplus Q_z^*)^+}]_{y^0}$
	and the second map is the canonical projection.

	To show that $\bar f$ is an isomorphism, we check the following (i) --- (iv).
	We use~$a'' \in \bd_n^{(y)}(\Delta^{-1} Q_z^*, Q_z)$ and~$b'' \in \bd^{(y)}_n(\Delta Q_z, Q_z^*)$
	to denote general elements.

	(i)~\emph{$\diag(\phi,-\phi)$ restricted to~$S$ is identically zero. In particular, $S \subseteq S^\perp$.}

	\noindent
	This is easy:
	$\phi(z \Delta^{-1}c, z \Delta^{-1} c) - \phi(\Delta^{-1}c , \Delta^{-1} c) = 0$
	due to sesquilinearity.

	(ii)~\emph{$f$ maps into~$S^\perp$.}

	\noindent
	We have to check if the image of~$f$ is orthogonal to~$S$.
	We look at the images of~$a''$ and~$b''$.
	First,
	$\Delta(z \Delta^{-1} c, a'') - \Delta(\Delta^{-1} c, \Delta^{-1}(1-\nu+\tfrac 1 z \nu)\Delta a'')
	=
	\Delta(\Delta^{-1} c, \tfrac 1 z a'') - \Delta(\Delta^{-1} c, \Delta^{-1}(1-\nu+\tfrac 1 z \nu)\Delta a'')
	=
	\Delta(\Delta^{-1} c, (\tfrac 1 z -1) \Delta^{-1} (1-\nu)\Delta a'')
	=
	\mp (\tfrac 1 z -1) c(\Delta^{-1}(1-\nu)\Delta a'')
	$,
	but $c \in \projy^{\ge n}$ while $\Delta^{-1}(1-\nu) \circ \in \projy^{<n}$.
	Hence, this expression has no $y^0$-term.
	Second, $\Delta^{-1} b'' \in \projy^{< n}$.
	Hence, $-\Delta(\Delta^{-1} c, \Delta^{-1}b'') = \pm c(\Delta^{-1}b'')$ has no $y^0$-term.

	(iii)~\emph{$\bar f$ is injective.}

	\noindent
	Suppose $f(a'',b'') \in S$.
	The first component of $f(a'',b'')$ is $a'' \in \bd_n^{(y)}(\Delta^{-1} Q_z^*, Q_z) \subseteq \Delta^{-1} \projy^{< n}_{Q_z^*}$,
	but the supposition says $a'' \in \Delta^{-1} \projy^{\ge n}_{Q_z^*}$.
	The intersection is zero.
	Hence,~$a'' = 0$ and~$f(a'',b'') = f(0,b'')=(0, \Delta^{-1} b'')$
	which must be equal to $(z\Delta^{-1}c, \Delta^{-1}c)$ for some $c$,
	which is possible only if $\Delta^{-1}b'' = 0$ and $b''= 0$.

	(iv)~\emph{$\bar f$ is surjective, i.e., $S^\perp \subseteq S + \im f$.}

	\noindent
	Due to our choice of $\psi$ that maps the nonnegative sector $\projy^{\ge 0}_{Q_z^*}$ into $\projy^{\ge 0}_{Q_z}$,
	we may write a general element~$(s+\psi v, s \pm \psi^\dag v)$ of~$T(Q_z \oplus Q_z^*)^+$ as 
	$(a, a + \Delta^{-1} b)$ where $a = s + \psi v \in Q_z^+$ and $b = -v \in Q_z^+$ are arbitrary.
	The orthogonal complement~$S^\perp$ consists of~$(a,b)$ such that
	\begin{align}
		\Delta(z \Delta^{-1} c, a) - \Delta(\Delta^{-1} c, a+ \Delta^{-1} b) &= 0,\\
		c(\tfrac 1 z a - a - \Delta^{-1} b ) &= 0\nonumber
	\end{align}
	for all $c \in \projy^{\ge n}_{Q_z^*} Q_{y,z}^*$.
	Since $[c(\ell)]_0 = 0$ for all~$c \in \projy^{\ge n}_{Q_z^*}$ if and only if~$\projy^{\ge n}_{Q_z} \ell = 0$,
	we arrive at an equation
	\begin{align}
		\projy^{\ge n} \Delta^{-1} b = (\tfrac 1 z -1) \projy^{\ge n} a.
	\end{align}
	In words, the $\projy^{\ge n}$-sector of $\Delta^{-1}b$ is determined by a given~$a \in \projy^{\ge 0}$.
	Since $(\tfrac 1 z - 1)\Delta \projy^{\ge n} a \in \projy^{\ge 0}$ for any $a \in \projy^{\ge 0}$,
	an arbitrary $a \in \projy^{\ge 0}$ is admissible.
	On the other hand, the $\projy^{<n}$-sector of $\Delta^{-1}b$ is unconstrained
	and ranges over the intersection of~$y^n Q_z^-$ and~$\Delta^{-1} (Q_z^*)^+$.
	Therefore,
	\begin{align}
		S^\perp &=
		\left\{
			\begin{pmatrix}
			a \\
			a + (\tfrac 1 z - 1) \projy^{\ge n} a + \Delta^{-1} b'
			\end{pmatrix}
			~\middle|~
			a \in Q_z^+, ~~b' \in \bd^{(y)}_n(\Delta Q_z, Q_z^*)
		\right\}. \label{eq:orthocalc}
	\end{align}
	Here, the term of~$b'$ is already $f(0,b')$.
	Using $Q_z^+ = \bd_n^{(y)}(\Delta^{-1} Q_z^*, Q_z) \oplus y^n \Delta^{-1} (Q_z^*)^+$,
	we write 
	\begin{align}
		a = a'' + z\Delta^{-1}c
	\end{align}
	uniquely where $c \in \projy^{\ge n}$.
	We have a decomposition
	\begin{align}
		\begin{pmatrix}
		a \\
		a + (\tfrac 1 z - 1) \projy^{\ge n} a
		\end{pmatrix}
		=
		f\begin{pmatrix} a''\\ (\tfrac 1 z - 1)(\Delta \projy^{\ge n} - \nu \Delta)a''\end{pmatrix} 
		+ (z-1)f\begin{pmatrix}0 \\ \Delta \projy^{<n} \Delta^{-1} c\end{pmatrix}
		+ \begin{pmatrix}	z \Delta^{-1} c \\ \Delta^{-1} c	\end{pmatrix},
	\end{align}
	where the arguments of~$f$ are in the correct domain because
	\begin{align}
		(Q_z^*)^+	
		\ni \quad
		&\underline{(\Delta \projy^{\ge n} - \nu \Delta) a'' }
		=
		(\Delta \projy^{\ge n} - \nu \Delta) (\projy^{\ge n} + \projy^{< n}) a''
		\nonumber\\
		&=
		- \nu \Delta \projy^{< n} a'' 
		=
		\Delta (\Delta^{-1}(1-\nu))\Delta \projy^{<n}a'' - \Delta \projy^{<n} a''
		\quad \in y^n \Delta Q_z^- ,\\
		(Q_z^*)^+
		\ni \quad
		&c - \Delta \projy^{\ge n} \Delta^{-1}c
		=
		\underline{\Delta \projy^{<n} \Delta^{-1} c}
		\quad \in
		y^n \Delta Q_z^- .
		\nonumber
	\end{align}
	This proves that $\bar f$ is surjective,
	and completes the proof that $\bar f$ is an $\Rz$-isomorphism.%
	\footnote{
		Ranicki~\cite[Lem.\,3.4]{Ranicki2} uses a similar map,
		which is defined by nominally the same matrix as~$f$
		but acting on a different domain~$P \oplus P^*$
		where $P = \bigoplus_{0 \le j < n} y^j Q_z$.
		It is however not explained in~\cite{Ranicki2} why the different domain results in 
		a Witt equivalent final answer~$[\xi]_0$.
		The difference between~\cite{Ranicki2} and ours 
		originates from the choice of~$S$.
		Ranicki takes $\{(zc',c')|c' \in \projy^{\ge n}\}$
		instead of our $\{(z\Delta^{-1}c,\Delta^{-1}c) | c \in \projy^{\ge n}\}$.
		Despite the difference, the formulas are nominally identical.
	}

	By~\ref{thm:separatorOnly},
	we obtain a representative~$\left[\xi\middle|_{\bd_n^{(y)}(\Delta^{-1} Q_z^*, Q_z) \oplus \bd^{(y)}_n(\Delta Q_z, Q_z^*)}\right]_{y^0}$ 
	of the form $\bass^\downarrow_{2\iota+1}(y) \bass^\uparrow_{2\iota}(z)[\phi]$ 
	where
	\begin{align}
		\xi = f^\dag \diag(\phi , -\phi) f.
	\end{align}
	Let us simplify the formula for~$\xi$.
	Recall~\eqref{eq:projzDagger}.
	\begin{align}
		&
		\begin{Bmatrix}	
		\bd_n^{(y)}(Q_z^*, \Delta Q_z)^* \\
		\bd^{(y)}_n(\Delta Q_z, Q_z^*)^*
		\end{Bmatrix}	
		\begin{pmatrix}
			\Delta^{-\dag} & 0 \\ 0 & I
		\end{pmatrix}
		\xi
		\begin{pmatrix}
			\Delta^{-1} & 0 \\ 0 & I
		\end{pmatrix}
		\begin{Bmatrix}	
		\bd_n^{(y)}(Q_z^*, \Delta Q_z) \\
		\bd^{(y)}_n(\Delta Q_z, Q_z^*)
		\end{Bmatrix}
		\nonumber\\
		&\sim
		\mp
		\begin{pmatrix}
			I &&& I - \nu + z \nu \\
			0 &&& I
		\end{pmatrix}
		\begin{pmatrix}
			\psi & 0 \\ 0 & -\psi
		\end{pmatrix}
		\begin{pmatrix}
			I &&& 0 \\
			I-\nu + \tfrac 1 z \nu &&& I
		\end{pmatrix}\qquad(\Delta^{-\dag} \phi \Delta^{-1} \sim \mp \psi) \nonumber\\
		&\sim
		\mp
		\left(
		\begin{pmatrix}
			\psi - (1-\nu + z\nu)\psi(1-\nu+\tfrac 1 z \nu) &&& -(1-\nu+ z\nu)\psi \\
			-\psi(1-\nu+\tfrac 1 z \nu) &&& -\psi
		\end{pmatrix}\right.\label{}\\
		&\qquad\qquad
		\left. +
		\begin{pmatrix}
		-(1-\tfrac 1 z)(1 - \nu )\psi\nu  \mp (1-z)\nu\psi^\dag(1-\nu)	&&& (1-\nu)\psi \pm z \nu \psi^\dag \\
		\pm \psi^\dag(1-\nu) + \tfrac 1 z \psi \nu	&&& 0
		\end{pmatrix}\right)
		\nonumber\\
		&=
		\mp
		\begin{pmatrix}
			(1-z)(\nu \Delta^{-1} - \nu \Delta^{-1} \nu) &&& -z \nu \Delta^{-1}\\
			-\Delta^{-1} (1-\nu) &&& -\psi
		\end{pmatrix} = \xi' \nonumber
	\end{align}
	where 
	the second $\sim$ is because the indented matrix on the fourth line is an even hermitian $\pm$-form.
	Therefore, $\bass^\downarrow_{2\iota+1}(y)\circ \bass^\uparrow_{2\iota}(z)([\phi])$
	is represented by $\left[\xi'\middle|_{\bd_n^{(y)}(Q_z^*, \Delta Q_z) \oplus \bd^{(y)}_n(\Delta Q_z, Q_z^*)}\right]_{y^0}$.

	Since $\psi \sim \mp \Delta^{-\dagger} \phi \Delta^{-1}$,
	we must have $\bass^\downarrow_{2\iota+1}(y)\circ \bass^\uparrow_{2\iota}(z)([\psi]) 
	= \bass^\downarrow_{2\iota+1}(y)\circ \bass^\uparrow_{2\iota}(z)([\mp \phi])$.
	This means that we may replace every~$\psi$ by~$\mp \phi$,
	{\em i.e.},
	\begin{align}
		&\bass^\downarrow_{2\iota+1}(y)\circ \bass^\uparrow_{2\iota}(z)([\phi])
		=
		[
		[\xi''|_W]_{y^0}
		],\nonumber\\
		&W = {\bd_n^{(y)}(Q_z, \Delta^{-1} Q_z^*) \oplus \bd^{(y)}_n(\Delta^{-1} Q_z^*, Q_z)},\\
		&\xi'' = 
			(-1)
			\begin{pmatrix}
				(1-z)(\nu \Delta \nu - \nu \Delta) &&& z \nu \Delta\\
				\Delta (1-\nu) &&& \phi
			\end{pmatrix}.\nonumber
	\end{align}

	On the other hand, we have computed $V^\dag \eta V$ 
	in~\eqref{eq:LEtaL},\eqref{eq:LstarEtaLstar},\eqref{eq:LstarEtaL},
	where $[V] = \bass^\downarrow_{2\iota}(y)([\Delta])$.
	Arranging them with respect to~$W$,
	we have
	\begin{align}
		V^\dag \eta V =	
		\begin{Bmatrix}
		\bd_n^{(y)}(Q_z, \Delta^{-1} Q_z^*)^* \\
		\bd^{(y)}_n(\Delta^{-1} Q_z^*, Q_z)^*
		\end{Bmatrix}
		\begin{pmatrix}
			\nu \Delta \nu - \Delta &&& - \nu \Delta\\
			\Delta (1-\nu) &&& \Delta
		\end{pmatrix}
		\begin{Bmatrix}
		\bd_n^{(y)}(Q_z, \Delta^{-1} Q_z^*) \\
		\bd^{(y)}_n(\Delta^{-1} Q_z^*, Q_z)
		\end{Bmatrix}.
	\end{align}
	By~\ref{rem:alt-bass-up-from-unitary},
	it is evident that $\bass^\uparrow_{2\iota-1}(z) \circ \bass^\downarrow_{2\iota}(y)([\Delta]) 
	= [-[\xi'' \mp (\xi'')^\dag |_W]_{y^0}]$.
	By~\ref{thm:WittGroupOfForms}, the minus sign means the skew-commutativity.
\end{proof}

\begin{proposition}\label{thm:evenBassExact}
	The following short sequence is split exact.
	\begin{align}
		0 \xrightarrow{\quad} \qwitt^{\mp}(R) \xrightarrow{\quad \bar\epsilon \quad} \qwitt^{\mp}(\Rz)
		\xrightarrow{\quad \bass^\downarrow_{2\iota} \circ \sym \quad} \qumod^{\pm}(R) \xrightarrow{\quad} 0
	\end{align}
\end{proposition}

\begin{proof}
	The composition of the two maps is zero,
	since, if $\Delta \in \im \sym \circ \bar\epsilon(z)$, 
	then $L(\Delta,B,n) = \cz(\Delta)\bigoplus_{j=0}^{n-1} z^j (B \oplus 0)$ is 
	the image of an elementary unitary~$\cz(\Delta)$ of a standard lagrangian.
	The map $\bar\epsilon(z)$ is an embedding because if an $\Rz$-isomorphism~$E$ shows an equivalence
	of $\bar\epsilon(z)(\Delta)$ to a trivial form for $[\Delta] \in \qwitt^\mp(R)$,
	then $\epsilon(E)$ shows an $R$-equivalence of $\Delta$ to a trivial form where $\epsilon : \Rz \ni z \mapsto 1 \in R$.
	In addition, $\bar\epsilon(z)$ has an left inverse induced by~$\epsilon$.
	So, if the sequence in the claim is exact, then it is split exact.
	By~\ref{thm:up-and-down-from-unitary-is-identity},
	the map~$\bass^\downarrow_{2\iota}(z)$ is surjective,
	and clearly, $\sym$ is surjective.
	Hence, $\bass^\downarrow_{2\iota}(z) \circ \sym$ is surjective.

	It remains to prove the exactness at the middle.
	We consider a right inverse~$\hat\sym$ of $\sym$,
	where $\hat \sym$ is a map from the set of all nonsingular even hermitian $\mp$-forms 
	to the set of all nonsingular quadratic $\mp$-forms;
	the right inverse $\hat\sym$ is otherwise arbitrary,
	and is \emph{not} meant to induce a group homomorphism between Witt groups.
	Then, $\sym([\xi] - [\hat\sym \sym \xi]) = [\sym \xi] - [\sym \xi] = 0 \in \switt^\mp$,
	and hence by~\ref{thm:hermitianMatrices} we see
	\begin{align}
	[\xi] - [\hat\sym \sym \xi] \in \qwitt^\mp(\FF) \subseteq \bar\epsilon(z) \qwitt^\mp(R).
	\label{eq:xihatsym}
	\end{align}
	This means that $\hat\sym$ induces a group homomorphism $\switt^\mp(R) \to \qwitt^\mp(R) / \qwitt^\mp(\FF)$.

	We claim that
	\begin{align}
		\qwitt^\mp(\Rz) &= \left\{
		[\bar\epsilon(z) \xi] + [ \hat\sym \circ \bass^\uparrow(z) U ] ~\middle|~
		[\xi] \in \qwitt^\mp(R), U \in \qhp^\pm(R) 
		\right\}. \label{eq:qwittExpress}
	\end{align}
	Then, the desired exactness will follow 
	because 
	\begin{align}
	\bass^\downarrow_{2\iota}(z)\sym( [\bar\epsilon(z)\xi] + [\hat\sym\bass^\uparrow(z)U])
	= \bass^\downarrow_{2\iota}(z) \bass^\uparrow_{2\iota-1}(z)[U] = [U]
	\end{align}
	where the last equality is by~\ref{thm:up-and-down-from-unitary-is-identity}.
	To show~\eqref{eq:qwittExpress},
	we look at the direct sum decomposition supplied by~\ref{thm:oddBassExact}:
	\begin{align}
		\qumod^\mp(\Ryz) = \bar\epsilon(z) \qumod^\mp(\Ry) \oplus \bass^\uparrow_{2\iota}(z) \qwitt^\mp(\Ry).
	\end{align}
	Applying $\bass^\downarrow_{2\iota+1}(y)$, we have
	\begin{align}
		&\bass^\downarrow_{2\iota+1}(y)\qumod^\mp(\Ryz) \nonumber \\
		&= \bass^\downarrow_{2\iota+1}(y)\circ \bar\epsilon(z) \qumod^\mp(\Ry) + \bass^\downarrow_{2\iota+1}(y)\circ \bass^\uparrow_{2\iota}(z) \qwitt^\mp(\Ry) \\
		&=
		\bar\epsilon(z) \circ \bass^\downarrow_{2\iota+1}(y) \qumod^\mp(\Ry) + \bass^\downarrow_{2\iota+1}(y)\circ \bass^\uparrow_{2\iota}(z) \qwitt^\mp(\Ry)\nonumber
	\end{align}
	because $z$ does not interfere with $\bass^\downarrow_{2\iota+1}(y)$.
	Since $\bass^\downarrow_{2\iota+1}$ is surjective by~\ref{thm:fromForm-up-and-down-is-identity}, we have
	\begin{align}
		\qwitt^\mp(\Rz) &=
		\bar\epsilon(z) \qwitt^\mp(R) + \bass^\downarrow_{2\iota+1}(y)\circ \bass^\uparrow_{2\iota}(z) \qwitt^\mp(\Ry)\nonumber\\
		&=
		\bar\epsilon(z) \qwitt^\mp(R) + \hat\sym \circ \sym \circ \bass^\downarrow_{2\iota+1}(y)\circ \bass^\uparrow_{2\iota}(z) \qwitt^\mp(\Ry) & \text{\eqref{eq:xihatsym}}\\
		&=
		\bar\epsilon(z) \qwitt^\mp(R) + \hat\sym \bass^\uparrow_{2\iota -1}(z) \circ \bass^\downarrow_{2\iota}(y) \circ \sym \qwitt^\mp(\Ry)& \text{(\ref{thm:up-down-and-down-up})}\nonumber\\
		&=
		\bar\epsilon(z) \qwitt^\mp(R) + \hat\sym \bass^\uparrow_{2\iota -1}(z) \qumod^\pm(R) &\text{(\ref{thm:up-and-down-from-unitary-is-identity})} \nonumber
	\end{align}
	where $\hat \sym$ should be understood as being applied to a representative of a Witt class,
	which is alright since the ambiguity of $\hat \sym$ is absorbed into $\bar\epsilon(z) \qwitt^\mp(R)$.
	The last expression is equivalent to~\eqref{eq:qwittExpress}.
	This completes the proof.
\end{proof}

\begin{proposition}\label{thm:tildeBass}
	Recall~\eqref{eq:tildeBass}.
	The map~$\tilde\bass : \qumod^\mp(R) \ni [U] \mapsto [\tilde\bass^\uparrow(U)] \in \qwitt^\pm(\Rz)$
	is well defined and is the right inverse of $\bass^\downarrow_{2\iota} \circ \sym$.
\end{proposition}

\begin{proof}
	Observe that a quadratic $\pm$-form
	$\epsilon(z) \circ \tilde \bass^\uparrow(U) = \begin{pmatrix} \xi & -\gamma \\ \delta & 0 \end{pmatrix} 
	\sim 
	\begin{pmatrix} \xi & -\gamma \pm \delta^\dag \\ 0 & 0 \end{pmatrix}
	=
	\begin{pmatrix} \xi & -I \\ 0 & 0 \end{pmatrix}$
	is a trivial quadratic $\pm$-form over~$R$ for an arbitrary choice of $\xi$ given $\xi \pm \xi^\dag$.
	So, if $\psi = \begin{pmatrix} \xi' & -z \gamma \\ \delta & (1-z)(\rho \pm \rho^\dag) \end{pmatrix}$ 
	differs from $\tilde\bass^\uparrow(U)$
	only in $\xi'$ where $\xi \pm \xi^\dag = \xi' \pm (\xi')^\dag$,
	then $\epsilon[\psi] = \epsilon[\tilde\bass^\uparrow(U)] =0$ 
	and $\bass^\downarrow_{2\iota}\circ \sym [\psi] = \bass^\downarrow_{2\iota}\circ \sym[\tilde\bass^\uparrow(U)] = [U]$.
	By~\ref{thm:evenBassExact}, we must have $[\psi] = [\tilde\bass^\uparrow(U)] \in \qwitt^\pm(\Rz)$.
\end{proof}

\section{Coarse graining and blending equivalence}\label{sec:cg}

There are split exact sequences for $\lambda$-unitary groups and hermitian Witt groups.

\begin{proposition}\label{thm:oddBassExactNonquadratic}
	Over $R = \FF_2[x_1,x_1^{-1},\ldots,x_\dd,x_\dd^{-1}]$ of characteristic~$2$,
	the following short sequence is split exact.
	\begin{align}
		0 \xrightarrow{\quad} \umod(R) \xrightarrow{\quad \bar\epsilon \quad} \umod(\Rz) 
		\xrightarrow{\quad \bass^\downarrow_{2\iota +1} \quad} \switt(R) \xrightarrow{\quad} 0
	\end{align}
\end{proposition}
\begin{proof}
	The injectivity of $\bar\epsilon$ is proved by the same argument as in~\ref{thm:oddBassExact}:
	if $\bar\epsilon(T) = S$ for some $S \in \ehp(\Rz)$,
	then $T = \epsilon \bar\epsilon T = \epsilon S \in \ehp(R)$.
	The surjectivity of $\bass^\downarrow_{2\iota+1}$ follows from~\ref{thm:fromForm-up-and-down-is-identity}.
	It is obvious that the composition of the two maps is zero, as before.
	The maps $\epsilon$ and $\bass^\uparrow_{2\iota}$ will give the splitting.

	To prove the exactness at the middle,
	let $[U] \in \umod(\Rz)$.
	By~\ref{thm:timereversal}, we may assume that $U \in \qhp(\Rz)$.
	By~\ref{thm:oddBassExact}, $\qumod(\Rz) \ni [U] = \bar\epsilon [U_0] + \bass^\uparrow([\xi])$
	for some $U_0 \in \qhp(R)$ and $\xi \in \qwitt(R)$.
	Suppose $\sym\circ \bass^\downarrow([U]) = 0 \in \switt(R)$.
	Then $\sym([\xi]) = 0$, and by~\ref{thm:hermitianMatrices} 
	we know $[\xi] = [\xi_0] \in \qwitt(R)$ for some nonsingular quadratic form $\xi_0$ over~$\FF_2$.
	Hence, $\bass^\uparrow([\xi]) = \bass^\uparrow([\xi_0]) \in \qumod(\FF_2[z,\tfrac 1 z])$.
	But, \ref{thm:MinusUnitaryIn1D} implies that $\qhp(\FF_2[z,\tfrac 1 z])$ becomes zero in $\umod(\FF_2[z,\tfrac 1 z])$.
	Hence, $[U] = \bar\epsilon [U_0]$ in the $\lambda$-unitary group.	
\end{proof}

\begin{proposition}\label{thm:evenBassExactNonquadratic}
	The following short sequence is split exact.
	\begin{align}
		0 \xrightarrow{\quad} \switt^\mp(R) \xrightarrow{\quad \bar\epsilon \quad} \switt^\mp(\Rz)
		\xrightarrow{\quad \bass^\downarrow_{2\iota} \quad} \qumod^\pm(R) \xrightarrow{\quad} 0
	\end{align}
\end{proposition}

In contrast to~\ref{thm:evenBassExact}, 
the sequence here has even hermitian Witt groups, rather than quadratic Witt groups.

\begin{proof}
	The surjectivity of $\bass^\downarrow_{2\iota}$ is proved in~\ref{thm:up-and-down-from-unitary-is-identity}.
	The injectivity of $\bar\epsilon$ is proved as in~\ref{thm:evenBassExact}:
	if an $\Rz$-isomorphism~$E$ shows an equivalence of~$\bar \epsilon(z)(\Delta)$
	to a trivial form for~$[\Delta] \in \switt^\mp(R)$,
	then $\epsilon(E)$ (replacing~$z$ by~$1$) is an $R$-isomorphism showing that~$\Delta$ is trivial.
	The composition being zero follows from that of~\ref{thm:evenBassExact}
	since $\sym : \qwitt^\mp \to \switt^\mp$ is surjective.
	To prove the exactness at the middle,
	let $[\Delta] \in \switt^\mp(\Rz)$ and suppose $\bass^\downarrow_{2\iota}[\Delta] = 0$.
	Borrowing the map~$\hat\sym$ in the proof of~\ref{thm:evenBassExact},
	we have $(\bass^\downarrow_{2\iota} \circ \sym)[\hat \sym \Delta] = 0$,
	implying that $[\hat\sym \Delta] \in \bar\epsilon\, \qwitt^\mp(R)$.
	Therefore, $[\Delta] = \sym [\hat\sym \Delta] \in \sym \bar\epsilon\,  \qwitt^\mp(R) 
	= \bar\epsilon\,  \sym \qwitt^\mp(R) = \bar\epsilon\, \switt^\mp(R)$.
	The sequence is split by $\epsilon$ and $\bass^\uparrow_{2\iota -1}$.
\end{proof}

Let $R(\dd,p) = \FF_p[x_1,x_1^{-1},\ldots,x_\dd,x_\dd^{-1}]$. 
Recall the notation in~\ref{defn:vtheory}, and introduce abbreviations
\begin{align}
	\vbass^\downarrow_{2\iota+1} &= \bass^\downarrow_{2\iota+1} : 
	\vtheory_{2\iota+1}(\dd+1,p) = \qumod^{(-1)^\iota}(R(\dd+1,p)) 
	\xrightarrow{\text{\ref{thm:boundaryXi-from-QCA}}} 
	\qwitt^{(-1)^\iota}(R(\dd,p)) = \vtheory_{2\iota}(\dd,p),
	\nonumber\\
	\vbass^\uparrow_{2\iota} &= \bass^\uparrow_{2\iota} :
	\vtheory_{2\iota}(\dd,p) = \qwitt^{(-1)^\iota}(R(\dd,p))
	\xrightarrow{\text{\ref{thm:form-to-up-unitary}}}
	\qumod^{(-1)^\iota}(R(\dd+1,p)) = \vtheory_{2\iota+1}(\dd+1,p),
	\\
	\vbass^\downarrow_{2\iota} &= \bass^\downarrow_{2\iota} \circ \sym :
	\vtheory_{2\iota}(\dd,p) = \qwitt^{(-1)^\iota}(R(\dd,p))
	\xrightarrow{\text{\ref{thm:FromFormsDownToUnitary}}}
	\qumod^{-(-1)^\iota}(R(\dd-1,p)) = \vtheory_{2\iota-1}(\dd-1,p),
	\nonumber\\
	\vbass^\uparrow_{2\iota-1} &= \tilde\bass^\uparrow :
	\vtheory_{2\iota-1}(\dd-1,p) = \qumod^{-(-1)^\iota}(R(\dd-1,p))
	\xrightarrow{\text{\ref{thm:BassUpFromUnitary}, \ref{thm:tildeBass}}}
	\qwitt^{(-1)^\iota}(R(\dd,p)) = \vtheory_{2\iota}(\dd,p),
	\nonumber
\end{align}
where the subscripts are defined mod~4.
The results of~\ref{thm:oddBassExact} and \ref{thm:evenBassExact} imply that, 
for any $\dd \ge 0$,
\begin{align}
	\vtheory_{n}(\dd,p) 
	&= 
	\bar\epsilon(x_\dd) \vtheory_n(\dd-1,p) \oplus \vbass^\uparrow_{n-1}\vtheory_{n-1}(\dd-1,p)\nonumber\\
	&=
	\left[\bigoplus_{k=0}^{\dd-1} \bar\epsilon(x_{\dd-k}) \left( \prod_{j=1}^{k} \vbass^\uparrow_{n-j} \right)\vtheory_{n-k}(\dd-1-k,p)\right]\label{eq:vtheoryRecursion}\\
	&\qquad \oplus
	\left( \prod_{j=1}^{\dd} \vbass^\uparrow_{n-j} \right) \vtheory_{n-\dd}(0,p), \nonumber
\end{align}
while \ref{thm:oddBassExactNonquadratic} and \ref{thm:evenBassExactNonquadratic} imply that, for $\dd \ge 2$,
\begin{align}
	\umod^{(-1)^\iota}(R(\dd,p)) 
	&= 
	\bar\epsilon(x_\dd) \umod^{(-1)^\iota}(R(\dd-1,p))
	\oplus
	\bass^\uparrow_{2\iota}\switt^{(-1)^\iota}(R(\dd-1,p))\label{eq:lambdaUnitaryToEtaUnitary}\\
	&=
	\bar\epsilon(x_\dd) \umod^{(-1)^\iota}(R(\dd-1,p))
	\oplus
	\bar\epsilon(x_{\dd-1})\bass^\uparrow_{2\iota}\switt^{(-1)^\iota}(R(\dd-2,p))\nonumber\\
	&\qquad \oplus
	\bass^\uparrow_{2\iota}\vbass^\uparrow_{2\iota-1}\vtheory_{2\iota-1}(\dd-2,p). \nonumber
\end{align}
Note here that we have given explicit embeddings of lower dimensional groups into the top dimensional group.
We liberally used the fact that $\bar\epsilon$ in a variable 
commutes with any going-up morphism in a different variable.

We are going to show that anything that is embedded into $\vtheory_n(\dd,p)$ by some~$\bar\epsilon$
can be ignored if we consider coarser translation invariance.

Recall that $\clifqca(\dd,p)$ is the group of all translation invariant Clifford QCA,
each of which defines a $\lambda^-$-unitary over $R = \FF_p[x_1,\tfrac 1 {x_1},\ldots,x_\dd,\tfrac 1 {x_\dd}]$,
modulo all Clifford circuits and shifts obeying coarser translation invariance.
The coarse graining is implemented~\cite{Haah2013}
by a functor~$\cg b$ from the category of $R$-modules to $S$-modules
induced by a base ring change~$\phi^{(b)} : S = \FF_p[x'_1,\tfrac 1 {x'_1},\ldots,x'_\dd,\tfrac 1 {x'_\dd}]
\ni x'_i \mapsto x_i^b \in R$ that is $\FF_p$-linear.
As abstract rings, $R$ and $S$ are the same, 
but we think of $S$ as a translation group algebra
on $\ZZ^\dd$ with a smaller translation group.
Any quadratic or hermitian form over $R$ 
is an $R$-module morphism from a module to its dual module,
so a matrix representation of a form after a coarse-graining can be obtained 
by evaluating the form on basis elements of a $S$-module~\cite[IV.\,3]{nta3}.
A matrix representation of a unitary after a coarse-graining can be obtained
by writing the image of an $S$-basis with respect to the $S$-basis.
In particular, elementary unitary remains elementary
and any trivial form remains trivial under any coarse-graining.
Hence, $\cg b$ maps each of $\hp,\qhp,\switt,\qwitt$ into itself,
and we can speak of the direct limit~$\cg{\infty} = \lim_{b \to \infty} \cg{b}$ 
in the directed system of these groups under composed coarse-grainings.

Now, the base ring change $\phi^{(b)}$ can be thought of as a composition of $\dd$ ring homomorphisms~$\phi_i$
for $i = 1,2,\ldots,\dd$,
where $\phi_i : x'_j \mapsto x_j$ for all $j \neq i$ but $x'_i \mapsto x_i^b$.
For any object in the image of $\bar\epsilon(x_i):R(\dd,p)|_{x_i=1} \cong R(\dd-1,p) \to R(\dd,p)$,
applying $(\phi_i)_\#$ means, at the matrix level, that we make the $b$-fold direct sum of the object.
Since $\qwitt^+(\FF_p)$ has an exponent~$4$ for all prime~$p$,
we conclude by~\ref{thm:UnitaryGroupExponent} and \ref{thm:WittGroupExponent} that
\begin{align}
	\cg{4} \bar\epsilon = 0. \label{eq:cg4}
\end{align}

\begin{lemma}\label{thm:cg-down}
	Let $\phi: R[z',\tfrac 1 {z'}] \ni z' \mapsto z^b \in \Rz$ be an $R$-linear ring homomorphism.%
	\footnote{Unlike $\cg{b}$, we only coarse-grain along one direction of $\ZZ^{\dd+1}$.}
	Let $x$ denote any variable of~$R$, independent from~$z'$ and~$z$.
	Then, for any $n \in \ZZ/4\ZZ$, we have
	\begin{align}
		\bass^\downarrow_n(z') \phi_\# &= \bass^\downarrow_n(z),\\
		\bass^\downarrow_n(x) \phi_\# &= \phi_\# \bass^\downarrow_n(x).\nonumber
	\end{align}
\end{lemma}
\begin{proof}
	The first is because the going-down morphisms are constructed first by taking a boundary $R$-module 
	(\ref{thm:boundaryXi-from-QCA} and~\ref{thm:FromFormsDownToUnitary}),
	which forgets the translation structure along the $z$-direction.
	The second is obvious.
\end{proof}

We finally find a complete invariant 
of~$\clifqca(\dd,p) = \cg{\infty} \umod^-(\FF_p[x_1,x_1^{-1},\ldots,x_\dd,x_\dd^{-1}])$.
\begin{theorem}\label{thm:CompleteInvariantOfQCA}
	$\umod^-(\FF_p[x_1,x_1^{-1},\ldots,x_\dd,x_\dd^{-1}]) = 0$ for $\dd = 0,1,2$ and all prime~$p$.
	For $\dd \ge 3$, it holds that
	\begin{align}
		\cg{\infty} \umod^-(\FF_p[x_1,x_1^{-1},\ldots,x_\dd,x_\dd^{-1}]) 
		&\cong
		\vtheory_{3-\dd}(\FF_p).
	\end{align}
\end{theorem}

The result for $\dd = 0,1,2$ has been known~\cite{Haah2013,clifqca1}.

\begin{proof}
	For $\dd = 0$, the result is~\ref{thm:n0elem}.
	For $\dd = 1$, the result is~\ref{thm:MinusUnitaryIn1D},
	which can also be seen by the first line of~\eqref{eq:lambdaUnitaryToEtaUnitary},
	that shows that $\umod^-(R(1,p)) \cong \bar\epsilon(x_1)\umod^-(\FF_p) \oplus \switt^-(\FF_p)$
	which is zero by~\ref{thm:n0elem} and~\ref{thm:sesquilinearWittGroup}.
	For $\dd = 2$, \eqref{eq:lambdaUnitaryToEtaUnitary} says that
	$\umod^-(R(2,p)) \cong \umod^-(R(1,p)) \oplus \switt^-(\FF_p) \oplus \qumod^+(\FF_p)$
	which is zero due to~\ref{thm:n0elem} and~\ref{thm:sesquilinearWittGroup} again.

	For $\dd \ge 3$, we consider $\cg{b} = (\phi_1)_\# \cdots (\phi_\dd)_\#$ 
	where each $\phi_i$ coarse-grains along $i$-direction only as in~\ref{thm:cg-down}.
	Then, \ref{thm:cg-down} implies that
	\begin{align}
		&\left(\prod_{j=1}^{\dd-2}\vbass^\downarrow_{2-j}(x'_{\dd-1-j})\right) \bass^\downarrow_2(x'_{\dd-1}) \bass^\downarrow_3(x'_\dd)
		\cg{b} \umod^-(R(\dd,p)) \\
		&= 
		\left(\prod_{j=1}^{\dd-2} \vbass^\downarrow_{2-j}(x'_{\dd-1-j}) (\phi_{\dd-1-j})_\# \right) 
		\bass^\downarrow_2(x'_{\dd-1}) (\phi_{\dd-1})_\# \bass^\downarrow_3(x'_\dd) (\phi_\dd)_\# \umod^-(R(\dd,p))\nonumber\\
		&=
		\left(\prod_{j=1}^{\dd-2}\vbass^\downarrow_{2-j}(x_{\dd-1-j})\right) \bass^\downarrow_2(x_{\dd-1}) \bass^\downarrow_3(x_\dd)
		\umod^-(R(\dd,p)) \nonumber\\
		&\cong
		\vtheory_{3-\dd}(0,p).\nonumber
	\end{align}
	Since this is independent of~$b$, we have an invariant of 
	the direct limit by $\cg{\infty}$ of the $\lambda^-$-unitary group,
	valued in~$\vtheory_{3-\dd}(0,p)$.

	Suppose a $\lambda^-$-unitary~$U$ has the invariant zero.
	Then, from \eqref{eq:lambdaUnitaryToEtaUnitary} and \eqref{eq:vtheoryRecursion},
	we see that $[U]$ is an embedding of some $\bar\epsilon$.
	By~\eqref{eq:cg4}, we see~$\cg{\infty} [U] = 0$.
	This shows that the invariant is complete.
\end{proof}

\begin{proof}[Proof of~\ref{thm:main}]
	The group $\vtheory_{3-\dd}(0,p)$ is a unitary group if $3-\dd = 1 \bmod 2$,
	in which case it is zero by~\ref{thm:n0elem}.
	The group $\vtheory_{3-\dd}(0,p)$ is a quadratic Witt group if $3-\dd = 0 \bmod 2$,
	in which case we have the full answer in~\ref{thm:ZeroDcalculation}.
	Therefore, by~\ref{thm:CompleteInvariantOfQCA}
	for all~$\dd \ge 3$ and all prime~$p$,
	\begin{align}
		\clifqca(\dd,p) \cong
		\vtheory_{3-\dd}(0,p) \cong
		\begin{cases}
			\ZZ/2\ZZ & (p=2 \text{ and } \dd = 1 \bmod 2)\\
			\ZZ/4\ZZ & (p = 3 \bmod 4 \text{ and } \dd = 3 \bmod 4)\\
			\ZZ/2\ZZ \oplus \ZZ/2\ZZ & (p = 1 \bmod 4 \text{ and } \dd = 3 \bmod 4)\\
			0 & (\text{otherwise})
		\end{cases}.
	\end{align}
	The cases of~$\dd = 0,1,2$ are covered in~\ref{thm:CompleteInvariantOfQCA}.
\end{proof}

Let us interpret the results in terms of blending~\cite{FHH2019}.
Recall~\ref{thm:CompactDualIsDual} where we have identified a dual module~$M^*$ over~$\Rz$ 
with a module~$M^{*c}$ of $R$-linear functionals of compact support.

\begin{definition}
	We say that $U, V \in \hp^\mp(\bigoplus_j z^j Q_0; \Rz)$ (resp.~$U, V \in \qhp^\mp(\bigoplus_j z^j Q_0;\Rz)$) 
	{\bf blend} into each other
	if there exists a $R$-module isomorphism~$W$ on~$\bigoplus_j z^j (Q_0 \oplus Q_0^*)$ 
	and an integer~$n \ge 0$
	such that (i)~$\lambda^\mp(W u, W v) = \lambda^\mp(u, v)$ 
	(resp.~$\eta(W u, W v) \sim \eta(u, v)$) for all $u,v \in M$
	and (ii)~$W$~agrees with~$U$ on~$\projz^{\ge n}_{Q_0 \oplus Q_0^*}$
	and with~$V$ on $\projz^{<-n}_{Q_0 \oplus Q_0^*}$.
	Similarly, two nonsingular quadratic (resp. hermitian) $\mp$-forms~$\xi,\phi$
	on an $\Rz$-module~$Q = \bigoplus_j z^j Q_0$
	{\bf blend} into each other
	if there exist an integer $n\ge 0$ and an $R$-sesquilinear function $\mu : Q \times Q \to R$ 
	that defines an $R$-module isomorphism $Q \ni u \mapsto \left(v \mapsto \mu(u,v) \mp \overline{\mu(v,u)}\right) \in M^{*c}$ 
	(resp. $M \ni u \mapsto (v \mapsto \mu(u,v) = \mp \overline{\mu(v,u)} ) \in M^{*c}$)
	such that
	$\mu$ agrees with~$\xi$ on~$\projz^{\ge n}$ and with~$\phi$ on~$\projz^{<-n}$.
	We extend the definitions to say that two unitaries or two forms blend into each other
	if some stabilizations do.
\end{definition}

If we consider a unitary or a form over~$\Rz$ as an infinite matrix over~$R$ 
with respect to an ordered basis~$\{\ldots,z^{-2},z^{-1},1,z,z^2,\ldots\}$,
the matrix will be band-diagonal due to $\Rz$-linearity.
In case of $\lambda^-$-unitary, 
this property corresponds to the causality along $z$-direction of an associated Clifford QCA.

It is clear that two $\eta$- or $\lambda$-unitaries $U,V$ blends into each other
if $U^{-1}V$ blends into the identity.
An obvious $\eta$- or $\lambda$-unitary that blends into the identity
is anything in the image of the map~$\bar\epsilon(z)$;
this is precisely one that restricts to an $R$-automorphism~$W_0$ on each~$z^j R^{2q}$,
and an interpolating automorphism~$W$ is defined
by~$W u = W_0 u$ if~$u \in \projz^{\ge 0}$ and~$W u = u$ if~$u \in \projz^{< 0}$.
Of course, any Clifford circuit deformation of an identity-blending unitary
blends into the identity, too, because gates in the circuit can be dropped~\cite{FHH2019}.
Hence, the blending equivalence is a property of the class of~$\qumod(\Rz)$ or~$\umod(\Rz)$.
An obvious form over~$\Rz$ that blends into the trivial form
is anything in the image of~$\bar\epsilon(z)$.

\begin{remark}
	Purely in terms of $R$-modules
	the identity-blending property of an elementary unitaries
	originates from the fact that an elementary matrix over~$\Rz$ ---
	an identity matrix plus at most one nonzero entry in the off-diagonal ---
	blends into the identity matrix.
	Just keep the action of the off-diagonal entry on $\projz^{\ge 0}$ and 
	eliminate the action on~$\projz^{< 0}$.
	The invertibility of this interpolating infinite matrix is not affected.
	By Suslin's stability theorem~\cite{Suslin1977Stability},
	it follows that any invertible matrix of determinant~$1$
	blends into the identity.
	The unitarity can pose some restriction (and that is the precisely the point of our main result~\ref{thm:main})
	but still elementary unitaries blend into identity:
	$\hada_\mp$ does not involve $z$ 
	and the argument for general elementary matrices applies for~$\cz(\theta)$ and for $\cx(\alpha)$.
\end{remark}

By this remark, we see that and two congruent forms on $\bigoplus_j z^j Q_0$ always blend into each other.

All the split short exact sequences above say that
if a unitary or a form vanishes under the descent map
then it comes from~$\bar\epsilon(z)$ so it blends into the trivial unitary or form.
Even if the image under the descent map is nonzero,
if it vanishes under some coarse graining along axes different from~$z$,
then then it is in the image of~$\bar\epsilon(z)$ and blends into the identity.
Therefore, our complete invariant in~\ref{thm:CompleteInvariantOfQCA} precisely captures 
whether a $\lambda^-$-unitary blends into the identity.
A similar remark also holds for quadratic and hermitian forms.

\begin{remark}\label{rem:skewcommutativity}
	In~\ref{thm:up-down-and-down-up} 
	we have proved the skewcommutativity of going-up and going-down morphisms
	by direct calculation.
	In view of the final result that those maps are virtually isomorphisms 
	(by ignoring images of $\bar \epsilon$, 
	which becomes precise by coarse graining),
	we instead consider a diagram where the groups are arranged the same way as~\eqref{eq:diagram}
	but with all the arrows coming down from the top $\eta$-unitary group,
	and develop some intuition about the skewcommutativity as follows.
	On~$yz$-plane of~$\ZZ^\dd$, an $\eta^\mp$-unitary in the ``bulk'' defines two hermitian forms,
	one~$\Delta_v$ on the right vertical boundary orthogonal to the $y$-axis
	and the other~$\Delta_h$ on the top horizontal boundary orthogonal to the $z$-axis.
	The two hermitian forms~$\Delta_v$ and $\Delta_h$ must blend into each other
	because they simultaneously appear at the kinked boundary of the bulk
	restricted on the third quadrant~$\{(y,z) \in \ZZ^2 ~|~ y,z < 0\}$ in the bottom left.
	So, the obstruction at the right end of~$\Delta_h$ must cancel
	that at the top end of~$\Delta_v$,
	giving the skewcommutativity.
\end{remark}

\section*{Declarations}

{\bf Data availability.}
Data availability does not apply to this research since no data was used or generated.

{\bf Competing Interests.}
The author declares that there is no conflict of interest.

\nocite{apsrev42Control}
\bibliographystyle{apsrev4-2}
\bibliography{../refs}

\end{document}